\documentclass[12pt]{article}
\pdfoutput=1 

\usepackage[utf8]{inputenc}
\usepackage{jheppub}
\usepackage{epsfig}
\usepackage{bbm}
\usepackage{bbding}
\usepackage{abstract}
\usepackage{amsmath}
\usepackage{amsfonts}
\usepackage{amssymb}
\usepackage{mathtools}
\usepackage{amsthm}
\usepackage{dsfont}
\usepackage{bm}
\usepackage{blkarray}
\usepackage{graphicx}
\usepackage{longtable}
\usepackage{pdflscape}
\usepackage{subcaption}
\usepackage{color}
\usepackage{psfrag}
\usepackage[capitalise,sort]{cleveref}
\usepackage{hyperref}
\usepackage{pgfplots}
\usepackage{pgfplotstable}
\usepackage{subfiles}
\usepackage{longtable}
\usepackage{tabularx}
\usepackage{ltablex}
\usepackage{booktabs}
\usepackage[export]{adjustbox}
\usepackage{listings}
\usepackage{multirow} 
\usepackage{cancel,slashed}
\usepackage{bbm}
\usepackage{graphicx}
\usepackage{tikz,textcomp}
\usepackage{tikz}
\usetikzlibrary{decorations.pathmorphing}
\usetikzlibrary{calc,fit,matrix}
\usetikzlibrary{cd}
\usetikzlibrary{shapes,arrows,positioning,chains,matrix}
\usetikzlibrary{positioning,arrows.meta}
\usepackage{latexsym}
\usepackage{color}
\usepackage{transparent}
\usepackage{mathtools}
\usepackage{array}
\usepackage{makecell}
\usepackage{graphics,psfrag}
\usepackage{placeins}
\usepackage{nowidow}
\usepackage{listings}
\usepackage[normalem]{ulem}
\usepackage{environ}
\usepackage{arydshln}
\usepackage{enumitem}
\usepackage{comment}
\usepackage{hhline}

\pgfplotsset{
        compat=1.9,
        compat/bar nodes=1.8,
    }

\makeatletter
\def\@xfootnote[#1]{%
	\protected@xdef\@thefnmark{#1}%
	\@footnotemark\@footnotetext}
\makeatother

\definecolor{prhigh}{HTML}{ff0000}
\definecolor{sechigh}{HTML}{e0fbfc}
\definecolor{prcolor}{HTML}{1d3557}
\definecolor{seccolor}{HTML}{457b9d}
\definecolor{tercolor}{HTML}{98c1d9}

\theoremstyle{plain}
\newtheorem{theorem}{Theorem}

\newtheorem{proposition}[theorem]{Proposition}

\theoremstyle{definition}
\newtheorem{definition}{Definition}

\theoremstyle{remark}

\makeatletter
\renewenvironment{proof}[1][\proofname]{%
  \par\pushQED{\qed}\normalfont%
  \topsep6\p@\@plus6\p@\relax
  \trivlist\item[\hskip\labelsep\bfseries#1\@addpunct{.}]%
  \ignorespaces
}{
  \popQED\endtrivlist\@endpefalse
}
\makeatother

\DeclareMathOperator{\U}{U}
\DeclareMathOperator{\SU}{SU}

\DeclareMathOperator{\re}{Re}
\DeclareMathOperator{\im}{Im}

\newcommand{\PP}{\mathbb{P}}

\newcommand{\FF}{\mathbb{F}}
\newcommand{\ID}{\mathds{1}}

\newcommand{\coma}{\, , \quad}
\newcommand{\fstop}{\, .}

\newcommand{\YM}{\text{\tiny YM}}

\newcommand{\ii}{{\rm i}}

\renewcommand{\epsilon}{\varepsilon}

\makeatletter
\newsavebox{\measure@tikzpicture}
\NewEnviron{scaletikzpicturetowidth}[1]{%
  \def\tikz@width{#1}%
  \begin{lrbox}{\measure@tikzpicture}%
  \BODY
  \end{lrbox}%
  \pgfmathparse{#1/\wd\measure@tikzpicture}%
  \BODY
}
\makeatother

\makeatletter
\newcommand{\supsetcong}{\mathrel{\mathpalette\supset@cong\relax}}
\newcommand{\supsetsim}{\mathrel{\mathpalette\supset@sim\relax}}

\newcommand{\supset@cong}[2]{%
  \vbox{\offinterlineskip\m@th
    \ialign{\hfil##\cr
      \scalebox{0.9}{$#1\sim$}\cr
      \noalign{%
        \ifx#1\displaystyle\kern-0.5pt\else
        \ifx#1\textstyle\kern-0.5pt\fi\fi
      }%
      $#1\supset$\cr
    }%
  }%
}
\newcommand{\supset@sim}[2]{%
  \vtop{\offinterlineskip\m@th
    \ialign{\hfil##\cr
      $#1\supset$\cr\noalign{\kern0.5pt}\scalebox{0.9}{$#1\sim$}\cr
    }%
  }%
}
\makeatother

\let\mathcaldefault\mathcal
\def\fnote#1#2{\begingroup\def\thefootnote{#1}\footnote{#2}
     \addtocounter{footnote}{-1}\endgroup}
\DeclareMathAlphabet{\mathdutchcal}{U}{dutchcal}{m}{n}

\catcode`\@=11   
\newdimen\@rotdimen
\newbox\@rotbox  

\def\@vspec#1{\special{ps:#1}}
\def\@rotstart#1{\@vspec{gsave currentpoint currentpoint translate
		#1 neg exch neg exch translate}}
\def\@rotfinish{\@vspec{currentpoint grestore moveto}}
\def\@rotr#1{\@rotdimen=\ht#1\advance\@rotdimen by\dp#1%
	\hbox to\@rotdimen{\hskip\ht#1\vbox to\wd#1{\@rotstart{90 rotate}%
			\box#1\vss}\hss}\@rotfinish}
\def\@rotl#1{\@rotdimen=\ht#1\advance\@rotdimen by\dp#1%
	\hbox to\@rotdimen{\vbox to\wd#1{\vskip\wd#1\@rotstart{270 rotate}%
			\box#1\vss}\hss}\@rotfinish}%
\def\@rotu#1{\@rotdimen=\ht#1\advance\@rotdimen by\dp#1%
	\hbox to\wd#1{\hskip\wd#1\vbox to\@rotdimen{\vskip\@rotdimen
			\@rotstart{-1 dup scale}\box#1\vss}\hss}\@rotfinish}%
\def\@rotf#1{\hbox to\wd#1{\hskip\wd#1\@rotstart{-1 1 scale}%
		\box#1\hss}\@rotfinish}%
\def\rotate{\@ifnextchar[{\@rotate}{\@rotate[l]}}
\def\@rotate[#1]#2{\setbox\@rotbox=\hbox{#2}\@nameuse{@rot#1}\@rotbox}

\catcode`\@=12

\tikzset{
    partial ellipse/.style args={#1:#2:#3}{
        insert path={+ (#1:#3) arc (#1:#2:#3)}
    }
}

\hypersetup{
	pdftitle={The Asymptotic Weak Gravity Conjecture for Open Strings},    
	pdfauthor={\textcopyright\ Cesar Fierro Cota, Alessandro Mininno, Timo Weigand, Max Wiesner},     
	pdfsubject={HEP},   
	pdfcreator={pdfLaTex},   
	pdfproducer={LaTex}, 
	pdfkeywords={},
}

\allowdisplaybreaks[1] 

\crefname{figure}{Figure}{Figures}
\crefname{table}{Table}{Tables}
\crefname{definition}{Definition}{Definitions}
\crefname{proposition}{Proposition}{Propositions}

\begin{document}
	\pagestyle{plain}

	\makeatletter
	\@addtoreset{equation}{section}
	\makeatother
	\renewcommand{\theequation}{\thesection.\arabic{equation}}
	\pagestyle{empty}
 \rightline{ZMP-HH/22-15}
\vspace{0.8cm}

\begin{center}
{\large \bf
The Asymptotic Weak Gravity Conjecture for Open Strings  
} 

\vskip 9 mm

Cesar Fierro Cota,${}^1$ Alessandro Mininno,${}^1$ Timo Weigand,${}^{1,2}$ Max Wiesner${}^{3}$

\vskip 9 mm

\small ${}^{1}$ 
{\it II. Institut f\"ur Theoretische Physik, Universit\"at Hamburg, Luruper Chaussee 149,\\ 22607 Hamburg, Germany} 

\vspace{2mm}

\small ${}^{2}${\it Zentrum f\"ur Mathematische Physik, Universit\"at Hamburg, Bundesstrasse 55, \\ 20146 Hamburg, Germany  }   \\[3 mm]

\small ${}^{3}$ 
{\it Center of Mathematical Sciences and Applications, Harvard University, 20 Garden Street,\\ Cambridge, MA 02138, USA} \\[3 mm]

\fnote{}{\hspace{-0.75cm} cesar.fierro.cota at desy.de, \\ alessandro.mininno at desy.de, \\ timo.weigand at desy.de, \\ mwiesner at cmsa.fas.harvard.edu}

\end{center}

\vskip 7mm

\begin{abstract}

We investigate the asymptotic Tower Weak Gravity Conjecture in weak coupling limits of open string theories with minimal supersymmetry in four dimensions, focusing for definiteness on 
gauge theories realized on 7-branes in F-theory. 
Contrary to expectations, we find that not all weak coupling limits contain an obvious candidate for a tower of states marginally satisfying the super-extremality bound.
The weak coupling limits are classified geometrically in the framework of EFT string limits and their generalizations. We find three different classes of weak coupling limits, whose physics is characterized
by the ratio of the magnetic weak gravity scale and the species scale. The four-dimensional Tower Weak Gravity Conjecture is satisfied by the (non-BPS) excitations of the weakly coupled EFT 
string only in emergent string limits, where the EFT string can be identified with a critical (heterotic) string.
All other weak coupling limits lead to a decompactification either to an in general strongly coupled gauge theory coupled to gravity or to a defect gauge theory decoupling
from the gravitational bulk, in agreement with 
the absence of an obvious candidate for a marginally super-extremal tower of states.
\end{abstract}

	\newpage
	\setcounter{page}{1}
	\pagestyle{plain}
	\renewcommand{\thefootnote}{\arabic{footnote}}
	\setcounter{footnote}{0}
	
	\tableofcontents
	
	
\section{Introduction and Summary}

The \textit{Swampland program}
initiated in \cite{Vafa:2005ui} advocates the idea 
that a consistent coupling to gravity poses severe constraints
on an effective field theory.
The criteria which an effective field theory (EFT) must satisfy
in order to comply with these constraints are the subject of the so-called 
{\it swampland conjectures}.
In view of the growing web of such conjectures formulated so far \cite{Brennan:2017rbf,Palti:2019pca,vanBeest:2021lhn,Grana:2021zvf}, it is important to continue to test and sharpen our candidates for swampland constraints in order to advance our understanding of the fundamental characteristics of a theory of quantum gravity. 

Among the most acclaimed swampland conjectures is the Weak Gravity Conjecture (WGC) \cite{Arkani-Hamed:2006emk}. In its most modest formulation, it posits that for a given gauge theory coupled to gravity, there exists at least one particle whose gravitational interaction is weaker than the interaction associated with the gauge symmetry. The usual bottom-up motivation comes from the requirement that extremal black holes should be able to decay to avoid clashes with certain entropy bounds \cite{Arkani-Hamed:2006emk,Banks:2006mm}, at least in flat space and with four or more spacetime dimensions, which is the physical setting which we exclusively focus on in this paper.

Many versions of the WGC have been formulated in the last decade  
\cite{Palti:2020mwc,Harlow:2022gzl}, and it is important to investigate the application and the limitations of each formulation in well-controlled setups of quantum gravity. 
A stronger version of the WGC that was first motivated by consistency with dimensional reduction is the tower WGC (tWGC) \cite{Heidenreich:2015nta,Heidenreich:2016aqi,Montero:2016tif,Andriolo:2018lvp}, which states that in any EFT coupled to gravity there must exist at least one formally infinite tower of super-extremal states of charge $q_k$ and mass $M_k$ satisfying the relation
\begin{align}
\frac{q_k^2}{M_k^2}  \gtrsim \frac{1}{\Lambda^2_\text{\tiny WGC}} \,,
\end{align}
where
$\Lambda_\text{\tiny WGC} = g_\text{\tiny YM} M_\text{\tiny Pl}$ denotes the magnetic weak gravity scale and
the precise numerical factor on the right depends on the details of the theory.
The motivation behind the appearance of a tower is two-fold:
First, if the WGC is satisfied by a super-extremal tower, then after circle reduction
the convex hull condition \cite{Cheung:2014vva} is  automatically obeyed in the gauge sector defined by the original gauge theory and the Kaluza-Klein (KK) gauge field \cite{Heidenreich:2015nta}. 
Note, however, that at least in the presence of massless charged states in the original theory, a tower is only a sufficient, rather than a necessary condition for the convex hull condition to
be satisfied after compactification.

Second, a tower of super-extremal states
is well-motivated in the weak coupling limit of a gauge theory. 
We will refer to this version as the {\it asymptotic tower WGC}.
The limit $g_\text{\tiny YM} \to 0$ for a gauge theory lies at infinite distance in field space. In view of the Swampland Distance Conjecture \cite{Ooguri:2006in}, it should therefore be accompanied by a tower of infinitely many states which become  asymptotically light. 
It is natural to identify at least a sub-tower of the tower of states predicted by the Swampland Distance Conjecture 
with a tower of super-extremal states. 
In fact, not only the tower WGC, but even the WGC as such is best motivated, from a bottom-up perspective, in the limit where $g_\text{\tiny YM} \to 0$, because at least the entropy arguments invoked in \cite{Arkani-Hamed:2006emk} point to a potential inconsistency of black holes in the absence of a super-extremal state strictly speaking only in this limit.
Another argument that favors a tower of charged asymptotically light states in the weak coupling limit is provided by the Emergence Proposal \cite{Harlow:2015lma,Heidenreich:2018kpg,Grimm:2018ohb,Palti:2019pca}, which attributes the asymptotic vanishing of the gauge coupling to the appearance of a charged tower that must be integrated out.

Overwhelming evidence has been accumulated in the past years for the tower WGC in the literature.
This evidence appears either in the context of the {asymptotic tower WGC}, i.e., in the limit $g_\text{\tiny YM} \to 0$, or, away from this limit, in theories where the super-extremality condition is protected \cite{Ooguri:2016pdq} by the BPS property of states.\footnote{We stress again that in this paper, we focus exclusively on the WGC in at least four dimensions and in Minkowski space. Holographic checks of the WGC in (three- or higher-dimensional) AdS space  include \cite{Nakayama:2015hga,Harlow:2015lma,Montero:2016tif,Aharony:2021mpc,Antipin:2021rsh,Palti:2022unw}.}

Concerning the latter type of setups, the tower WGC has been analyzed in detail in \cite{Alim:2021vhs} in the context of five-dimensional M-theory compactifications on Calabi-Yau $3$-folds. The super-extremal states are furnished by M2-branes wrapping 2-cycles on the Calabi-Yau. Interestingly, the BPS condition and the super-extremality condition are both saturated for BPS states precisely in those situations where a tower of states can be identified among the BPS particles \cite{Alim:2021vhs}.

As for the asymptotic tWGC, two different types of situations have been investigated in detail: In the first class of limits, a tower of BPS states becomes light in four-dimensional compactifications of string theory with ${\cal N}=2$ supersymmetry near infinite distance boundaries of moduli space. This tower of states includes a WGC tower for the asymptotically weakly coupled abelian gauge theories associated with the closed string sector \cite{Grimm:2018ohb,Heidenreich:2018kpg,Gendler:2020dfp,Bastian:2020egp,Palti:2021ubp}.
In the second class of theories, the tower of states comes from the excitations of a weakly coupled heterotic string in various dimensions \cite{Arkani-Hamed:2006emk,Heidenreich:2016aqi}.
In the latter case, modular invariance is in fact sufficient to prove the appearance of a super-extremal tower \cite{Heidenreich:2016aqi,Montero:2016tif,Lee:2018urn,Lee:2018spm,Lee:2019tst,Lee:2020gvu,Lee:2020blx}.
This second type of constructions includes certain weak coupling limits in F-theory in which a solitonic heterotic string becomes light \cite{Lee:2018urn,Lee:2018spm,Lee:2019tst,Lee:2020gvu,Klaewer:2020lfg}.

It is interesting to note that the two types of super-extremal towers observed so far -- the BPS towers and the heterotic string excitation tower -- precisely correspond to the two possible types of towers of states which should generally become light at infinite distance in moduli space according to the Emergent String Conjecture \cite{Lee:2019tst}. 
The BPS towers are interpreted as KK towers in a dual frame, while the heterotic string excitations are the excitations of an emergent critical string, which is predicted to become tensionless in any infinite distance limit that is not a decompactification limit. 
In this sense, at least in its asymptotic, weak coupling form, the tower WGC can, so far, be viewed as a consequence of the Emergent String Conjecture. At a technical level, it rests either on
the BPS nature of states or on the modular properties of the heterotic string partition function. In both cases, the towers arise from the closed string sector. 

\subsubsection*{Summary of results}

In this work, we analyze the asymptotic tower WGC for more general weak coupling limits with ${\cal N}=1$ supersymmetry in four dimensions.
Our goal is to understand weak coupling limits
in which no critical heterotic string becomes light, especially for open string realizations of the gauge sector. 
According to the Emergent String Conjecture, such limits must be decompactification limits.
Neither of the two known mechanisms that have guaranteed the appearance of a super-extremal tower so far  -- BPS protection or the modular properties of the heterotic string -- are at work. 
A particularly fruitful approach to understand such weak coupling limits at infinite distance is the framework of {\it EFT} or {\it axionic string limits} introduced in
\cite{Lanza:2020qmt,Lanza:2021udy,Lanza:2022zyg} (and further studied in \cite{Marchesano:2022avb,Grimm:2022sbl}).
The starting point of this construction is the observation that 
an infinite distance limit in four dimensions is characterized by an emergent axionic shift symmetry. By dualizing the underlying axion to a 2-form, one finds that such limits are automatically accompanied by weakly coupled solitonic strings charged under the 2-form. In fact, 
according to the Distant Axionic String Conjecture (DASC) of \cite{Lanza:2020qmt,Lanza:2021udy,Lanza:2022zyg}, every infinite distance limit in four dimensions can be realized as the endpoint of a flow induced by the backreaction of precisely those axionic strings which become weakly coupled in the limit.
The emergent strings featuring in the Emergent String Conjecture are special examples of such axionic or EFT strings given by {\it critical} strings.

It is then natural to speculate whether the asymptotically weakly coupled EFT strings could be the source of the super-extremal tower of states in general weak coupling limits. Intriguingly, it was observed in \cite{Heidenreich:2021yda,Kaya:2022edp} that the tension of these strings parametrically coincides with the magnetic weak gravity scale,
\begin{align} \label{scalere-1}
 T_\text{\tiny EFT} \sim  \Lambda_\text{\tiny WGC} = g_\text{\tiny YM} M_\text{\tiny Pl}  \,. 
    \end{align}
Therefore, if one treats the EFT strings as the source of particle-like excitations, as would be appropriate for a critical string, the relation \eqref{scalere-1} may suggest that this tower contains the super-extremal tower required by the asymptotic WGC.

To test this idea, we will work, for concreteness, in the context of F-theory compactified to four dimensions. 
We will
analyze the possible weak coupling limits for a gauge theory realized on a stack of 7-branes wrapping a divisor on the base of an elliptic Calabi-Yau $4$-fold, extending both the analysis of \cite{Lanza:2020qmt,Lanza:2021udy,Lanza:2022zyg} and of \cite{Lee:2019tst,Klaewer:2020lfg}.

As our first result, which in fact is of interest independently of the WGC, we will classify the possible infinite distance limits in the K\"ahler moduli space of the base $B_3$ in the language of EFT strings.
The building blocks of such limits are what we call quasi-primitive EFT strings, which are obtained by wrapping D3-branes along certain curves in the movable cone \cite{Lanza:2020qmt,Lanza:2021udy,Lanza:2022zyg} of $B_3$.\footnote{See Appendix \ref{app:geometrytools} for the definition of the movable cone, which as explained in \cite{Lanza:2020qmt,Lanza:2021udy} underlies the structure of EFT strings in the K\"ahler moduli sector of F-theory, and Definitions \ref{def:EFTlimit} and \ref{def:quasiprimitive} for the precise notion of (quasi-) primitive EFT strings.} We will characterize the quasi-primitive strings via a topological invariant of their associated curve. Depending on the type of string under consideration, its associated limit either corresponds to an emergent string limit or to a decompactification limit to six or ten dimensions (see Figure \ref{fig:q=1}).
This is in agreement with the findings of \cite{Lee:2019tst,Klaewer:2020lfg} and sheds further light on the systematics of the decompactification limits.

We then turn to the effective action governing the dynamics of the EFT strings which characterize the various weak coupling limits. For the special case of primitive EFT strings, the action resembles the effective action of a weakly coupled heterotic
string, up to a numerical factor that is related to the topological invariant classifying the strings. While this might motivate  treating the asymptotically weakly coupled EFT strings as having particle-like excitations, the would-be excitations turn out to violate the super-extremality condition -- except in the special case of a critical heterotic string. 
At first sight, this seems to conflict with the predictions of the asymptotic tower WGC. Upon closer inspection, however, we find that all weak coupling limits other than the emergent string limits leave the realm of validity of the asymptotic tower WGC conjecture. The reason is that in such limits, the species scale associated with the KK tower, $\Lambda_\text{\tiny sp,KK}$, sits at or even below the WGC scale $\Lambda_\text{\tiny WGC}$, which coincides with the asymptotic EFT string tension. This interplay between the different characteristic scales is summarized in Figure \ref{fig:q=012sp}.
As a result, the gauge theory either decompactifies to an  -- in general --  non-weakly coupled theory in higher dimensions, or even to a defect theory from which gravity decouples completely. In both cases, the usual bottom-up motivation for the tower WGC, and in fact even for the WGC as such, does not apply.

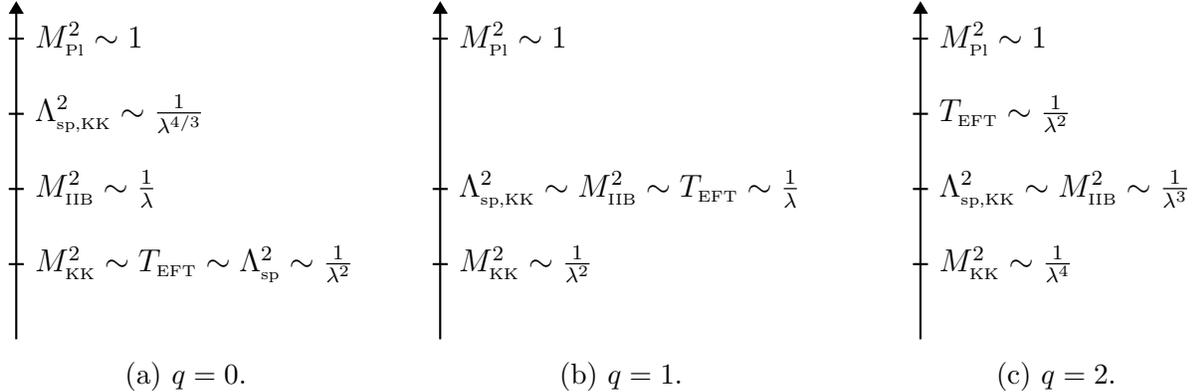
\begin{figure}[!htp]
    \centering
    \begin{subfigure}[t]{0.33\textwidth}
    \centering
 \begin{tikzpicture}
   \draw[thick,-Triangle] (0,0) -- (0,4.5);
   \draw[thick] (-0.1,4) -- (0.1,4) node[right] {$M_\text{\tiny Pl}^2\sim 1$};
   \draw[thick] (-0.1,3) -- (0.1,3) node[right] {$\Lambda_\text{\tiny sp,KK}^2 \sim \frac{1}{\lambda^{4/3}}$};
   \draw[thick] (-0.1,2) -- (0.1,2) node[right] {$M_\text{\tiny IIB}^2 \sim \frac{1}{\lambda}$};
   \draw[thick] (-0.1,1) -- (0.1,1) node[right] {$M_\text{\tiny KK}^2 \sim T_\text{\tiny EFT}\sim \Lambda_\text{\tiny sp}^2\sim \frac{1}{\lambda^2}$};
\end{tikzpicture}
    \caption{$q=0$.}
    \label{fig:q=0sp}
    \end{subfigure}\hfill
    \begin{subfigure}[t]{0.33\textwidth}
    \centering 
    \begin{tikzpicture}
   \draw[thick,-Triangle] (0,0) -- (0,4.5);
   \draw[thick] (-0.1,4) -- (0.1,4) node[right] {$M_\text{\tiny Pl}^2\sim 1$};
   \draw[thick] (-0.1,2) -- (0.1,2) node[right] {$\Lambda_\text{\tiny sp,KK}^2\sim M_\text{\tiny IIB}^2 \sim T_\text{\tiny EFT} \sim \frac{1}{\lambda}$};
   \draw[thick] (-0.1,1) -- (0.1,1) node[right] {$M_\text{\tiny KK}^2\sim \frac{1}{\lambda^2}$};;
\end{tikzpicture}
    \caption{$q=1$.}
    \label{fig:q=1sp}
    \end{subfigure}\hfill
    \begin{subfigure}[t]{0.33\textwidth}
    \centering
    \begin{tikzpicture}
   \draw[thick,-Triangle] (0,0) -- (0,4.5);
   \draw[thick] (-0.1,4) -- (0.1,4) node[right] {$M_\text{\tiny Pl}^2 \sim 1$};
   \draw[thick] (-0.1,3) -- (0.1,3) node[right] {$T_\text{\tiny EFT} \sim \frac{1}{\lambda^2}$};
   \draw[thick] (-0.1,2) -- (0.1,2) node[right] {$\Lambda_\text{\tiny sp,KK}^2 \sim M_\text{\tiny IIB}^2\sim \frac{1}{\lambda^3}$};
   \draw[thick] (-0.1,1) -- (0.1,1) node[right] {$M_\text{\tiny KK}^2\sim \frac{1}{\lambda^4}$};
\end{tikzpicture}
    \caption{$q=2$.}
    \label{fig:q=2sp}
    \end{subfigure}
    \caption{Characteristic scales for the three possible types of quasi-primitive EFT string limits in the K\"ahler moduli space of F-theory compactified to four dimensions. The limits are parametrized by $\lambda \to \infty$.
    The topological invariant $q$ classifies the possible quasi-primitive EFT string limits, as explained in Section \ref{ssec:Weakcouplinglimits}.
    The limit with $q=0$ is an emergent string limit, which stays effectively four-dimensional. Limits with $q=1$ describe effective decompactifications to six dimensions, while limits with $q=2$ decompactify to ten dimensions and the 7-brane gauge sector is a defect theory. This picture summarizes the findings of Section \ref{ssec:stringscalevsSC}.}
    \label{fig:q=012sp}
\end{figure}

The upshot of our findings is that the excitations of a weakly coupled EFT string in F-theory can account for a tower of super-extremal states if and only if the string corresponds to a heterotic emergent string.
Consistently, in all other weak coupling limits of a 7-brane gauge theory in F-theory the WGC is less obviously motivated also from a bottom-up perspective, to the extent that we leave the regime of a weakly coupled gauge theory coupled to a perturbative gravitational sector.

In principle, this does not preclude that the WGC may nonetheless be satisfied even in such situations, but there are no obvious candidates for the super-extremal states, in particular not for a marginally super-extremal tower of states. At the same time, the existence of such a tower would, in a sense, run counter to the very motivation behind the WGC. 

 While we focus in this work on weak coupling limits for gauge theories on 7-branes in four-dimensional F-theory, the three types of patterns which we find (see Figure \ref{fig:q=012sp}) generalize to other setups. 
 For example, as we will explain, the weak coupling limit on a stack of D3-branes in Type IIB orientifolds decouples the gauge and gravity sectors 
 similar to the behavior in Figure \ref{fig:q=2sp}. 
 The perturbative open string excitations only furnish a finite number of highly super-extremal states, rather than a tower of marginally super-extremal states as in weak coupling limits of the emergent heterotic string type.
 This is no surprise because a decoupling between the gauge and gravity sector trivializes the weak gravity constraint.

\subsubsection*{Structure of the paper}

The paper is organized as follows: 
Section \ref{sec:primitiveEFTFth} contains the geometric part of our analysis, which is of interest independently of the Weak Gravity Conjecture.
In Section \ref{ssec:review} and Appendix \ref{app:geometrytools} we review the definition of EFT strings and the F-theory K\"ahler field space. Section \ref{ssec:Weakcouplinglimits} introduces the notion of quasi-primitive EFT strings and their classification, together with the systematics of the EFT string limits in F-theory. 
For better readability of the paper, the oftentimes technical proofs of our main results are relegated to Appendix \ref{app:proofs}. In Section \ref{ssec:P1fibFn} and Appendix \ref{app:examples} we discuss an exhaustive series of examples that illustrate all the main results of our geometric findings.

Section \ref{sec:NoncritEFTWGC}  analyzes the tower Weak Gravity Conjecture for weak coupling limits in F-theory. In Section \ref{ssec:Repulsiveforcecond} we study the super-extremality bound for the (putative) excitations of a primitive EFT string, showing that the tower WGC is satisfied by the excitations of primitive EFT strings only when the latter correspond to emergent critical heterotic strings. In Section \ref{ssec:stringscalevsSC} we derive 
the relation summarized in Figure \ref{fig:q=012sp} between the weak gravity scale and the species scale for more general quasi-primitive EFT string limits.
In Section \ref{ssec:nonEFTgaugetheory} 
we generalize this to the most general weak coupling limits in F-theory.
 In Section \ref{ssec:4dWGC} we
explain that, except for the emergent string limit, the weak coupling limits leave the regime of applicability of at least the {\it asymptotic tower WGC}, reconciling our findings of Section \ref{ssec:Repulsiveforcecond} with the Swampland philosophy. Section \ref{sec:Conclusions} contains the summary of our results, along with a discussion of other setups in which the tower Weak Gravity Conjecture is not required to be satisfied and more speculative remarks. 

\section{EFT string limits in F-theory}
\label{sec:primitiveEFTFth}

The starting point for our classification of different weak coupling limits is a certain class of infinite distance limits introduced in \cite{Lanza:2020qmt,Lanza:2021udy,Lanza:2022zyg} and studied further in \cite{Marchesano:2022avb,Grimm:2022sbl}: As their main property, they can be reached as the end-point of Renormalization Group (RG) flows induced by axionic, or EFT, strings in four dimensions.
Throughout this work, we refer to them as \textit{EFT string limits}. This section serves to introduce the basic notion of such EFT string limits and to describe how they are realized in the F-theory K\"ahler field space. In Section~\ref{ssec:review} we start by giving a brief review of the EFT string flows realizing infinite distance limits in field space \cite{Lanza:2020qmt,Lanza:2021udy} and introduce the basics of the F-theory K\"ahler field space to which we will apply the EFT string analysis. In Section \ref{ssec:Weakcouplinglimits} we then provide a refined classification of certain EFT string limits in F-theory which we will call quasi-primitive EFT string limits: Their associated EFT strings are obtained by wrapping D3-branes on special movable curves in the base $B_3$ of the elliptically fibered Calabi--Yau $4$-fold. These curves can be written as intersections of certain K\"ahler classes and admit an intriguing classification via a certain topological invariant that will govern the physics of the EFT string limit.
We illustrate our findings in a simple example in Section \ref{ssec:P1fibFn} and in more complicated settings in Appendix \ref{app:examples}.

The quasi-primitive limits will then serve as the starting point to characterize even the most general weak coupling limits for gauge theories in Section \ref{sec:NoncritEFTWGC}.

\subsection{Review of EFT strings and F-theory K\"ahler field space}
\label{ssec:review}

Consider a  $\mathcal{N}=1$ supersymmetric EFT in four dimensions. Complex scalar fields in such a theory reside in chiral multiplets. Let us denote a subset of these chiral scalar fields by $T_i$, $i=1,\ldots, n$, and the field space spanned by them as $\mathcal{M}$. We further assume that the imaginary parts of the $T_i$ are periodic, such that we can identify $T_i\sim T_i +\ii$. For a general $\mathcal{N}=1$ EFT in four dimensions, the field space $\mathcal{M}$ is not an actual moduli space, since the fields $T_i$ can become massive due to the presence of a non-trivial scalar potential $V$. For instance, in the EFT defined by a string compactification, such a potential can be induced by background fluxes or through non-perturbative effects. In this work, we do not turn on fluxes such that we only have to worry about a potential induced by non-perturbative effects. In a general $\mathcal{N}=1$ EFT in four dimensions the scalar F-term potential can be expressed in terms of a superpotential $W$ as 
\begin{align}
    V =e^K\left(g^{i\bar j} D_i W \bar{D}_{\bar j} \bar W-3|W|^2\right)\,.
\end{align}
Here $K=K(T_i)$ is the K\"ahler potential, $g^{i\bar j} = \partial_{T_i} \partial_{\bar{T}_{\bar j}} K$ the metric on field space and $D^i = \partial_{T_i} +\partial_{T_i} K$ the K\"ahler covariant derivative. In the absence of classical contributions, the superpotential $W$ only receives corrections from instantons such that
\begin{align}
    W = \sum_{{\bf{m}}} \mathcal{A}_{{\bf{m}}} e^{-S_{{\bf m}}}\,.
    \label{eq:Wnonpert}
\end{align}

Here ${\bf{m}}=(m^1, \ldots, m^n)$ labels the instanton charge under the shift symmetry $T_i\rightarrow T_i +\ii$, $S_{{\bf m}}=2\pi m^iT_i$ is the action of an instanton, and $\mathcal{A}$ is the one-loop determinant. It is expected \cite{Palti:2020qlc} that for a generic $\mathcal{N}=1$ theory at least some $\mathcal{A}_{\bf{m}}$ are non-zero. We are interested in situations in which we can neglect also 
these non-perturbative contributions to $W$. This is the case if $S_{\bf{m}}\rightarrow  \infty$ for all $\bf{m}$. In this regime, all instanton effects become irrelevant, and we recover a continuous shift symmetry $T_i\rightarrow T_i +\ii c_i$ for $c_i\in \mathbb{R}$. Thus, the imaginary part of $T_i$ can be treated as an axion, i.e., we can write $T_i =  s_i + \ii a_i$, and $\mathcal{M}$ takes the r\^ole of a quasi-moduli space. The relevant part of the 4d action is then given by 
\begin{align}
    S_{4d} = \frac{M_\text{\tiny Pl}^2}{2} \int \left( R \star \ID - g^{i\bar j} \mathrm{d}T_i \wedge \star \mathrm{d}\bar{T}_{\bar j}\right)\,. 
\end{align}
In the perturbative regime, the axion $\im T_i$ can be dualized into a $2$-form $B_{2,i}$ and we can consider the object carrying electric charge under this $2$-form, i.e., a string in 4d. The tension of this string is controlled by the linear multiplet obtained upon dualizing the chiral field $T_i$, 
\begin{align}
   L^i =- \frac{1}{2}\frac{\partial K}{\partial \re T_i}\,. 
\end{align}

The associated cosmic string solutions are determined as four-dimensional solutions to the equations of motion, preserving two-dimensional Poincar\'e invariance along the directions parallel to a string in 4d. This motivates the ansatz (cf. \cite{Lanza:2020qmt,Lanza:2021udy})
\begin{align}
    ds^2 = -dt^2 + dx^2 + e^{2D(z)} dz d\bar z\,,
\end{align}
where $z\in \mathbb{C}$ is the coordinate transverse to the string. Supersymmetric solutions to the equations of motion now correspond to \cite{Greene:1989ya}
\begin{align}
     \partial_{\bar z} T_i(z) =0 \,,\qquad e^{2D} = |f(z)|^2 e^{K}\,,
\end{align}
with $f(z)$ a holomorphic function.

Thus, the chiral fields $T_i$ must have a holomorphic profile along the directions transverse to the string. Of particular interest are such $\frac12$-BPS cosmic string solutions associated to a string carrying charge ${\bf e}=(e_1,\ldots, e_n)$ under the 2-form $B_{2,i}$. When encircling the string core, the chiral fields are expected to undergo a monodromy transformation of the form $T_i\rightarrow T_i+\ii e_i$. 
The holomorphic profile for $T_i$ respecting this symmetry is given by 
\begin{equation}\label{eq:tiprofile}
    T_i(z)= T_i^{(0)} - \frac{e_i}{2\pi}\log \left(\frac{z}{z_0}\right)\,,
\end{equation}
where $T_i^{(0)}$ is some background value for the scalar field and $z_0$ some constant. This profile for $T_i(z)$ is valid as long as $|z|\ll |z_0|$. The description of the backreaction of the string is self-consistent if the limit $z\rightarrow 0$ corresponds to the regime where $S_{\bf m}\rightarrow  \infty$ for all instantons charged under the shift induced by the string. In this case, the continuous shift symmetry of $T_i(z)$ is approximately realized, and we can indeed treat the fields $\im T_i$ as axions. Furthermore, the contribution to the non-perturbative superpotential from the instantons charged under the string shift symmetry can be neglected. If we tune the ``spectator fields'', i.e., the fields that do not exhibit a $z$-dependent profile due to the backreaction of the string, to suitable values we can then completely neglect $W$ and hence any scalar potential for the fields $T_i$.  In the vicinity of the core, we reach the limit $T_i\rightarrow \infty$. A string leading to a self-consistent backreaction was dubbed EFT string in \cite{Lanza:2021udy}. In particular, the limit $z\rightarrow 0$ then corresponds to an infinite distance limit in $\mathcal{M}$. Therefore, the EFT strings provide a very useful way to study infinite distance limits in field space.

We can make the notion of EFT string more precise, by considering the non-perturbative corrections as, e.g., in \eqref{eq:Wnonpert}. Defining $T_i = s_i + \ii a_i$ and ${\bf{s}}=(s_1,\ldots,s_n)$ the collection of all saxions, the perturbative regime is obtained whenever 
\begin{equation}
    \left|e^{-2\pi m^i T_i}\right| = e^{-2\pi \langle {\bf m}, {\bf s}\rangle} \ll 1\coma
    \label{eq:pertregime}
\end{equation}
where we have introduced the pairing $\langle {\bf m}, {\bf s}\rangle=m^is_i$ between the instanton charges $m^i$ and the saxions $s_i$. Following \cite{Lanza:2021udy}, let us denote the set of all instanton charges correcting the effective action by $\mathcal{C}_I$. For \eqref{eq:pertregime} to hold, we need that $\langle {\bf m},{\bf s}\rangle\gg1$. Therefore, the saxions $s_i$ need to lie inside the \textit{saxionic cone} $\Delta \equiv \left\{{\bf s}|\langle {\bf m}, {\bf s}\rangle>0\,,\, \forall {\bf m} \in \mathcal{C}_I\right\}$, which can be viewed as $\Delta=\mathcal{C}_I^\vee \otimes \mathbb{R}$. Inserting the profile \eqref{profilebackreact} in \eqref{eq:pertregime} one finds that under a string flow the non-perturbative contributions behave like 
\begin{align}
    e^{-2\pi \langle {\bf m},{\bf s}\rangle} = e^{-2\pi \langle {\bf m}, {\bf s}_0\rangle} \left(\frac{z}{z_0}\right)^{\langle {\bf{m}},{\bf{e}}\rangle}\,,
\end{align}
where ${\bf e}\in \mathcal{C}^S$ is an element of the lattice $\mathcal{C}^S$ of BPS string charges. Assuming that ${\bf s}_0$ is chosen inside $\Delta$, in order not to spoil the perturbative description in the limit $z\rightarrow 0$, we need $\bf{e}$ to satisfy $\langle {\bf m},{\bf e}\rangle>0$ defining the sub-cone 
\begin{align}
\mathcal{C}^S_\text{\tiny EFT}\equiv \left\{ {\bf e} \in \mathcal{C}_S| \langle {\bf m}, {\bf e}\rangle\geq 0\,,\, \forall {\bf m} \in \mathcal{C}_I\right\}\subset \mathcal{C}_S,
\end{align}
i.e., the cone dual to $\mathcal{C}_I$ \cite{Lanza:2021udy}. In particular, an \emph{elementary} EFT string charge is a generator of the cone $\mathcal{C}^S_\text{\tiny EFT}$ corresponding to a string that carries unit charge under a single 2-form $B_{2,i}$ \cite{Lanza:2021udy}. In the following, we will be interested in EFT strings that are charged under a single 2-form, but not necessarily with unit charge. We will refer to such strings as \emph{primitive} EFT strings. Thus, every elementary string is a primitive string, but not every primitive string is elementary. 

In this work, we are studying infinite distance limits corresponding to weak coupling limits for gauge theories in F-theory compactified on an elliptically Calabi--Yau $4$-fold $Y_4$ with a smooth projective variety $B_3$ as base. 
Given the relation between EFT strings and infinite distance limits, we should be able to realize certain weak coupling limits as EFT string limits for some choice of EFT string charges. We are primarily interested in the scalar field space $\mathcal{M}$ of the effective 4d EFT associated to the K\"ahler deformations of $B_3$. A basis of the complex scalar fields spanning $\mathcal{M}$ is given by
\begin{equation}\label{eq:Ti}
    T_i = \frac{1}{2} \int_{D_i} J\wedge J + \ii \int_{D_i} C_4\,.
\end{equation}
Here, $\{D_i\}$, $i=1,\ldots, h^{1,1}(B_3)$ is a basis of generators of the cone of effective divisors of $B_3$, $C_4$ is the type IIB RR $4$-form and $J$ the K\"ahler form on $B_3$. 
Note that if the effective cone is non-simplicial, the definition of the scalar fields $T_i$ depends on the choice of a basis of generators of $\text{Eff}^1(B_3)$. In the following we denote by $\{D_i\}$ such a basis of generators.
Similarly, we can expand the K\"ahler form as 
\begin{equation}
    J= v^i J_i\,,
    \label{eq:Jform}
\end{equation}
where $J_i$ are a basis of generators of the K\"ahler cone $\mathcal{K}(B_3)$.  The geometry of the scalar field space is governed by a K\"ahler potential given by 
\begin{equation}\label{eq:K}
    K=-2\log \mathcal{V}_{B_3}\,,
\end{equation}
where $\mathcal{V}_{B_3}$ is the volume of $B_3$ measured in units of the type IIB string scale $M_\text{\tiny IIB}$ and given by 
\begin{equation}
    \mathcal{V}_{B_3}= \frac{1}{6}\int_{B_3} J^3\,, 
\end{equation}
in the large volume approximation. The effective bosonic 4d action including the scalar fields then takes the form
\begin{equation} \label{kinetic-2}
    S_{4d}= \frac{M_\text{\tiny Pl}^2}{2} \int \left( R\star \ID - g^{i \bar j}\,dT_i \wedge \star d\bar{T}_{\bar j}\right)+\ldots \,,
\end{equation}
where $g^{i\bar j}=\partial_{T_i} \partial_{\bar T_{\bar j}} K$ is the classical metric on $\mathcal{M}$. The dots in $S_{4d}$ stand for additional terms describing, e.g., 1-form gauge theories. In F-theory, this includes the gauge sector from stacks of 7-branes wrapping effective divisors $\mathcal{S}=\sum a^iD_i$, with $a^i \in \mathbb{R}_{\geq 0}$. The gauge coupling for these gauge theories is set by the volume $\mathcal{V}_{\mathcal{S}}$ of $\cal{S}$ as
\begin{equation}
    \frac{2\pi}{g_\text{\tiny YM}^2} =  \mathcal{V}_{\mathcal{S}}\,.
\end{equation}
The weak coupling limits thus correspond to the regime $\mathcal{V}_{\mathcal{S}}\rightarrow \infty$. Since $\mathcal{S}$ is a linear combination of the $D_i$ with non-negative coefficients, we can study weak coupling limits by considering limits where $T_i\rightarrow \infty$ for some $i$. As required by the DASC \cite{Lanza:2020qmt,Lanza:2021udy}, these limits can be obtained as end-points of certain EFT string flows. Since the relevant strings should be magnetically charged under the axionic component of $T_i$, they correspond to D3-branes wrapped on curves in $B_3$. 

  According to our previous discussion, the backreaction of such D3-brane strings induces a profile for the saxionic partners of the axions, which in this case are the divisor volumes $\re T_i$. To find a basis of EFT strings, we first consider \textit{elementary} EFT strings, i.e., strings carrying unit charge under only one $2$-form with respect to the chosen basis of chiral fields $T_i$ and their dual linear multiplets $L^i$.  In this case, the saxionic profile close to the EFT string core reads
\begin{equation}\label{profilebackreact}
    T_0 = T_0^{(0)} -\frac{1}{2\pi}\log\frac{z}{z_0}+\ldots\,, \qquad T_i=T_i^{(0)} \quad  \forall i\neq 0\,. 
\end{equation}
In the vicinity of $z=0$ this leads to a limit in field space given by
\begin{equation}\label{eq:singleEFTlimit}
    \mathcal{V}_{D_0}\rightarrow \infty \coma \mathcal{V}_{D_i}=\text{const.} \,\quad \text{for} \quad i\neq 0\,.
\end{equation}
As discussed in \cite{Lanza:2021udy}, the elementary EFT strings inducing these kinds of limits are obtained from D3-branes wrapping the generators of the cone of movable curves, $\text{Mov}_1(B_3)$, of $B_3$. To see this, let us consider the saxionic form
\begin{equation}\label{eq:bolds}
    {\bf{s}}=\frac{1}{2} J\wedge J\,. 
\end{equation}
This form can be expanded in terms of a basis of curve classes $C^i$ as 
\begin{equation}
    {\bf{s}} = s_i C^i\,.
\end{equation}
By comparison with \eqref{eq:Ti}, the $s_i$ can be identified with the saxions, $s_i=\re T_i$, by choosing the curves $C^i$ such that  
\begin{equation}\label{duality}
    C^i\cdot D_j = \delta^i_j\,. 
\end{equation}
 If $\text{Eff}^1(B_3)$ is simplicial, this means that the $C^i$ are a basis of curves dual to the $D_j$. In this case, since $D_j$ are generators of the cone of effective divisors, $\text{Eff}^1(B_3)$, and $\overline{\text{Eff}}^1(B_3)^\vee=\text{Mov}_1(B_3)$~\cite{MR3019449}, we conclude that the curves $C^i$ are generators of the movable cone. Modulo a subtlety that will be addressed momentarily, one expects the F-theory EFT strings to arise from D3-branes on curves $C\in \text{Mov}_1(B_3)$ leading to limits of the form \eqref{eq:singleEFTlimit}. If $\text{Eff}^1(B_3)$ is non-simplicial, while it is still true that $\overline{\text{Eff}}^1(B_3)^\vee=\text{Mov}_1(B_3)$,  for a given basis $\{D_i\}$ of generators not all curves satisfying \eqref{duality} are in $\text{Mov}_1(B_3)$. If for instance $C^0$ is non-movable, it cannot lead to an EFT string since there exists an effective divisor that shrinks in the EFT string limit and hence gives rise to an unsuppressed instanton close to the string core. Based on the given choice of basis for $\text{Eff}^1(B_3)$, one would then conclude that there is no EFT string realizing the EFT string limit \eqref{profilebackreact} for $\text{Re}\,T_0$. By choosing a different basis of generators of $\text{Eff}^1(B_3)$ this problem may be circumvented, but it is not guaranteed that such a basis exists within the given K\"ahler cone. 
 
We now address the subtlety alluded to above. Even if $\text{Eff}^1(B_3)$ is simplicial, 
 not every EFT string limit may be realizable in a given chamber of the K\"ahler cone. To see this, 
  notice that by \eqref{eq:bolds}, the saxionic form $\bf{s}$ depends on the choice of the K\"ahler form $J$. As stressed in \cite{Lanza:2021udy}, the map $\mathcal{K}(B_3)\ni J\mapsto \frac12 J\wedge J\in \text{Int}(\text{Mov}_1(B_3)) $ is in general not surjective. Therefore, not all limits of the form \eqref{eq:singleEFTlimit} necessarily exist for a given base $B_3$. 
  We will see the origin of this complication more clearly in the next section. This problem can be circumvented by working with bases $B_3$ modulo birational equivalence \cite{Lanza:2021udy}, i.e., by taking into account all bases $ B_3'$ that are isomorphic to $B_3$ outside higher-codimension loci. For instance, if $B_3'$ and $B_3$ are related by a flop transition, their respective cones of effective generators agree, but the cone of nef divisors might differ, i.e. $\mathcal{K}(B_3)\neq \mathcal{K}(B_3')$. Gluing together all K\"ahler cones obtained by such flops, one arrives at the extended K\"ahler cone
\begin{align}\label{eq:Kext}
    \mathcal{K}_\text{ext}(B_3) = \bigcup_{B_3' \sim B_3} \mathcal{K}(B_3')\,,
\end{align}
where ``${\sim}$" means that $B_3'$ and $B_3$ are related by a flop transition. The claim of \cite{Lanza:2021udy} (see also \cite{Xiao2016:aaa}) is now that the map
\begin{align}
   \mathcal{K}_{\text{ext}}(B_3) \ni J \mapsto \frac12 J\wedge J \in \text{Int}(\text{Mov}_1(B_3))
\end{align}
is a bijection. We are going to confirm this expectation in explicit examples in Appendix \ref{app:examples}. For most of this work we consider, however, the bases $B_3$ without taking into account the flop transitions. In other words, we restrict ourselves to a single chamber of the $\mathcal{K}_\text{ext}$ corresponding to a single element on the RHS of \eqref{eq:Kext}. The benefit of this is that
\begin{enumerate}[label={{\itshape\roman*})},ref={{\itshape\roman*})}]
\item we can work with an explicit base $B_3$ with a fixed set of generators of the K\"ahler cone, and 
\item we do not have to worry about subtleties arising at the boundaries $\mathcal{K}(B_3)$ corresponding to a flop transitions, for instance due to tensionless strings obtained from D3-branes wrapping the flopped curve. 
\end{enumerate}

\subsection{EFT string limits}
\label{ssec:Weakcouplinglimits}

We now study the limits of the form \eqref{eq:singleEFTlimit} that can be realized in a given chamber of $\mathcal{K}_\text{ext}$. Since in F-theory the saxionic components of the chiral fields correspond to the volumes of a basis of the generators of $\text{Eff}^1(B_3)$, the EFT string limits describe limits where (a subset) of these volumes diverge. Crucially, in order for such a limit to be an EFT string limit all volumes that are not kept constant have to scale at the same rate, since, by \eqref{profilebackreact}, all saxions with non-trivial profile have to scale homogeneously in the limit $z\rightarrow 0$. Any limit for which the volumes of the generators of $\text{Eff}^1(B_3)$ do not scale homogeneously can thus not be obtained as an EFT string limit. In F-theory we can therefore characterize an EFT string limit, i.e., a limit that can be obtained as the $z\rightarrow 0$ limit for a cosmic EFT string solution, via 

\begin{definition}[EFT string limits]\label{def:EFTlimit}
An EFT string limit in the F-theory K\"ahler quasi-moduli space is a limit in which the volume of a subset $\mathcal{I}\subset\{D_i\}$ of a given basis of generators of the effective cone of divisors $\text{Eff}^1(B_3)$ diverges homogeneously, i.e.,
\begin{equation}
    \mathcal{V}_{D}\sim \lambda \rightarrow \infty\,,\qquad \forall D\in \mathcal{I}\,,
\end{equation} 
while $\mathcal{V}_{\hat D}<\infty$ for $\hat{D}\notin \mathcal{I}$. In this case, we call the set $\mathcal{I}$ \emph{homogeneously expandable}. In particular, a \textit{primitive} EFT string limit corresponds to a limit for which $|\mathcal{I}|=1$.
\end{definition}

Among the EFT string limits, the {\it primitive} EFT string limits can be viewed as the basic  building blocks. Recall from the discussion in the paragraph after \eqref{eq:tiprofile} that a primitive EFT string limit is the limit induced by an EFT string charged only under a single 2-form $B_{2,i}$. As a result, out of the basis of generators $\{D_i\}$ of $\text{Eff}^1(B_3)$ only a single generator $D_i$ acquires a large volume in such limits, as in Definition \ref{def:EFTlimit}. This prompts the question under which condition the limit $\mathcal{V}_{D_i}\rightarrow \infty$ can be reached as such a primitive EFT string limit for a fixed base $B_3$ without taking into account flop transitions. A potential obstacle arises since the volume of divisors of $B_3$ are not all independent, but depend quadratically on the volumes of curves arising as the expansion parameters of the K\"ahler form $J$ in terms of a basis of K\"ahler cone generators (cf. \eqref{eq:Jform}).\footnote{\label{foot:non-simK} Notice that if the K\"ahler cone of $B_3$ is simplicial, the expansion parameters of $J$ are simply the volumes of the dual Mori cone generators. In case the K\"ahler cone is non-simplicial, one should specify a basis of K\"ahler cone generators in which $J$ is expanded. In this case the coefficients of the expansion do not necessarily correspond to the volumes of Mori cone generators.} It is therefore more practical to check for primitive EFT string limits at the level of curves in $B_3$. In the following, we always assume that we have specified a basis of K\"ahler cone generators. To obtain a minimal set of generators of $\text{Eff}^1(B_3)$ acquiring a large volume, we are advised to scale up a minimal set of $v^i$. Given an element $J_0$ of our chosen basis of K\"ahler cone generators we may thus ask if the limit $v^0\rightarrow \infty$ can correspond to a primitive EFT string limit.

If $v^0$ appears in the volume of more than one generator of $\text{Eff}^1(B_3)$, then in order to obtain a primitive EFT string limit, we must perform a {\it co-scaling}. This means  
 that we need to scale to zero or to infinity other $v^i$ in order to realize a 
 primitive EFT string limit, as in Definition \ref{def:EFTlimit}. 
 The scaling to zero may be required in order to minimize the set of divisors becoming large, in particular, for a primitive EFT string limit, to ensure there is only one generator of $\text{Eff}^1(B_3)$ with this property. 
  The following proposition gives a geometric criterion for a primitive EFT string limit to exist within the chosen chamber of $\mathcal{K}_{\rm ext}$:
\begin{proposition}\label{prop:primitiveEFT}
Consider a primitive EFT string limit for a generator $D$ of $\text{\normalfont Eff}\,{}^1(B_3)$,  corresponding to the 
limit $v^0\rightarrow \infty$  with $v^0$ being the volume of a curve $\mathcal{C}^0$ in the Mori cone contained in $D$.
Then 
all other generators of $\text{\normalfont Eff}\,{}^1(B_3)$ containing $\mathcal{C}^0$ are rational or genus-one fibrations with  $\mathcal{C}^0$ contained in the base of this fibration. 
\end{proposition}

We provide the proof of this Proposition in Appendix \ref{app:proofs}. While in classical moduli space, this is a sufficient criterion, we must in addition ensure that in the process of co-scaling no divisor volume scales to zero. If this happens, quantum corrections, for instance from D3-brane instantons, become non-negligible and the classical analysis breaks down. Hence, we only consider limits in which no divisor volume shrinks to zero to ensure validity of the effective supergravity approximation. To this end, the co-scaling may also require taking some of the $v^{i}$, $i\neq 0$ to infinity. Accordingly, Proposition \ref{prop:primitiveEFT} only provides a necessary condition for the EFT string limit to exist for a given generator $D$ of $\text{Eff}^1(B_3)$. 

From Proposition \ref{prop:primitiveEFT}, it is clear that a primitive EFT string limit may not exist for every generator of $\text{Eff}^1(B_3)$ within a given K\"ahler cone chamber. This means that in case a gauge theory is realized on a 7-brane wrapping such a divisor, its weak coupling limit cannot be studied by only taking primitive EFT strings into account. We therefore need to broaden the notion of ``minimal'' EFT strings to include the large volume limits that cannot be realized as primitive EFT string limits within the chosen chamber of $\mathcal{K}_\text{ext}(B_3)$. Again, we start by studying limits in which we scale up a single $v^i$ and check whether, via suitable co-scalings, we can obtain an EFT string limit. If this is possible, we are interested in limits in which the set $\mathcal{I}$ in Definition \ref{def:EFTlimit} is minimized. We dub the resulting limits \emph{quasi-primitive} EFT string limits, defined as follows: 
\begin{definition}[Quasi-primitive EFT string limit]\label{def:quasiprimitive}
Given a K\"ahler cone generator $J_0$  we can associate to it a \emph{quasi-primitive} EFT string limit  as the limit  $v^0\rightarrow \infty$ co-scaled in such a way to obtain an EFT string limit and to minimize $|\mathcal{I}|>0$ for the set $\mathcal{I}$ as in Definition \ref{def:EFTlimit}. In particular, a minimal set must only contain divisors whose volume does not scale to infinity (at the same rate as the other divisors in $\mathcal{I}$) only as a result of the required co-scaling of some other $v^i \neq v^0 \to \infty$. 
\end{definition}
Notice that, according to this definition, every primitive string limit is quasi-primitive.
 Note furthermore that there can be multiple quasi-primitive EFT string limits for a given K\"ahler  cone generator $J_0$ depending on the chosen co-scaling.  At the same time, it is still possible that a generator of $\text{Eff}^1(B_3)$ is not contained in $\mathcal{I}$ for any quasi-primitive EFT string limit. 
This can happen if the required co-scaling $v^0 \neq v^i \to \infty$ either leads to divisors with expanding volume not contained in the set of divisors becoming large only as a consequence of $v^0 \to \infty$, or if we leave the class of EFT limits altogether because some of the expanding divisors do so at different rates.
We will discuss the weak coupling limits for these theories separately in Section \ref{ssec:nonEFTgaugetheory}.

Having defined the EFT string limits in field space, we may ask about the EFT strings giving rise to such limits via their backreaction. As noticed before, an EFT string in F-theory can be obtained by wrapping a D3-brane on a movable curve. Following the discussion in Section \ref{ssec:review}, for a given EFT string limit, we can identify the EFT string giving rise to the limit by considering the instantons that become suppressed in that limit. The relevant instantons in our setup are Euclidean D3-branes wrapped on effective divisors $D=\sum_i m^i D_i$ of $B_3$. Their action is simply given by 
\begin{align}
    S_{\bf m} =2\pi m^i \,T_i\,,
\end{align}
where $T_i$ are the chiral scalars introduced in \eqref{eq:Ti}. These instantons are charged under the axionic shift induced by the EFT string solution \eqref{eq:tiprofile}. In the F-theory setting, this amounts to the statement that the curve $C$ giving rise to the string and the instanton divisor have non-vanishing intersection, i.e. $C\cdot D\neq 0$. In an EFT string limit, the action of an instanton vanishes, i.e. $\re S_{D3}\rightarrow \infty$, if and only if the instanton is charged under shift symmetries induced by the string. For a given quasi-primitive EFT string limit, this allows us to identify the string giving rise to the corresponding limit. More precisely we have the following
\begin{proposition}\label{prop:EFTstring}
Given a quasi-primitive EFT string limit associated to a K\"ahler cone generator $J_0$, then the EFT string giving rise to this limit is obtained by wrapping a D3-brane on 
\begin{enumerate}[label={\normalfont P\ref*{prop:EFTstring}\alph*.},ref={\normalfont P\ref*{prop:EFTstring}\alph*}, labelindent=0pt]
    \item\label{list:EFT1new} $C= \alpha\, J_0^2$ if $J_0^2\neq 0$, or
    \item\label{list:EFT2new} $C= \alpha \,J_0\cdot J_i$ if $J_0^2=0$ with $J_i\neq J_0$ a suitable K\"ahler cone generator, 
\end{enumerate}
for some $\alpha \in \mathbb{Q}_{>0}$.
\end{proposition}

We prove this Proposition in Appendix \ref{app:proofs}. 

Since the curves in Conditions \ref{list:EFT1new} and \ref{list:EFT2new} have non-negative intersection with the generators of $\text{Eff}^1(B_3)$, these curves are necessarily movable. However, in general, these curves are not necessarily generators of the movable cone. Notice that Proposition~\ref{prop:EFTstring} {\it assumes} that for a given K\"ahler cone generator, an associated quasi-primitive EFT string limit exists within the chosen K\"ahler cone chamber. This means that there is a co-scaling such that neither does $|\mathcal{I}|$ increase, nor does any generator of $\text{Eff}^1(B_3)$ shrink to zero size.  This is a non-trivial assumption, which may fail in general. In Appendix \ref{app:examples}, we discuss examples where curves of the form as in Condition \ref{list:EFT1new} or \ref{list:EFT2new} do not give rise to a quasi-primitive EFT string limit. However, following the discussion around \eqref{eq:Kext}, we expect that for each generator of $\text{Mov}_1(B_3)$, there exists a chamber of $\mathcal{K}_{\rm ext}$ in which the corresponding primitive EFT string limit can be obtained. 
By Proposition~\ref{prop:EFTstring}, this means that for each generator of $\text{Mov}_1(B_3)$ there exists a chamber of the extended K\"ahler cone in which the generator can be written as in Condition \ref{list:EFT1new} or \ref{list:EFT2new}. Indeed, this is confirmed in the examples of Appendix \ref{app:examples}.

As we have seen in Proposition \ref{prop:EFTstring}, every quasi-primitive EFT string can be written as the product of two K\"ahler cone generators. For a general curve of the form $C=D_1\cdot D_2$, with $D_1$ and $D_2$ effective (but not necessarily movable) divisors, we can define the two quantities 
\begin{equation}
    Q_1=D_1^2 \cdot D_2\coma  Q_2= D_1 \cdot D_2^2\,,
\end{equation}
which only depend on the intersection numbers of $D_1$ and $D_2$. As we will see below, such curves $D_1\cdot D_2$ on $B_3$ can be nicely classified according to the value of 
\begin{equation}\label{eq:q}
    q=\Theta(Q_1) + \Theta(Q_2) \,,
\end{equation}
where we defined the step-function 
\begin{equation}
    \Theta(x) =\left\{\begin{matrix} 1\,,& x>0\\0\,,&x=0\\ -1\,,&x<0
    \end{matrix}\right.\,. 
\end{equation}
Accordingly, $q$ takes values in $q\in \{0,\pm1,\pm 2\}$ and is non-negative for movable curves hosting (quasi-)primitive EFT strings since for these $D_{1,2}$ are K\"ahler cone generators. As it turns out, the crucial properties of (quasi-)primitive EFT string only depend on the value of the topological quantity $q$ which can therefore be used to classify (quasi-)primitive EFT string. For instance for $q=0$ we can show the following 
\begin{proposition}\label{prop:q0}
A quasi-primitive EFT string limit for a string with $q=0$ is a primitive EFT string limit.
\end{proposition}

The proof for this proposition can be found in Appendix \ref{app:proofs}.  Like Proposition~\ref{prop:EFTstring}, Proposition \ref{prop:q0} {\it assumes} that there exists a co-scaling leading to the quasi-primitive EFT string limit. As we illustrate in the examples in Appendix \ref{app:examples}, there may exist geometries where a $q=0$ curve of the form $J_0^2$ as in Proposition \ref{prop:EFTstring} does in fact not give rise to a (quasi-)primitive EFT string limit within a chosen chamber of $\mathcal{K}_{\rm ext}$.

It was shown in \cite{Lee:2019tst}, that a curve with $q=0$ corresponds to the generic fiber of a genus-one/rational fibration and the D3-brane wrapped on it is dual to a critical type II/heterotic string. The relation between the $q=0$ curve $C$, the generator of the K\"ahler cone $J_0$ and the Mori cone generator $\mathcal{C}^0$ is illustrated in Figure \ref{fig:q=0}. 

On the other hand, in order for a $q=1$ curve $C$ to exist, there must exist at least one generator $J_0$ of the K\"ahler cone such that $J_0^2=0$. This implies that $B_3$ is a surface fibration with generic fiber $J_0$. From Proposition \ref{prop:EFTstring}, in order for the curve $C$ to have $q=1$, there must exist a second K\"ahler cone generator $J_1$ with $J_1^2\cdot J_0\neq 0$.
This means that the generator $J_1$ must be a vertical divisor with respect to the surface fibration. 
The curve $C$ is then a curve inside the surface fiber. Figure \ref{fig:q=1} provides an illustration of the relation between $J_0$, $C$ and the curve $\mathcal{C}^0$ with volume $\mathcal{V}_{\mathcal{C}^0}=v^0$ associated to the (quasi-)primitive EFT string limit. 

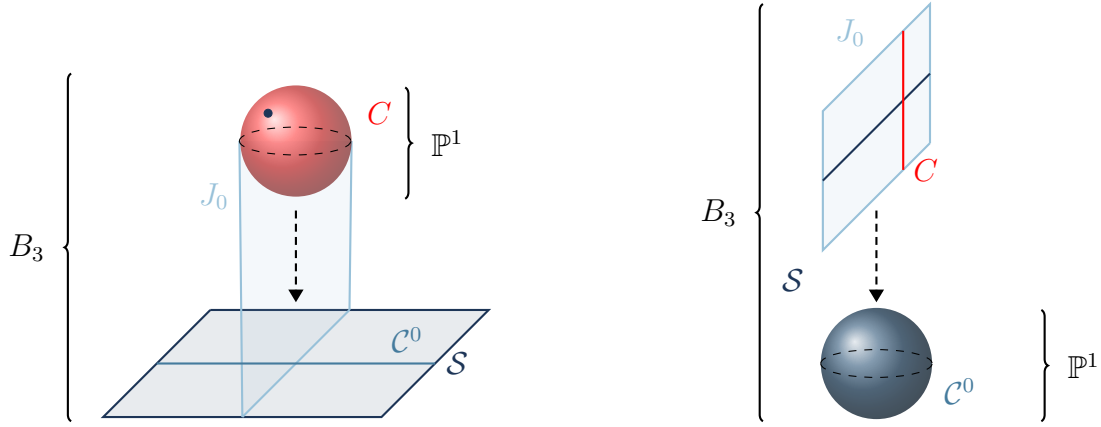
\begin{figure}[!htp]
    \centering
    \begin{subfigure}[t]{0.495\textwidth}
    \centering
    \begin{tikzpicture}[scale=3.7]
   \draw[fill=prcolor,fill opacity=0.1,draw=prcolor,thick] (0,0,0) -- (1,0,0) -- node[right,midway,opacity=1,prcolor] {$\mathcal{S}$} (1,0,1) -- (0,0,1) -- (0,0,0);
   \draw[fill=tercolor,fill opacity=0.1,draw=none,thick] (0.5,0,0) -- (0.5,0,1) -- node[left,pos=0.8,tercolor,opacity=1] {$J_0$} (0.297,0.8,0.5) -- (0.7,0.8,0.5) -- (0.5,0,0);
   \draw[thick,tercolor] (0.5,0,0) -- (0.5,0,1);
   \draw[thick,tercolor] (0.5,0,1) -- (0.297,0.8,0.5);
   \draw[thick,tercolor] (0.5,0,0) -- (0.7,0.8,0.5);
   \draw[thick,seccolor] (0,0,0.5) -- node[above,pos=0.9] {$\mathcal{C}^0$} (1,0,0.5);
   \shade[ball color = prhigh, opacity = 0.6] (0.5,0.8,0.5) circle (0.2);
   \node[prhigh] at (0.8,0.9,0.5) {$C$};
   \draw[dashed] (0.297,0.8,0.5) arc (180:0:0.2 and 0.05);
   \draw[dashed] (0.297,0.8,0.5) arc (180:0:0.2 and -0.05);
   \node[circle,draw=prcolor,fill=prcolor,inner sep=0pt,minimum size=3pt] at (0.4,0.9,0.5) {};
   \draw[thick,-Triangle,densely dashed] (0.5,0.55,0.5) -- (0.5,0.22,0.5);
   \draw [thick,decorate, decoration = {brace}] (-0.5,-0.4,0) --  node[left,pos=0.5,xshift=-0.2cm] {$B_3$} (-0.5,0.85,0);
   \draw [thick,decorate, decoration = {brace,mirror}] (0.7,0.4,0) --  node[right,pos=0.5,xshift=0.2cm] {$\PP^1$} (0.7,0.8,0);
\end{tikzpicture}
\caption{A $q=0$ EFT string curve $C=J_0^2$, with $J_0^3=0$, is a $\PP^1$-fiber.}
\label{fig:q=0}
    \end{subfigure}\hfill
    \begin{subfigure}[t]{0.495\textwidth}
    \centering
    \begin{tikzpicture}[scale=3.7]
   \draw[fill=tercolor,fill opacity=0.1,draw=tercolor,thick] (0,0,0) -- (0,0.5,0) -- node[above left ,midway,opacity=1,tercolor] {$J_0$} (0,0.5,1) -- (0,0,1) -- (0,0,0);
   \draw[thick,prcolor] (0,0.25,0) -- (0,0.25,1);
   \draw[thick,prhigh] (0,0.5,0.25) -- node[prhigh,pos=1,xshift=0.3cm] {$C$} (0,0,0.25);
   \draw[thick,-Triangle,densely dashed] (0,-0.05,0.5) -- (0,-0.38,0.5);
   \shade[ball color = seccolor!50!prcolor, opacity = 0.8] (0,-0.6,0.5) circle (0.2);
   \node[prcolor] at (-0.3,-0.3,0.5) {$\mathcal{S}$};
   \node[seccolor] at (0.3,-0.7,0.5) {$\mathcal{C}^0$};
   \draw[dashed] (-0.2,-0.6,0.5) arc (180:0:0.2 and 0.05);
   \draw[dashed] (-0.2,-0.6,0.5) arc (180:0:0.2 and -0.05);
   \draw [thick,decorate, decoration = {brace}] (-0.6,-1,0) --  node[left,pos=0.5,xshift=-0.2cm] {$B_3$} (-0.6,0.5,0);
   \draw [thick,decorate, decoration = {brace}] (0.4,-0.6,0) --  node[right,pos=0.5,xshift=0.2cm] {$\PP^1$} (0.4,-1,0);
\end{tikzpicture}
\caption{A $q=1$ EFT string curve $C=J_0\cdot J_i$ lies in a surface fiber of $B_3$.}
\label{fig:q=1}
    \end{subfigure}
    \caption{Schematic representation of quasi-primitive EFT strings with $q = 0$ (Figure \ref{fig:q=0}) and $q=1$ (Figure \ref{fig:q=1}). The limits correspond to expanding the base of the rational fibration ($q=0$) or of the surface fibration ($q=1$). EFT string limits with $q=2$ correspond to a homogeneous decompactification.}
    \label{fig:q=01}
\end{figure}

\subsection{A simple example}
\label{ssec:P1fibFn}

Let us illustrate our discussion so far, and in particular Proposition \ref{prop:EFTstring}, in a simple example. In Appendix \ref{app:examples}, we discuss various variations of this geometry, exemplifying how in settings where $\mathcal{K}_{\rm ext}$ has multiple chambers, not all primitive EFT string limits can necessarily be obtained in all chambers of $\mathcal{K}_{\rm ext}$. 

The example we would like to study corresponds to a $\PP^1$-fibration over $\mathbb{F}_n$. We denote the generators of the K\"ahler cone of $\FF_n$ by $j_0$ and $j_1$, chosen such that the intersection polynomial is
\begin{equation}
    \mathcal{I}(\FF_n) = n j_0^2 +j_0\cdot j_1\fstop
\end{equation}
The twist of the $\PP^1$-fibration is encoded in a line bundle $\mathcal{T}$ with 
\begin{equation}\label{ex1:twist}
   c_1 (\mathcal{T}) = s j_0  \coma s \geq 0\fstop
\end{equation}
The K\"ahler cone of $B_3 = \PP^1 \stackrel{\mathcal{T}}{\rightarrow} \FF_n$ is generated by
\begin{equation}
    J_0  = p^*j_0 \coma J_1 = p^*j_1 \coma J_3 = S_- + p^*c_1(\mathcal{T})\coma
\end{equation}
where $S_-$ is the exceptional section of the fibration $p :\, \PP^1 \rightarrow \FF_n$.\footnote{We call the third generator $J_3$ instead of $J_2$ to unify the notation with the example in Appendix \ref{app:toricBlF2}, where we will consider blow-ups of $\FF_n$.} The intersection ring for $B_3$ can be computed as
\begin{equation}
\begin{split}
   \mathcal{I}(B_3)  =&  s^2 nJ_3^3 +s nJ_0\cdot  J_3^2 +sJ_1\cdot J_3^2+nJ_0^2\cdot J_3+J_0 \cdot J_1 \cdot J_3\fstop
   \end{split}
\end{equation}
For a $\PP^1$-fibration the anticanonical class reads \cite{Friedman:1997yq}
\begin{equation}\label{eq:antiKP1fib}
    \overline{K} = 2S_- + p^*c_1(\mathcal{T}) + p^*c_1(B_2)\coma
\end{equation}
which means that for the present example $\overline{K} = 2J_3 -(s-2)J_0+(2-n)J_1$.
The cone of effective divisors is generated by 
\begin{equation}
\label{eq:EffcondivP1Fn}
   \text{Eff}^1(B_3) = \text{Cone}\left\langle D_0, D_1, D_3\right\rangle\,,
\end{equation}
where we have introduced the set of prime divisors $D_i$ that can be written in terms of the $J_i$ K\"ahler cone generators of $B_3$ as
\begin{equation}
\label{eq:EffcondivP1Fn-Ka}
   D_0=J_3 - sJ_0\coma D_1= J_0-nJ_1 \coma D_3=  J_1\,.
\end{equation}
These divisors can be obtained from the divisors in \eqref{eqn:effDivisorsBlF2} for $u=J_2=0$. Their respective volumes are
\begin{equation}
    \begin{split}
        \mathcal{V}_{D_0} & = \frac{1}{2}v^0(n v^0+2v^1)\coma \\
        \mathcal{V}_{D_1} & = v^1v^3\coma\\
        \mathcal{V}_{D_3} & = \frac{1}{2}v^3(2v^0+sv^3)\fstop
    \end{split}
\end{equation}
From \eqref{eq:EffcondivP1Fn}, we define the generators of the movable cone $\text{Mov}_1(B_3)=\overline{\text{Eff}}^1(B_3)^\vee$ as
\begin{equation}\label{ex1:Movablecone}
    C^0= J_0\cdot J_1\coma C^1= J_1\cdot J_3 \coma C^3=  J_0 \cdot J_3\coma
\end{equation}
whose volumes are
\begin{equation}
   \begin{split}
        \mathcal{V}_{C^0} & = v^3\coma \\
         \mathcal{V}_{C^1} & = v^0+ s v^3\coma\\
        \mathcal{V}_{C^3} & = n v^0+v^1+s nv^3\fstop
    \end{split}
    \end{equation}
Notice that we can write $J_0^2=nC^0$ and $J_3^2= sC^3$.

Let us now discuss the (quasi-)primitive EFT string limits and the curves in $B_3$ giving rise to these limits when wrapped by a D3-brane. 

Using Definition \ref{def:quasiprimitive}, we can consider the limits $v^a\rightarrow 0$ for some Mori cone generator volume $v^a$ and then use co-scalings in order to arrive at an EFT string limit in the sense of Definition \ref{def:EFTlimit}. For instance, for $v^0\sim \lambda \rightarrow \infty$ we can obtain a primitive EFT string limit by co-scaling $v^1\sim \lambda$ and $v^3\rightarrow \lambda^{-1}$. For this choice of scalings, $\mathcal{V}_{D_0}\rightarrow \infty$ while the volume of the other generators of $\text{Eff}^1(B_3)$ remain finite. In this case, a curve $C^0$ giving rise to this primitive EFT string limit is proportional to $J_0^2$ in accordance with Proposition \ref{prop:EFTstring}. Next, we can consider the limit $v^1\sim \lambda \rightarrow \infty$. In this case, there are two co-scalings possible. First, we can again take $v^0\sim \lambda$ and $v^3\rightarrow \lambda^{-1}$, corresponding to the same primitive EFT string limit as before. Another possible co-scaling is given by $v^0\sim \lambda^{-1}$ and yields a primitive EFT string limit in which only $\mathcal{V}_{D_1}$ becomes large. Its dual curve  $C^1=J_1\cdot J_3$ yields a $q=1$ string in agreement with Proposition~\ref{prop:EFTstring}. 

Finally, we can consider the limit $v^3\sim \lambda \rightarrow \infty$, which is sufficient in order to blow up $\mathcal{V}_{D_3}$ and realizes the last primitive EFT limit we were looking for. 
For $s\neq0$, we can summarize the curves giving rise to quasi-primitive EFT strings and the corresponding limits as follows:
\begin{equation}
\renewcommand{\arraystretch}{1.25}
\begin{tabular}{c|c|l|c}
    \text{Movable curve} & \text{$q$ factor} & \multicolumn{1}{c|}{Primitive EFT limit} & ${\cal V}_{D_i} \to \infty$\\
    \hline
   $C^0= J_0\cdot J_1$ & $q=0$ & $v^0,v^1\rightarrow \infty  \coma v^3 \rightarrow 0$ & $D_0$\\
    $C^1=J_1\cdot J_3$ & $q=1$ & $v^1 \rightarrow \infty \coma v^0 \rightarrow 0 \coma v^3 \simeq \text{const.}$ & $D_1$\\
    $C^3=J_0\cdot J_3$ & $q=2$ & $v^3 \rightarrow \infty \coma v^0\simeq \text{const.} \coma v^1 \simeq \text{const.}$ & $D_3$
\end{tabular}
\label{eq:EFTlimP1fibFn}
\end{equation}
The three EFT strings have $q=0,1,2$, respectively, and are precisely of the form required by Proposition \ref{prop:EFTstring} (recall $C^0=nJ_0^2$ and $C^3=sJ_3^2$). Thus, in this example, the number of (quasi-) primitive EFT strings agrees with the dimension of the scalar field space and all the generators of $\text{Mov}_1(B_3)$ give rise to a primitive EFT string. 

In Appendix \ref{app:examples}, we illustrate that some primitive EFT string limits may not be realizable in a given chamber of the K\"ahler cone. Put differently, 
there may be a generator of the effective cone which cannot acquire infinite volume in an
EFT string limit without any other effective cone generator becoming large. 
To realize the limit as an EFT string limit, one must instead pass to a different chamber in the K\"ahler cone by performing a flop transition. 

\section{Non-critical EFT strings and the Weak Gravity Conjecture}
\label{sec:NoncritEFTWGC}

We now build on the EFT string limits discussed in the previous section to investigate the weak coupling regimes for gauge theories on 7-branes in F-theory. Our goal is to test the asymptotic tower Weak Gravity Conjecture at weak coupling. 
Our analysis proceeds in three steps.

In the first step, we focus on weak coupling limits that can be obtained as {\it primitive} EFT string limits.\footnote{Since for each generator of $\text{Mov}_1(B_3)$ there exists a chamber of $\mathcal{K}_{\rm ext}$ in which the corresponding primitive EFT string limit can be attained, there exists a weak coupling limit of this kind for each gauge theory.} We therefore assume that a gauge theory with gauge group $G$ is realized on a divisor $\mathcal{S}$ containing a generator of $\text{Eff}^1 (B_3)$ which by itself is  homogeneously expandable (cf. Definition \ref{def:EFTlimit}) within a given chamber of $\mathcal{K}_{\rm ext}$. 
Using \cref{prop:primitiveEFT,prop:EFTstring}, we can then identify a curve $C$ on $B_3$ giving rise to the EFT string that becomes weakly coupled in the EFT string limit. By Proposition \ref{prop:EFTstring}, this ensures that $2m\equiv {\mathcal{S}}\cdot C>0$. As a result, the ${\cal N}=(0,2)$ supersymmetric worldsheet theory of the string has charged zero modes, as explained microscopically in \cite{Lawrie:2016axq}. Due to these charged zero modes, some string excitations are charged under the gauge theory, {\it provided} the --  in general non-critical -- EFT string has particle-like excitations. One might then contemplate that these charged states constitute a tower of super-extremal states required by the tower Weak Gravity Conjecture along the lines of \cite{Heidenreich:2021yda}.
For critical strings becoming light in 
Emergent String limits, this
had previously been shown to be the case  \cite{Lee:2018urn,Lee:2018spm,Lee:2019wij,Klaewer:2020lfg}.

As we will explain in Section \ref{ssec:Repulsiveforcecond}, in primitive EFT string limits corresponding to the weak coupling limit for gauge theory, the effective action resembles that of a perturbative, critical string. We can therefore evaluate the Repulsive Force Condition \cite{Palti:2017elp,Heidenreich:2019zkl} for possible string excitations explicitly. We will show that in general the excitations of the primitive EFT string cannot account for the states required by the tWGC due to a mismatch in the $\mathcal{O}(1)$ factors. The only exception are primitive EFT strings with $q=0$, 
which precisely correspond to the emergent strings in \cite{Lee:2018spm,Lee:2018urn,Lee:2019wij,Klaewer:2020lfg}.

In the second step, we analyze more general weak coupling limits realized as {\it quasi-primitive} EFT string limits in Section \ref{ssec:stringscalevsSC}. We will show that for such weak coupling limits, the tension of the relevant EFT string always lies at or above the species scale of a KK-tower becoming light in the limit, again unless $q=0$. 

This result is generalized, in step three (see Section \ref{ssec:nonEFTgaugetheory}), to 
the most general weak coupling limits for gauge theories on 7-branes which can not necessarily be obtained as quasi-primitive EFT string limits or even as EFT string limits in the first place.

Based on these observations, 
we argue, in Section \ref{ssec:4dWGC}, why it is consistent that we have not found any tower of marginally super-extremal states for the gauge theories associated to primitive EFT strings with $q> 0$ and why in general we do not even expect the tWGC to be realized by marginally super-extremal excitations of quasi-primitive EFT strings with $q>0$.

\subsection{Repulsive force condition}
\label{ssec:Repulsiveforcecond}

In this section, we will test the tower WGC for gauge theories on 7-branes in F-theory that become weakly coupled in a primitive EFT string limit in the sense of Definition \ref{def:EFTlimit}.

To this end, we will first show that for such a {primitive} EFT string limit the relevant part of the effective action has the general form
\begin{equation}
\label{eq:pertaction}
\begin{split}
     S_{4d} = &\frac{M_\text{\tiny Pl}^2}{2}\int \left(R\star \ID - \xi \frac{{\rm d}T_0\wedge \star {\rm d} \bar T_0}{(T_0+\bar T_0)^2}\right) - \frac{\kappa M_\text{\tiny Pl}^2}{8}\int \left(\re T_0 \;\text{tr}|F|^2 - i \im T_0\; \text{tr}(F\wedge F)\right)+\ldots\,. 
\end{split}
\end{equation}

Here the chiral field $T_0$ denotes the complexified volume of a generator $D_0$ of $\text{Eff}^1(B_3)$ with the property that $\re T_0\rightarrow \infty$ in the primitive EFT string limit. According to Definition \ref{def:EFTlimit}, $D_0$ is the only generator of
$\text{Eff}^1(B_3)$ which becomes large in the primitive EFT limit.
Furthermore, $F$ is the gauge field strength of any gauge theory becoming weakly coupled in the primitive EFT string limit.

To derive \eqref{eq:pertaction}, we first observe
that the gauge kinetic function for any gauge theory becoming weakly coupled in a primitive EFT string limit is asymptotically proportional to $T_0$. To see this, note that such a gauge theory  
necessarily originates from a 7-brane wrapping a divisor $\mathcal{S}$ that can be written as 
\begin{align}\label{eq:boldS}
    {\mathcal{S}} = \kappa D_0 + \ldots\,
\end{align}
for some numerical constant $\kappa$. In the primitive EFT string limit, where $\re T_0\gg 1$, the volume of $\mathcal{S}$ (and hence the inverse gauge coupling squared of the corresponding gauge theory) is, to leading order, given by $\kappa \re T_0$.

We now turn to the kinetic term for the complex scalar field $T_0$, which derives from the K\"ahler potential. 
In fact, we claim that, in the limit, the K\"ahler potential takes the form
\begin{equation}\label{Ksimplified}
    K=-\xi \log\left[T_0+\bar T_0\right] +\ldots\,, \qquad \xi = 1+q \,,
\end{equation}
from which the kinetic term in \eqref{eq:pertaction} follows.
Indeed, in our F-theory setup $K$ is given by \eqref{eq:K}, and we 
claim that in the primitive EFT string limit the volume $\mathcal{V}_{B_3}$ factorizes in such a way that $K$ acquires the form  \eqref{Ksimplified}. 
The reason for this factorization is that in the primitive EFT string limit, $\re T_0\rightarrow \infty$ while the volumes of all other divisors remain finite, see \eqref{eq:singleEFTlimit}. More precisely, for a primitive EFT string, we can show the following 
\begin{proposition}\label{prop:volumefactorization}
Given a primitive EFT string obtained in F-theory from a D3-brane wrapped on a curve $C^0$ in $B_3$ charged under an axion ${\rm Im}\,T_0$. Then in the corresponding EFT string limit, the volume of $B_3$ behaves as
\begin{equation}\label{eq:volumelimit}
    \mathcal{V}_{B_3}^2 \rightarrow ({\rm Re}\,T_0)^{1+q} \;P_{2-q}({\rm Re}\,T_{i\neq0})\,,
\end{equation}
where ${\rm Re}\,T_0$ is the saxionic partner of the axion ${\rm Im}\,T_0$, $q$ is the topological quantity associated to $C^0$ defined in \eqref{eq:q}, and $P_{2-q}$ is a polynomial of degree $2-q$ in the remaining saxions. This implies that the parameter $\xi$ appearing in \eqref{eq:pertaction} is given by 
\begin{equation}\label{eq:dqrel}
    \xi=1+q\,.
\end{equation}
\end{proposition}
The proof of Proposition \ref{prop:volumefactorization} is technical and is provided in Appendix \ref{app:proofs}. With this key result, the effective action \eqref{eq:pertaction} indeed follows.

Notice that \eqref{eq:pertaction} resembles the effective action of the weakly coupled heterotic string in 4d once we identify $T_0$ with the heterotic axio-dilaton and $F$ with the field strength of the perturbative heterotic gauge group. In particular, this form of the action suggests that the axion decay constant (and hence the 2-form gauge coupling for the 2-form under which the string is charged) and the 1-form gauge couplings are all determined by a single saxion. This saxion takes over the r\^ole of the dilaton for the primitive EFT string. 
 In the special case $q=0$, i.e. $\xi =1$, the EFT string {\it is} the critical heterotic string, and the above observations are part of the content of the Emergent String Conjecture \cite{Lee:2019wij}.  The excitations of the emergent heterotic string indeed contain a marginally super-extremal tower, satisfying the tower WGC \cite{Lee:2018spm,Lee:2018urn,Lee:2019tst,Klaewer:2020lfg}.

Given this state of affairs, it is natural
to speculate that the WGC in primitive weak coupling limits for general values of $q$ is realized similarly.
To show this, one would first have to {\it postulate} that the primitive EFT string with $q>0$ can be treated as having a spectrum of excitations similar to a critical string. Of course, this is a strong assumption, but for the primitive EFT strings considered in this section one might motivate this by the observed
similarity between
 the effective action \eqref{eq:pertaction} and the action of a weakly coupled heterotic string, apart from the numerical factor $\xi$.

From an EFT point of view, the relation between weak coupling limits of effective 4d strings and the tWGC for the gauge symmetry has been analyzed in detail in \cite{Heidenreich:2021yda}. 
As a consequence of anomaly inflow on the string, the EFT string charged under $\theta = \im T$ needs to have excitations charged under the gauge group. It was then shown in \cite{Heidenreich:2021yda} that as a consequence of the WGC for the axion to which the string couples, these charged excitations of the string must satisfy the relation
\begin{equation}
    M_{k} \lesssim g_\text{\tiny YM} M_\text{\tiny Pl}\,,
\end{equation}
which is to be interpreted as a parametric relation up to $\mathcal{O}(1)$ coefficients. Here, $M_{k}$ is the mass of a charged string excitation at level $k$ (assuming it exists), $g_\text{\tiny YM}$ is the gauge coupling and $M_\text{\tiny Pl}$ the $4$d Planck mass. Notice that this parametric scaling is a necessary condition for the WGC for the gauge theory to be fulfilled by the tower of charged string excitations. However, to decide whether the WGC is indeed satisfied by these states, the $\mathcal{O}(1)$ coefficients must be determined.

We now turn to this question, assuming for simplicity a gauge group of the form $G= \U(1)$.\footnote{In this case, the gauge divisor $\mathcal{S}$ is to be identified with the height pairing of a rational section of the elliptic fibration (see e.g. \cite{Weigand:2018rez,Cvetic:2018bni} for reviews), but this technicality plays 
no essential r\^ole in the forthcoming discussion.}
To be precise, the tWGC is fulfilled if the states in the tower of charged excitations satisfy the relation 
\begin{equation}\label{eq:WGC}
    \frac{g_\text{\tiny YM}^2 q_k^2}{M_{k}^2} \geq \frac{1}{M_\text{\tiny Pl}^2} \left[\left.\frac{d-3}{d-2}\right|_{d=4} +\frac{1}{4} \frac{M_\text{\tiny Pl}^4}{M_k^4} g^{rs} \partial_r \left(\frac{M_k^2}{M_\text{\tiny Pl}^2} \right)\partial_s \left(\frac{M_k^2}{M_\text{\tiny Pl}^2} \right) \right]\,.
\end{equation}
Here, $q_k$ is the charge of the $k$-th excitation, and we have included the effect of scalar fields $\phi^r$ in the second term with $g^{rs}$ the inverse metric on the scalar field space. This relation arises from the Repulsive Force Conjecture \cite{Palti:2017elp,Heidenreich:2019zkl} that requires that the super-extremal state does not form bound states with itself, since its Coulomb repulsion is stronger than the attractive gravitational and Yukawa forces. In its form \eqref{eq:WGC}, it is assumed that we have a perturbative description for all three forces involved. In particular, the first term on the RHS arises from a Newton-like gravitational potential and therefore requires the scale at which we evaluate \eqref{eq:WGC} to correspond to weakly coupled gravity. It can be shown that in weak coupling limits, the equality in the condition \eqref{eq:WGC} corresponds to the extremality bound of a dilatonic Reissner-Nordstrom black-hole (cf. \cite{Lee:2018spm} for the corresponding discussion in 6d). 

To test this relation for the putative excitations of a primitive EFT string, we assume, as before, that the primitive EFT string is obtained from a D3-brane wrapping a curve $C^0$ in $B_3$. 
The gauge divisor $\mathcal{S}$ is characterized by 
 $m\equiv \frac{1}{2} \,C^0\cdot {\mathcal{S}}>0$ (cf. Proposition \ref{prop:EFTstring}). The gauge coupling $g_\text{\tiny YM}^2$ for the gauge group is given by 
\begin{align}
    \frac{2\pi}{g_\text{\tiny YM}^2} =  \mathcal{V}_{\bf S}\,,
    \label{eq:gengYMV}
\end{align}
which in the primitive EFT string limit reduces to
\begin{align}
    \mathcal{V}_{\bf S} = \kappa \,\re T_0 + \ldots\,. 
    \label{eq:VST0}
\end{align}
From the modular properties of the elliptic genus of the effective string, one can infer the existence of states at mass level $n_k$ and charge \cite{Lee:2019tst,Lee:2020gvu,Lee:2020blx,Klaewer:2020lfg}
\begin{align}\label{qk}
    q_k^2 = 4m n_k\,.  
\end{align}
As explained in \cite{Lee:2019tst,Lee:2020gvu,Lee:2020blx,Klaewer:2020lfg}, these are the candidates for the super-extremal states, and we henceforth focus on them.
As stressed before, via the action \eqref{eq:pertaction}, the EFT string resembles a perturbative string in the corresponding weak coupling limit. One may hence {\it assume} that the quantization of the string proceeds similarly to the quantization of critical strings, such that the mass of an excitation at level $k$ is given by 
\begin{align}\label{eq:Mk}
    M_k^2 = 8 \pi T_\text{\tiny EFT} (n_k-a)\,,
\end{align}
where $a$ is the vacuum energy of the string that can be calculated as $a=\frac12 \bar K\cdot C^0$ \cite{Lee:2019tst}. Thus, we can rewrite the LHS of \eqref{eq:WGC}, in the limit of large $n_k$, for the candidate super-extremal states with the property \eqref{qk} as
\begin{align}
    \frac{q_k^2 g_\text{\tiny YM}^2}{M_k^2} = \frac{m }{T_\text{\tiny EFT} \, \kappa \, \re T_0\, }\,. 
\end{align}
On the other hand, the tension of the EFT string in the EFT string limit is given by the linear multiplet 
\begin{align}\label{eq:TL0}
    \frac{T_\text{\tiny EFT}}{M_\text{\tiny Pl}^2}=e_0 L^0 = - \frac{e_0}{2}\frac{\partial K}{\partial \re T_0}=\frac{e_0}{2}\frac{\left(1+q\right)}{\re T_0}\,,
\end{align}
where we used Proposition \ref{prop:volumefactorization} to infer the dependence of $K$ on $\re T_0$ in the primitive EFT string limit. Here, $e_0$ is the charge of the string under the 2-form dual to $\im T_0$. Since the curve with $e_0=1$ corresponds to the string associated with the generator of the movable cone dual to $D_0$, we find 
\begin{align}
    2m = \kappa e_0\,,
\end{align}
leading to 
\begin{align}
   \frac{q_k^2 g_\text{\tiny YM}^2}{M_k^2/M_\text{\tiny Pl}^2} = \frac{1}{1+q}  \,.
\end{align}

Let us now turn to the RHS of \eqref{eq:WGC}. Whereas the first term always yields a contribution of $\frac{1}{2}$ in $4$d, the second term is determined by how the mass scale, and hence the EFT string tension, depends on the scalar fields in the theory. By \eqref{eq:Mk} and \eqref{eq:TL0} the mass of the excitations only depends on a single scalar field $L^0$. Hence 
\begin{equation}
    g^{rs}\partial_s \left(\frac{M_k^2}{M_\text{\tiny Pl}^2}\right)\partial_r \left(\frac{M_k^2}{M_\text{\tiny Pl}^2}\right)=\left(8\pi e_0 (n_k-a)\right)^2 g^{00}\,.
\end{equation}
We are thus left with evaluating the metric component $g^{00}$:
\begin{align}
    g^{00} = \frac{1}{2}\frac{\partial^2 K}{\partial (\re T_0)^2} = -\frac{\partial L^0}{\partial (\re T_0)}= \frac{1}{2}\frac{1+q}{(\re T_0)^2}\,. 
\end{align}
Putting things together, we obtain 
\begin{align}
    \frac{1}{4} \frac{M_\text{\tiny Pl}^4}{M_k^4} g^{rs} \partial_r \left(\frac{M_k^2}{M_\text{\tiny Pl}^2} \right)\partial_s \left(\frac{M_k^2}{M_\text{\tiny Pl}^2}\right)=\frac14 \left(\frac{2\re T_0}{(1+q)}\right)^2\frac{(1+q)}{2(\re T_0)^2}=\frac{1}{2(1+q)}\,. 
\end{align} 
For a string with given $q$, the repulsive force condition then requires 
\begin{align}\label{eq:WGCfinal}
    \frac{g_\text{\tiny YM}^2 q_k^2}{M_{k}^2/M_\text{\tiny Pl}^2} = \frac{1}{1+q} \stackrel{!}{\geq} \frac{1+\frac{q}{2}}{1+q}= \left[\left.\frac{d-3}{d-2}\right|_{d=4} +\frac{1}{4} \frac{M_\text{\tiny Pl}^4}{M_k^4} g^{rs} \partial_r \left(\frac{M_k^2}{M_\text{\tiny Pl}^2} \right)\partial_s \left(\frac{M_k^2}{M_\text{\tiny Pl}^2} \right)\right]\,. 
\end{align}
We notice that the inequality in the above expression is only satisfied for $-1<q\leq 0$. Thus, since for primitive EFT strings $q$ is restricted to $0$, $1$ or $2$, we expect such strings to lead to a tower of states satisfying the repulsive force condition \eqref{eq:WGC} only for $q=0$. The result for $q=0$, of course, matches with the computation in \cite{Klaewer:2020lfg} for the (heterotic) emergent string limit: In this case, one finds a marginally super-extremal tower of states, in the sense that at higher and higher excitation level the super-extremal states asymptotically become extremal.

Let us stress that still, as noticed in \cite{Heidenreich:2021yda}, the repulsive force condition is parametrically satisfied in that any dependence on $\re T_0$ drops out of the relation \eqref{eq:WGC}. This follows from the dependence of $K$ on $\re T_0$ as given by Proposition \ref{prop:volumefactorization} and the relation between the string tension and $\re T_0$, which in turn is a consequence of the BPS properties of the string. Thus, indeed, the BPS property of the string ensures that \eqref{eq:WGC} is parametrically satisfied. However, the $\mathcal{O}(1)$ coefficients do not match. These $\mathcal{O}(1)$ coefficients are sensitive to the actual particle excitations of the string via \eqref{qk} and \eqref{eq:Mk}. 

One might contemplate whether the relation \eqref{eq:Mk} receives corrections for non-critical strings with $q\geq 1$ which change the quantization condition to
\begin{align} \label{rescaledmasses}
    M_k^2 = 8\pi\,\mathfrak{n}(q)\,  T_\text{\tiny EFT} (n_k-a)\,,
\end{align}
for some $q$-dependent factor $\mathfrak{n}$. Following the same steps as before, one finds that $\mathfrak{n}$ would have to be given by 
\begin{align}
    \mathfrak{n}(q)= \frac{2}{2+q}\,
\end{align}
in order to obtain a marginally asymptotic tower, as for the emergent string with $q=0$. 
However, we actually do not expect that the tWGC is realized in this manner:
First, the simple rescaling \eqref{rescaledmasses} is not only completely ad hoc, but more importantly it does no longer work for more general limits beyond the primitive strings considered in this section.

Furthermore,  
as we argue in the next section, we do not expect the (quasi-)primitive EFT strings with $q\geq 1$ to have particle-like excitations in 4d in the first place. In fact, we will conclude from this, in Section \ref{ssec:4dWGC}, that the limits with $q\geq 1$ do not even motivate an asymptotic WGC in the usual sense.

\subsection{String scale vs. species scale for (quasi-)primitive EFT string limits}
\label{ssec:stringscalevsSC}

So far, we have treated the primitive EFT strings on similar grounds to a critical string in four dimensions and assumed that it has particle-like excitations. However, for this to be the case, we need to ensure that in the EFT string limit we can still consider the D3-brane string as an effective string in four dimensions. In this section, we will assess the validity of this assumption section by analyzing the relation between the KK-scale of the F-theory compactification and the tension of the primitive EFT string.
The consequences of this analysis for the WGC will then be discussed in Section \ref{ssec:4dWGC}.

We can in fact widen the scope of our analysis and focus on quasi-primitive EFT string limits in the sense of Definition \ref{def:quasiprimitive}, which include the primitive string limits studied in the previous section as special cases.\footnote{In the next section, the restriction to EFT string limits will be dropped altogether.}
As in \cite{Lanza:2021udy}, the relation between the KK-scale and the EFT string tensions is governed by the so-called scaling weight $w$ defined as 
\begin{equation}\label{eq:scalingweight}
    m_*^2 \sim  A\, M_\text{\tiny Pl}^2\left(\frac{T_\text{\tiny EFT}}{M_\text{\tiny Pl}^2}\right)^w\,,
\end{equation}
where $m_*$ is the cut-off scale of the EFT, i.e., the mass scale of the lightest tower, and $A$ is some constant depending on the free parameters of the string flow. In our discussion of EFT string limits in the K\"ahler field space of F-theory, the lightest scale is always (at least parametrically) the KK-scale. For our three cases $q=0,1,2$, we can calculate the respective value of $w$ and find
\begin{equation}\label{qwidentification}
    q=(0,1,2) \quad \longleftrightarrow \quad w=(1,2,2)\,. 
\end{equation}
Thus, unless $q=0$, the tension of the string is always above the KK-scale. These limits therefore correspond to decompactification limits, in agreement with the Emergent String Conjecture. However, one might still be tempted to view the strings as effectively four-dimensional objects as long as their tension remains below the species scale $\Lambda_\text{\tiny sp,KK}$ associated to the tower of KK-modes. 

We therefore should evaluate the species scale associated to the KK-tower with mass scale $M_\text{\tiny KK}$. When decompactifying $n$ dimensions, the number $N$ of KK states with mass $m^2\leq k^2 M_\text{\tiny KK}^2$ for some $k\in \mathbb{N}$ grows like 
\begin{equation}\label{Ngrowth}
    N\sim k^{n}\,. 
\end{equation}
On the other hand, the species scale in four dimensions is defined as \cite{Dvali:2007hz}
\begin{equation}
    \Lambda_\text{\tiny sp}^2 = \frac{M_\text{\tiny Pl}^2}{N_\text{\tiny sp}}\,,
\end{equation}
where $N_\text{\tiny sp}$ is the number of species with masses $\Lambda_\text{\tiny sp}$. Using \eqref{Ngrowth} we find 
\begin{equation}\label{eq:Lambdasp}
    \Lambda_\text{\tiny sp,KK}^2 = k_{\rm max}^2 M_\text{\tiny KK}^2 \stackrel{!}{=} \frac{M_\text{\tiny Pl}^2}{k_{\rm max}^n}\,\qquad \Rightarrow \qquad \frac{\Lambda_\text{\tiny sp,KK}^2}{M_\text{\tiny Pl}^2} = \left(\frac{M_\text{\tiny KK}^2}{M_\text{\tiny Pl}^2}\right)^{\frac{n}{2+n}}\,. 
\end{equation}
Here $k_{\rm max}$ is the maximal excitation level of the states with masses below the species scale. As we show in the sequel, in terms of the type IIB string scale $M_\text{\tiny IIB}$ the scale $\Lambda_\text{\tiny sp,KK}$ for the different values of $q$ is given as follows:
\begin{center}
\renewcommand{\arraystretch}{1.5}
    \begin{tabular}{c|c|c|c}
    $q$ & $0$& $1$ &$2$ \\ \hhline{=|=|=|=} 
    $\dfrac{\Lambda_\text{\tiny sp,KK}^2}{M_\text{\tiny Pl}^2}$ & $\sim \left(\dfrac{M_\text{\tiny IIB}^2}{M_\text{\tiny Pl}^2}\right)^{4/3}$ &$\sim \dfrac{M_\text{\tiny IIB}^2}{M_\text{\tiny Pl}^2}$&$\sim \dfrac{M_\text{\tiny IIB}^2}{M_\text{\tiny Pl}^2}$ \\
    \end{tabular}
\end{center}

\vspace{2mm}

Thus, for limits of type $q=1,2$ the KK species scale coincides with the Type IIB string scale $M_\text{\tiny IIB}$, which turns out to be the higher dimensional Planck mass for the decompactification limits under consideration. However, in order for $\Lambda_\text{\tiny sp,KK}$ to correspond to the actual species scale, we need to ensure that the tension of the EFT string giving rise to the asymptotic limit does not lie between $M_\text{\tiny KK}$ and $\Lambda_\text{\tiny sp,KK}$. If this was the case, the actual species scale would be set by $T_\text{\tiny EFT}$ \cite{Dvali:2009ks,Dvali:2010vm} and we could still consider the EFT string limit as effectively four-dimensional.\footnote{Here we assume for the time being, as before, that it makes sense to speak of a tower of EFT string excitations in four dimensions. The following considerations serve as a consistency check of this assumption, and in fact will show that for $q >0$ this assumption is not justified.} As we show below, the relation between $M_\text{\tiny IIB}$ and $T_\text{\tiny EFT}$ can be conveniently written as
\begin{equation}\label{IIBTensionrel}
    M_\text{\tiny IIB}^2 \sim M_\text{\tiny Pl}^2 \left(\frac{T_\text{\tiny EFT}}{M_\text{\tiny Pl}^2}\right)^{\frac{q+1}{2}}\,.
\end{equation}
Therefore for $q>0$ the EFT-string tension sits at or above $\Lambda_\text{\tiny sp,KK}$ and hence the EFT string should effectively be viewed as an object in a higher-dimensional theory. In particular, in this case $\Lambda_\text{\tiny sp,KK}$ is the \emph{actual} species scale since any other tower of states sits at or above $\Lambda_\text{\tiny sp,KK}$ and should therefore not be taken into account when calculating the species scale. On the other hand, for $q=0$ we have $T_\text{\tiny EFT}\precsim \Lambda^2_\text{\tiny sp,KK}$ such that here the effective asymptotes to a bona fide four-dimensional theory in the EFT string limit. Notice that for $q=0$, the emergent critical string, the KK-species scale does not coincide with $M_\text{\tiny IIB}$ as would be expected for a decompactification limit. This discrepancy in fact already signals that the $q=0$ limit cannot be a decompactification and that consistency requires the presence of an additional tower of states. Thus, for $q=0$ one could have inferred the existence of the emergent string merely based on the scaling of the species scale associated to the KK-tower.

We will now show how to obtain the relation \eqref{IIBTensionrel} between $\Lambda_\text{\tiny sp,KK}$, $M_\text{\tiny IIB}$ and $T_\text{\tiny EFT}$ for the different values of $q$:
\begin{enumerate}[align=left,wide, labelindent=0pt]
\item[$\bf{q = 0}$:]

Recall that in a $q=0$ EFT string limit,
the F-theory base $B_3$ admits a rational or genus-one fibration (see Figure \ref{fig:q=0}) and in the limit the volume of the base of this fibration is scaled up homogeneously as
$\lambda^2 \to \infty$, while the volume of the fiber $C$ shrinks as $\lambda^{-1}$ (each in Type IIB string units). 
Here $v^0\sim \lambda \rightarrow \infty$ is the volume of a curve ${\cal C}^0$ in the base of the fibration which becomes large.
The $q=0$ string arises from a D3-brane wrapped around the shrinking fiber. Altogether, the volume of $B_3$ scales up as $\mathcal{V}_{B_3}\sim \lambda$ (cf. \cite{Klaewer:2020lfg} and the proof of Proposition \ref{prop:q0} in Appendix \ref{app:proofs}).
It follows that the KK scale is 
given by 
\begin{align}\label{KKscaleq0}
     \frac{M_\text{\tiny KK}^2}{M_\text{\tiny Pl}^2} \sim \frac{1}{\lambda^2}\,   \quad \text{for}\, \,   q=0 \,,
\end{align}
where we used the standard relation between type IIB scale and $M_\text{\tiny Pl}$, 
\begin{equation}
    \frac{M_\text{\tiny Pl}^2}{M_\text{\tiny IIB}^2} = 4\pi \mathcal{V}_{B_3}\,.
\end{equation}
Since the base of the fibration 
expands while the fiber shrinks,
 the $\lambda\rightarrow \infty$ limit {\it na\"{\i}vely} resembles a decompactification to 8d, i.e., if we take into account the KK modes from the expanding base in the computation of the KK species scales as in \eqref{eq:Lambdasp}, we must set $n=4$. The species scale associated to the KK tower is then given by 
\begin{align}
    \frac{\Lambda_\text{\tiny sp,KK}^2}{M_\text{\tiny Pl}^2} \sim \left(\frac{1}{\lambda^2}\right)^{2/3} \sim \frac{1}{\lambda^{4/3}} \precsim \frac{1}{\lambda}\sim \frac{M_\text{\tiny IIB}^2}{M_\text{\tiny Pl}^2} \quad \text{for}\, \,   q=0 \,.
\end{align}

At the same time, if the theory actually decompactified to an effective theory in eight dimensions, $\Lambda_\text{\tiny sp,KK}$ would have to coincide with the higher-dimensional Planck scale, which would, up to order one factors, be set by the ten-dimensional string scale
${M_\text{\tiny Pl}}$.\footnote{In the putative eight-dimensional theory, no residual scaling limit is taken; the claim then follows from the usual relation between the Planck scale in ten dimensions and the Planck scale after decompactification.}
The parametric discrepancy between these two scales thus indicates that the $q=0$ limit cannot correspond to a bona fide decompactification limit, but requires the excitations of the $q=0$ EFT string for consistency. This is of course in precise agreement with the Emergent String Conjecture, according to which this type of limit is an effectively four-dimensional weak coupling limit, with the r\^ole of the new fundamental string played by the $q=0$ EFT string.

By Proposition \ref{prop:q0}, all quasi-primitive EFT string limits with $q=0$ are in fact primitive, and we can therefore use Proposition \ref{prop:volumefactorization} to find 
    \begin{equation}
   \mathcal{V}_{B_3}\sim (\re T_0)^{\frac{1}{2}} \quad \text{for}\, \,   q=0 \,.
    \end{equation}
    On the other hand, since the tension of the $q=0$ EFT string is given by 
    \begin{equation}
        \frac{T_\text{\tiny EFT}}{M_\text{\tiny Pl}^2}\sim (\re T_0)^{-1}\,\quad \text{for}\, \,   q=0 \,,
    \end{equation}
    we arrive at
    \begin{equation}\label{IIBTensionrel-q0}
        M_\text{\tiny IIB}^2 \sim M_\text{\tiny Pl}^2 \left(\frac{T_\text{\tiny EFT}}{M_\text{\tiny Pl}^2}\right)^{\frac{1}{2}}\quad \text{for}\, \,   q=0 \,.
    \end{equation}
    Therefore, as expected, the EFT string scale is below the species scale for the KK tower, as shown also in Figure \ref{fig:q=0sp}.

\item[$\bf{q = 1}$:] For $q=1$, we know, from the discussion at the end of Section \ref{ssec:Weakcouplinglimits}, that for such a curve to exist, $B_3$ needs to admit a surface fibration (cf. Figure \ref{fig:q=1}). In this case, the quasi-primitive EFT string limit for the D3-brane on the curve $C$ with $q=1$ is given by the limit where the base of this fibration blows up. Let us denote the volume of the base $\mathbb P^1$ by $v^0\sim \lambda \rightarrow \infty$. From the proof of Proposition \ref{prop:EFTstring} we know that the total volume scales as
\begin{equation} \label{volq1}
    \mathcal{V}_{B_3}\sim \lambda \quad \text{for}\, \,   q=1 \,,
    \end{equation}
because the surface fiber does not scale in the limit.
All this points to a decompactification limit to six dimensions, corresponding to the value $n=2$ in \eqref{eq:Lambdasp}. The KK scale is given by  
 \begin{align}\label{KKscaleq1}
     \frac{M_\text{\tiny KK}^2}{M_\text{\tiny Pl}^2} \sim \frac{1}{\lambda^2}\quad \text{for}\, \,   q=1 \,,
 \end{align}
leading to the species scale
\begin{align} \label{LspKKq1}
    \frac{\Lambda_\text{\tiny sp,KK}^2}{M_\text{\tiny Pl}^2} \sim \left(\frac{1}{\lambda^2}\right)^{1/2} \sim \frac{M_\text{\tiny IIB}^2}{M_\text{\tiny Pl}^2} \quad \text{for}\, \,   q=1 \,.
\end{align}
 The result \eqref{LspKKq1} confirms that indeed the $q=1$ limit corresponds to a decompactification limit to 6d, and we arrive at a six-dimensional theory for which no particular scaling limit is taken. If we do not take any scaling limit, the 6d Planck scale is indeed just set by $M_\text{\tiny IIB}$ (cf. \cite{Lee:2019tst} for a similar discussion of decompactification limits). 

To find the relation between $M_\text{\tiny IIB}$ and $T_\text{\tiny EFT}$ we proceed as follows: Let $\mathcal{C}^0$ be the curve associated to the $q=1$ quasi-primitive EFT string limit, such that $v^0\sim \lambda\rightarrow \infty$ in the quasi-primitive EFT string limit. Then from Proposition \ref{prop:EFTstring}, we have \eqref{volq1}.
    On the other hand, the same proposition tells us that the EFT string curve is given by $J_0\cdot J_1$ for some K\"ahler cone generator $J_1$. Since the volume of all divisors intersecting $J_0\cdot J_1$  blow up in the EFT string, but no other divisors, the limit cannot involve any co-scaling $v^1 \to \infty$.\footnote{To see this, we notice that there needs to be a divisor for which the volume contains a term $(v^1)^2$ since $J_1^2=0$. If this term would only be contained in the divisors intersecting $J_0\cdot J_1$, the EFT string limit could equally be reached just by sending $v^1\rightarrow \infty$ with EFT string given by a D3-brane on $\alpha J_1^2$. This would, however, correspond to a $q=2$ string.}  On the other hand, since $J_0\cdot J_1^2\neq 0$ we know that 
    \begin{align}
        \mathcal{V}_{J_0\cdot J_1} \sim v^1 + \ldots\,, 
    \end{align}
    where the dots stand for possibly sub-leading terms. Therefore, we find 
    \begin{align}
        M_\text{\tiny IIB}^2 \sim M_\text{\tiny Pl}^2 \left(\frac{T_\text{\tiny EFT}}{M_\text{\tiny Pl}^2}\right)\,\quad \text{for}\, \,   q=1 \,,
    \end{align}
    such that $T_\text{\tiny EFT}$ is of the order of $M_\text{\tiny IIB}$. The scalings are shown in Figure \ref{fig:q=1sp}.

\item[$\bf{q = 2}$:]  By Proposition \ref{prop:EFTstring}, a quasi-primitive EFT string limit with $q=2$ corresponds to the limit $v^0\sim \lambda \rightarrow \infty$ for a K\"ahler cone generator $J_0$ with $J_0^3\neq0$. Hence
the base $B_3$ blows-up homogeneously
and
\begin{align}
         \mathcal{V}_{B_3}\sim \lambda^3\,\quad \text{for}\, \,   q=2 \,.
    \end{align}
We thus encounter a decompactification to 10d, corresponding to a value of $n=6$ in \eqref{eq:Lambdasp}. The KK-scale is given by 
\begin{align}
     \frac{M_\text{\tiny KK}^2}{M_\text{\tiny Pl}^2} \sim \frac{1}{\lambda^4}\,\quad \text{for}\, \,   q=2 \,,
\end{align}
such that via \eqref{eq:Lambdasp} the species scale follows as
\begin{align}
     \frac{\Lambda_\text{\tiny sp,KK}^2}{M_\text{\tiny Pl}^2} \sim \left(\frac{1}{\lambda^4}\right)^{3/4} \sim \frac{M_\text{\tiny IIB}^2}{M_\text{\tiny Pl}^2}\quad \text{for}\, \,   q=2 \,.
\end{align}
Thus, the species scale coincides with $M_\text{\tiny IIB}$, the ten-dimensional Planck scale (recall that the axio-dilaton is not scaled).
On the other hand, to find the relation between $M_\text{\tiny IIB}$ and $T_\text{\tiny EFT}$ we can use
that, by Proposition \ref{prop:EFTstring}, the volume of the curve giving rise to the quasi-primitive EFT string scales like 
    \begin{align}
        \mathcal{V}_{\alpha J_0^2} \sim v^0 +\ldots \sim \lambda\,. 
    \end{align}
    Altogether, this leads to
    \begin{align}
         M_\text{\tiny IIB}^2 \sim M_\text{\tiny Pl}^2 \left(\frac{T_\text{\tiny EFT}}{M_\text{\tiny Pl}^2}\right)^{3/2}\,\quad \text{for}\, \,   q=2 \,,
    \end{align}
    such that $T_\text{\tiny EFT}$ is parametrically above $M_\text{\tiny IIB}$. The scalings are schematically shown in Figure \ref{fig:q=2sp}.
\end{enumerate}

\subsubsection*{Consequences for the nature of the EFT strings and the weak coupling limits}

For $q\geq 1$, the scale set by the tension of the EFT string is bounded from below by the higher dimensional Planck scale. Therefore, these strings do not have excitations that can be treated as particle-like excitations in a weakly-coupled theory of gravity. The two cases $q=1$ and $q=2$ are nonetheless considerably different: In the quasi-primitive EFT string limit for a $q=1$ string, the theory decompactifies to 6d since the volume of the  $\mathbb{P}^1$ base of a surface fibration $\Sigma\rightarrow \mathbb{P}^1$ is scaled up. On the other hand, the divisor $\mathcal{S}$ on which the weakly coupled gauge theory is realized must contain this $\mathbb{P}^1$. Therefore, the EFT string limit effectively gives rise to a six-dimensional gauge theory coupled to gravity. However, this six-dimensional gauge theory is, generically, not weakly coupled since after decompactification no additional limit is taken in the six-dimensional field space, at least not in the most general type of limits.\footnote{At the end of the next section, we will discuss situations in which an additional limit leads to a weak coupling regime in six dimensions.} On the other hand, the quasi-primitive EFT string is obtained by wrapping a D3-brane on a movable curve $C^0\subset \Sigma$. Hence, the resulting 6d theory still contains an effective string which, just as the gauge theory, is not weakly coupled. In particular, the tension of this string can never drop below the six-dimensional Planck scale since $C^0$ satisfies 
\begin{align} \label{C06d}
    C^0\cdot_{\Sigma} C^0>0\,. 
\end{align}
In the resulting six-dimensional theory, this string is thus a supergravity string in the language of \cite{Kim:2019vuc} (see also \cite{Long:2021jlv}) whose tension is bounded by the six-dimensional Planck scale. Notice that this is consistent with the estimate \eqref{IIBTensionrel} that the tension of the $q=1$ EFT string is always of the order of $M_\text{\tiny IIB}$.

The situation for the $q=2$ quasi-primitive EFT string limit is different: Here the theory decompactifies all the way to ten dimensions. The weakly coupled gauge theory now flows to an eight-dimensional defect theory within the full ten-dimensional gravitational bulk theory, such that effectively the gauge and the gravity sector completely decouple. In contrast to the $q=1$ EFT string, which remains an effective string in six dimensions, the $q=2$ string ceases to be an effective string in the full ten-dimensional theory since the limit resolves the internal directions of the D3-brane. This is consistent with the relation \eqref{IIBTensionrel} telling us that for $q=2$ the tension of the EFT string is always parametrically above $M_\text{\tiny IIB}$ in the weak coupling limit. Whereas the $q=1$ string can be thought of as a version of a six-dimensional supergravity string, the $q=2$ string does not have a higher-dimensional analogue. As an effective string, it thus only exists in four dimensions.  

The nature of the $q=0$ strings is notably different: Here the limit results in a genuinely four-dimensional gauge theory weakly coupled to gravity, which also satisfies the tWGC through the excitations of the $q=0$ EFT string (see also \cite{Lee:2019tst,Klaewer:2020lfg}). Notice that the $q=0$ string also has a higher-dimensional analogue since it already exists in eight dimensions, i.e., in F-theory compactified on an elliptic K3 surface. Moreover, in 8d the $q=0$ string could be viewed as a genuine supergravity string since its tension cannot drop below the 8d Planck scale (both, the tension of the D3-brane on the base of the K3 and the 8d Planck scale are proportional to the volume of the base of the K3). We can summarize our findings as follows: 
\begin{center}
    \begin{tabular}{c|c|c}
    & $\text{dim}_\text{eff}(\text{gauge theory)}$ & Supergravity string  \\
    \multirow{-2}{*}{4d effective string} &for $g_\text{\tiny YM}\rightarrow 0$ & exists in:\\
    \hhline{=|=|=}
    $q=0$ & $d=4$ & $d\leq 8$\\ \hline 
    $q=1$& $d=6$ & $d\leq 6$ \\ \hline 
    $q=2$& $d=8$ & $d=4$ \\ 
    \end{tabular}
\end{center}
\vspace{2mm}

Let us stress that while the $q=0$ string already exists in eight dimensions, its EFT limit does not lead to a decompactification for the associated gauge theories; by contrast the $q=2$ string only exists as a string in four dimensions, but its EFT string limit corresponds to a decompactification to eight dimensions for the gauge sector within a ten-dimensional gravitational bulk theory.

\subsubsection*{Comparison to supergravity strings in 5d}

It is instructive to compare our classification of the four-dimensional supergravity strings in terms of the quantity $q$ to the five-dimensional supergravity strings discussed in \cite{Katz:2020ewz}. Therefore, consider the Coulomb branch of a 5d supergravity theory with effective action 
\begin{align}
    S_{5d} = \frac{M_\text{\tiny Pl}^3}{2} \int \left( R\star \ID  - G_{IJ} d\phi^I \wedge \star d\phi^J - G_{IJ} F^I\wedge \star F^J -\frac16 C_{IJK} A^I \wedge F^J \wedge F^K\right)\,.
\end{align}
Here $F^I$ are the field strengths of the $\U(1)$ gauge fields $A^I$ in the 5d vector multiplets, $\phi^I$ the corresponding scalars, $G_{IJ}$ the metric on the scalar fields space and $C_{IJK}$ the integer coefficient of the cubic Chern-Simons term. The supergravity strings now arise as monopole strings for the gauge fields $A^I$ carrying charge 
\begin{align}
    p^I = \frac{1}{2\pi} \int_{S^2} A^I\,,
\end{align}
where $S^2$ is a sphere encircling the string. According to \cite{Katz:2020ewz}, a string is a supergravity string if all supersymmetrically compatible BPS-particles in the theory carry non-negative electric charge under the Abelian gauge field $A^I$. In \cite{Katz:2020ewz}, the worldsheet theory of these monopole strings is investigated and the anomaly inflow on the string due to space-time gauge theories is used to constrain the possible gauge theories in 5d supergravity theories. In particular, \cite{Katz:2020ewz} identifies a class of supergravity strings for which they conjecture that the worldsheet theory flows to a $\mathcal{N}=(0,4)$ SCFT with $\SU(2)$ R-symmetry. This class of supergravity strings is characterized by the condition $C_{IJK}p^Ip^Jp^K>0$. On the other hand, supergravity strings that arise from higher dimensional supergravity strings (e.g., 6d) or strings for which the worldsheet supersymmetry gets enhanced, can have $C_{IJK}p^Ip^Jp^K=0$. In view of our classification of supergravity string in 4d in terms of the parameter $q$ we are thus led to identify the $q=2$ strings as the 4d analogue of the $C_{IJK}p^Ip^Jp^K>0$ strings as these strings do not arise from a higher dimensional supergravity string. On the other hand, the $q=0,1$ strings should be viewed as being the analogues of the strings in the $C_{IJK}p^Ip^Jp^K=0$ class. In this paper, we do not attempt to scrutinize this analogy further, e.g., by investigating the worldsheet theory on the different 4d supergravity strings, but leave this task for future work. 

We observe that whereas in six and five dimensions, the relevant scale to distinguish different kinds of supergravity string is the respective Planck scale, the analysis in this section shows that the relevant scale in four dimensions is the species scale associated to the KK tower.\footnote{In the case $q=0$ this is not the actual species scale, but just the would-be species scale if we did not have an EFT string.} In particular, we have the three options: 
\begin{align}
    \frac{T_\text{\tiny EFT}}{\Lambda_\text{\tiny sp,KK}^2} \ll \mathcal{O}(1) \,, \qquad \frac{T_\text{\tiny EFT}}{\Lambda_\text{\tiny sp,KK}^2} \sim \mathcal{O}(1)\coma \text{or}\quad \quad \frac{T_\text{\tiny EFT}}{\Lambda_\text{\tiny sp,KK}^2} \gg \mathcal{O}(1)\,,
\end{align}
corresponding to $q=0,1,$ and $2$, respectively. Replacing $\Lambda_\text{\tiny sp,KK}$ by $M_\text{\tiny Pl}$ in six dimensions, the first case is the analogue of a string with charge vector $q$ satisfying $q \cdot q=0$ whereas the second case corresponds to $q \cdot q>0$. The last case does not have an analogue in six dimensions. 

We will discuss the consequences of the findings of this section for the WGC in Section \ref{ssec:4dWGC}.

\subsection{Non-EFT string weak coupling limits}
\label{ssec:nonEFTgaugetheory}

So far, we have focused our discussion on weak coupling limits for gauge theories that can be thought of as perturbative gauge theories for a quasi-primitive EFT string. We now turn to gauge theories whose weak coupling limits cannot be achieved as quasi-primitive EFT string limits in our chosen chamber of $\mathcal{K}_{\rm ext}$. For these gauge theories, we can show the following 

\begin{proposition}\label{prop:nonprimitive}
Given a generator $D_0$ of $\text{Eff}\,{}^1(B_3)$ such that the limit $\mathcal{V}_{D_0}\rightarrow \infty$ cannot be realized as a  quasi-primitive EFT string limit. Then the weak coupling limit for a gauge theory on any divisor ${\cal S} = \kappa D_0 + \ldots$ corresponds either to a limit in which the gauge theory effectively becomes a defect theory in an 8d or 10d gravitational bulk theory, or to a limit in which the gauge theory effectively becomes a generically non-weakly coupled 8d theory coupled to gravity.  
\end{proposition}

The proof is again provided in Appendix~\ref{app:proofs}.

According to Proposition \ref{prop:nonprimitive}, weak coupling limits for gauge theories that do not correspond to a quasi-primitive EFT string limit always lead to higher dimensional theories. As for the quasi-primitive limits studied in the previous section, the asymptotic theory cannot contain any non-critical EFT string whose tension lies between the scale of the KK-tower and the associated KK species scale. By contrast, there can appear a critical string between these two scales. This is possible if we take a $q=1$ EFT string limit, leading to a decompactification to six dimensions, and on top of this an additional limit is taken in the six-dimensional theory. The combined limit is then no longer a $q=1$ EFT string limit. To allow for such a limit, $B_3$ must be a surface fibration over $\mathbb{P}^1$ (in order for a decompactification limit to six dimensions to exist) and the surface fiber itself must admit for a $\PP^1$-fibration (in order for a critical heterotic string to exist in the six-dimensional effective theory). 

To illustrate this point, we assume for simplicity that the F-theory base $B_3$ admits a simple fibration structure $p:\mathbb{P}^1_0\rightarrow \left(\mathbb{P}^1_1 \rightarrow \mathbb{P}^1_2\right)$. In this case, the Mori cone is generated by the curve classes of the three  $\PP^1_i$, $i=0,1,2$, with volumes  $v^i$. In order to reach a decompactification limit to six dimensions with a critical string below the species scale, we can assume the general scaling
\begin{align}
    v^0 \sim \lambda^{-a}\coma v^1 \sim \lambda^b \coma v^2\sim \lambda^c\,,\qquad a,b,c\geq 0\,. 
\end{align}
Here we have chosen $a\geq 0$ to engineer the relation $T_\text{\tiny het}\lesssim M_\text{\tiny IIB}^2$ between the tension of the heterotic string and the type IIB scale. In order for the limit $\lambda\rightarrow \infty$ to correspond to a decompactification to six (rather than to eight or ten) dimensions, the parameters must lie in the range where $c>b$ such that 
\begin{align}
  \frac{M_\text{\tiny Pl}^2}{M_\text{\tiny IIB}^2}  \sim  \mathcal{V}_{B_3} \sim \lambda^{b+c-a}\,,
\end{align}
independently of the chosen twists. The KK-scale is given by
\begin{align}
    \frac{M_\text{\tiny KK}^2}{M_\text{\tiny IIB}^2} \sim \lambda^{-c}\,,
\end{align}
and, by \eqref{eq:Lambdasp}, the species scale is 
\begin{align}
    \frac{\Lambda_\text{\tiny{sp,KK}}^2}{M_\text{\tiny Pl}^2} \sim \lambda^{\frac{a-b}{2}-c}\,.
\end{align}
Using $T_\text{\tiny het}/M_\text{\tiny IIB}^2 \sim \lambda^{-a}$ we find that $T_\text{\tiny het}\lesssim \Lambda_\text{\tiny sp,KK}^2$ provided $a+b\geq 0$, which is the case by assumption. However, if $a=b=0$ the theory does not undergo any additional limit after decompactifying to six dimensions. The simplest option to engineer an additional limit is to take $a=b\neq0$. In this case, the species scale coincides with  $M_\text{\tiny IIB}$, signaling that the six-dimensional Planck scale is constant in units of $M_\text{\tiny IIB}$. The resulting six-dimensional limit then corresponds to the emergent string limits analyzed in \cite{Lee:2018spm,Lee:2018urn}. The volume of the divisors scale as 
\begin{align}
    \mathcal{V}_{p^*(\PP_2^1)}\sim \lambda^{c-a} \coma \mathcal{V}_{p^*(\PP_1^1)}\sim \text{const.} \coma \mathcal{V}_{\PP_1^1\rightarrow \PP_2^1}\sim \lambda^{c+a}\,.
\end{align}
Notice that the condition $c>b=a$ ensures that no divisor is shrinking in the limit, which is necessary to retain perturbative control \cite{Klaewer:2020lfg}. We observe that a gauge theory on a 7-brane wrapping the divisor $\PP_1^1\rightarrow \PP_2^1$, the base $B_2$ of the fibration $p$, becomes weakly-coupled at the fastest rate, but unlike in the $q=0$ primitive EFT string limit, the gauge theory becomes effectively six-dimensional in the asymptotic limit.

To summarize, there also exist certain limits in which 
the tension of a critical string lies above the KK-scale but below the KK-induced species scale,
\begin{align}  \label{MKKTcrit}
M^2_\text{\tiny KK}  \ll T_\text{\tiny crit} \ll \Lambda^2_\text{\tiny sp,KK}   \,.
\end{align}
As a result, the actual species scale is set by $T_\text{\tiny crit}$. Since  $\Lambda_\text{\tiny sp,KK}$ can be identified with the Planck scale in the higher-dimensional, i.e., six-dimensional, theory, this implies that the latter undergoes an additional infinite distance limit corresponding to a weak coupling limit for the six-dimensional gauge theory on $B_2$. From this perspective it is not surprising that a relation of the form \eqref{MKKTcrit} is only possible for critical, rather than general EFT, strings since in six dimensions critical strings are the only strings which can become tensionless, in Planck units, in weak coupling limits.\footnote{Non-critical strings from D3-branes wrapping curves of negative self-intersection on the base of an elliptic $3$-fold become tensionless in the strongly coupled SCFT regime, while, as discussed around \eqref{C06d}, strings associated with curves of positive self-intersection cannot become tensionless in Planck units.} Therefore, in four dimensions, there cannot appear any other weakly coupled strings with a tension below $\Lambda_\text{\tiny sp,KK}^2$, which in the decompactification limit is the six-dimensional Planck scale.

\subsection{Consequences for the four-dimensional WGC} 
\label{ssec:4dWGC}

Let us summarize our findings so far and discuss their consequences for the WGC.

We have studied weak coupling limits 
for four-dimensional gauge theories realized on a stack of 7-branes in F-theory.
Our goal was to identify a marginally super-extremal tower of states as predicted by the asymptotic tWGC.
  A natural candidate for such states are the excitations of a certain (solitonic) EFT string.
The EFT string is
obtained by wrapping a D3-brane on a suitable curve on the base of the F-theory elliptic fibration such that the EFT string tension $T_\text{\tiny EFT}$
sits at the expected weak gravity scale
\begin{align}
    \Lambda_\text{\tiny WGC} \sim g_\text{\tiny YM} M_\text{\tiny Pl}\,.
\end{align}
Contrary to expectations based on this parametric behavior alone,
we have found that the excitations of the EFT string satisfy the asymptotic tower WGC in its form \eqref{eq:WGC} only if 
\begin{enumerate}[label=\alph*)]
    \item the solitonic EFT string with $T_\text{\tiny EFT} = \Lambda_\text{\tiny WGC}$ is a heterotic string and
    \item the gauge theory can be identified with a \emph{perturbative} gauge theory in the dual heterotic string duality frame associated with the EFT string.
\end{enumerate}
We arrived at this conclusion by classifying the weak coupling limits imposed by the backreaction of EFT strings. For these limits, we found that conditions a) and b) are necessary and sufficient to ensure that the extremality bound \eqref{eq:WGC} is satisfied by the hypothetical excitations of the EFT string and that the theory can be treated as a four-dimensional theory even in the weak coupling limit.
The hallmark of such situations is that the relevant scales satisfy the relation 
\begin{align}\label{eq:scales}
    M_\text{\tiny KK}^2 \sim T_\text{\tiny EFT} \sim \Lambda_\text{\tiny{WGC}}^2 \ll \Lambda_\text{\tiny{sp, KK}}^2\,.
\end{align}
Here, $\Lambda_\text{\tiny WGC}$ serves as the cut-off scale for the gauge theory required by the magnetic WGC. 
The crucial point is that the first three scales in \eqref{eq:scales} lie below the would-be species scale $\Lambda_\text{\tiny{sp, KK}}$ associated to the KK-tower of mass $M_\text{\tiny KK}$. This ensures that the EFT string and its excitations can be viewed as genuinely four-dimensional and weakly coupled to gravity. On the other hand, the fact that  $\Lambda_\text{\tiny WGC}\ll \Lambda_\text{\tiny sp,KK}$ indicates that also the gauge theory can be treated as a weakly coupled gauge theory in a weakly coupled four-dimensional theory of gravity even at the scale at which we expect to encounter the super-extremal states required by the WGC. Note that the actual species scale in this limit is determined by the excitations of the (critical) EFT string tension itself \cite{Dvali:2009ks,Dvali:2010vm}, and hence lies slightly above the WGC scale $\Lambda_\text{\tiny WGC}$.

In our classification of Section \ref{ssec:Weakcouplinglimits}, weak coupling limits of the above type correspond to the so-called $q=0$ EFT string limits.\footnote{Notice that one could also modify the $q=0$ EFT string limit slightly and impose a non-homogeneous scaling. As long as one can still interpret this limit as a $q=0$ EFT string limit plus small corrections, one still expects the asymptotic tWGC to be fulfilled, see the discussion in \cite{Klaewer:2020lfg}.}
Theses are the emergent string limits studied in four dimensions in \cite{Lee:2019tst,Klaewer:2020lfg}.

The situation for the weak coupling limits that cannot (to leading order) be described as a $q=0$ EFT string limit is different. We have found two possible asymptotic hierarchies:
\begin{equation}\label{scaleshierarchy}
\renewcommand{\arraystretch}{1.2}
\begin{array}{crcccl}
    & M_\text{\tiny KK} & \ll & \Lambda_\text{\tiny sp,KK} & \sim & \Lambda_\text{\tiny WGC}\\
    \text{or } & M_\text{\tiny KK} &\ll & \Lambda_\text{\tiny sp,KK} &\ll & \Lambda_\text{\tiny WGC}\,. 
\end{array}
\end{equation}
The first type of hierarchies is realized in generic $q=1$ EFT string limits in the language of Section \ref{ssec:Weakcouplinglimits} and similar non-EFT string limits (as discussed in Section \ref{ssec:nonEFTgaugetheory}). Since the magnetic WGC cut-off $\Lambda_\text{\tiny WGC}$ of the gauge theory parametrically lies at the KK species scale, the gauge theory should be viewed as a gauge theory in a higher-dimensional setting. 
Without further specializations of the limit, the higher-dimensional gauge theory is not weakly coupled.
It is therefore not surprising that the na\"{\i}ve tower of charged excitations of, e.g., the $q=1$ string does not satisfy the repulsive force condition \eqref{eq:WGC}. On the one hand,
in order to formulate a condition such as
\eqref{eq:WGC}, the gravitational and Coulomb force must effectively be described by a Newton and Coulomb potential. In a strongly coupled gauge and gravity theory, this assumption is certainly not valid. On the other hand, in the higher-dimensional duality frame also the EFT string is strongly coupled. The na\"{\i}ve assumption  of having a perturbative string spectrum, which underlies the analysis of Section \ref{ssec:Repulsiveforcecond}, is therefore not justified. Notice that this does not mean that there is no tower of super-extremal states at all, but that such a tower cannot come from the excitations of a perturbative string and that the actual super-extremality condition might differ significantly from \eqref{eq:WGC}.\footnote{Note also that by specifying the limit further one can engineer a weak coupling limit in the effective higher-dimensional theory such that the tWGC is satisfied by excitations of an emergent heterotic string in the higher dimensional theory. See Section \ref{ssec:nonEFTgaugetheory} for a discussion.} We will come back to this point in the next section.

The second type of hierarchies in \eqref{scaleshierarchy} is different and corresponds to $q=2$ EFT string limits and their analogous non-EFT string limit counterparts. Here, the gauge and the gravitational theory decouple entirely, as is signaled by the relation $\Lambda_\text{\tiny sp,KK}/\Lambda_\text{\tiny WGC} \gg 1$. This is consistent with the fact that in such limits, the gauge theory asymptotically reduces to a defect theory in a higher-dimensional gravity theory. Since gravity decouples from the gauge sector, the WGC conjecture should become trivial in this limit. 
Indeed, consider the charged states associated with $(p,q)$-strings with mass $M_k^2 \sim k M_\text{\tiny IIB}^2$. The charge-to-mass ratio for these states diverges asymptotically, 
\begin{align}
    \frac{g_{\YM}^2 q^2}{M_k^2/M_\text{\tiny Pl}^2} \rightarrow \infty\,,
\end{align}
 such that the constraint from the WGC becomes trivial, as expected. There is hence no need for an additional, {\it marginally} super-extremal tower of states coming from the excitations of an EFT string. On the contrary, it would be surprising if there was such a tower of states that marginally satisfies the WGC, since the gauge theory and the gravity theory decouple at scales corresponding to $\Lambda_\text{\tiny WGC}$.

Our analysis reveals that in order to test the asymptotic WGC in four dimensions based on the properties of axionic or EFT strings, it is not enough to just consider the weak coupling limit for a gauge theory in an EFT with some cut-off $\Lambda=m_*$.  Instead, one also needs to know the nature of the lightest tower of states with mass of order, $m_*$ since this determines the species scale. We have found that the WGC in its usual form \eqref{eq:WGC} is only  satisfied by the tower of excitations of an axionic string if the latter sets the species scale. Otherwise, the species scale associated with the KK tower per se is always at or below the string scale and one either ends up with a (generically) strongly coupled theory in higher dimensions or with a defect gauge theory from which gravity is decoupled.

\section{Discussion}
\label{sec:Conclusions}

In this section, we would like to extend the conclusions of our analysis of weak coupling limits   beyond the concrete realization of the gauge sector via 7-branes in F-theory.

Quite generally,
the nature of a weak coupling limit $g_{\YM}\rightarrow 0$ of a four-dimensional gauge theory coupled to quantum gravity is controlled by the ratio
of the magnetic weak gravity scale $\Lambda_\text{\tiny WGC}\sim g_\text{\tiny YM} M_\text{\tiny Pl}$ and the species scale $\Lambda_\text{\tiny sp}$ of the quantum theory. There are four possible conceivable regimes:

\begin{enumerate}[label={{\itshape\roman*})},ref={{\itshape\roman*})}]
    \item\label{list:limits0} 
   The WGC scale lies parametrically below the species scale, i.e. 
    \begin{align} \label{scale-rel0}
        \frac{\Lambda_\text{\tiny WGC}}{\Lambda_\text{\tiny sp}}\ll\mathcal{O}(1)\,.
    \end{align}
    This type of weak coupling limits characterizes decompactification limits in which $\Lambda_\text{\tiny WGC}$ can be identified with the scale of the KK tower and in which the latter furnishes the super-extremal tower for a KK $\U(1)$: The species scale is set by $\Lambda_\text{\tiny sp,KK}$, and  \eqref{scale-rel0} follows from \eqref{eq:Lambdasp} by identifying $M_\text{\tiny KK} \sim \Lambda_\text{\tiny WGC}$. By the Emergent String Conjecture, decompactification limits should be the only type of limits falling into the class \eqref{scale-rel0}.
    In particular, all
     limits in which the super-extremal tower is provided by BPS states (not associated with a tower of string excitations) should have a (possibly dual) interpretation of this type. Note that in string theory, the super-extremal states and the gauge sector are realized in the closed sector.

    \item\label{list:limits1} The WGC scale lies marginally below the species scale, i.e., 
    \begin{align} \label{scale-rel1}
        \frac{\Lambda_\text{\tiny WGC}}{\Lambda_\text{\tiny sp}}\lesssim \mathcal{O}(1)\,,
    \end{align}
    such that gravity remains weakly coupled at the WGC scale. This scenario is realized in heterotic string theory and for theories which asymptote, up to duality, to a four-dimensional heterotic theory in the weak coupling regime (emergent string limit).\footnote{Here the species scale is determined by the degeneracy of the heterotic string excitations and is thus (slightly) above the heterotic string scale \cite{Dvali:2009ks,Dvali:2010vm} and the WGC scale.} 
    In such limits, the gauge theory is effectively  weakly coupled in four dimensions and the WGC predicts a tower of states satisfying the repulsive force condition \eqref{eq:WGC}. 
    Indeed, the emergent string tower contains an infinite tower of charged states which are {\it marginally} super-extremal, as in Figure \ref{fig:extrplot}.
    By the Emergent String Conjecture, all weak coupling limits with the property \eqref{scale-rel1} are of the emergent string type.

    \item\label{list:limits2} The WGC scale is of the order of the species scale, i.e., 
    \begin{align}
        \frac{\Lambda_\text{\tiny WGC}}{\Lambda_\text{\tiny sp}}\sim \mathcal{O}(1)\,.
    \end{align}
    Hence gravity is strongly coupled at the WGC scale. As a consequence, the gauge theory cannot be viewed as a weakly coupled theory in four dimensions, and one  does not expect a tower of perturbative states satisfying the four-dimensional repulsive force condition as in \eqref{eq:WGC}. The EFT string limits with $q=1$ (and their generalizations) are examples of this behavior. Further refinements of the limit may result in a weak coupling limit of a higher-dimensional gauge theory coupled to gravity, in which the relation \eqref{scale-rel1}
    holds with respect to the higher dimensional weak gravity and species scale.

    \item\label{list:limits3} The WGC scale parametrically lies above the species scale, i.e.,
    \begin{align}
        \frac{\Lambda_\text{\tiny WGC}}{\Lambda_\text{\tiny sp}} \gg \mathcal{O}(1)\,. 
    \end{align}
    At the WGC scale, the gauge theory is effectively decoupled from the gravity sector, as for example in the EFT string limits with $q=2$. In open string realizations, the repulsive force condition is trivially satisfied by highly super-extremal open string states. Any marginally asymptotic tower of states, even if present, would necessarily decouple from the gauge theory. 
\end{enumerate}
Notice that the asymptotic decoupling of the gravitational and gauge sector does not need to correspond to a geometric decompactification limit in which  the gauge theory becomes a defect theory in a higher dimension, even though this was the case for the corresponding limits on 7-branes studied in this paper.
As an example, consider a D3-brane gauge theory in type IIB Calabi--Yau orientifolds. The gauge coupling strength is set by the string coupling, 
\begin{align}
    g_{\YM}^2 \sim g_s \,, 
\end{align}
such that the weak coupling limit corresponds to the regime where $g_s\rightarrow 0$ (with the other moduli unchanged). In this case, the species scale can be identified with $M_\text{\tiny IIB}$, the tension of the fundamental type IIB string. Since $M_\text{\tiny Pl}^2$ is suppressed with respect to $M^2_s$ by one additional power of $g_s$,
\begin{align}
    \frac{M_\text{\tiny Pl}^2}{M_\text{\tiny IIB}^2} \sim g_s^{-2}\,,
\end{align}
one concludes that 
\begin{align} \label{WGCsp-summary}
 \frac{\Lambda^2_\text{\tiny WGC}}{\Lambda^2_\text{\tiny sp}} 
 \sim \frac{g_{\YM}^2 M_\text{\tiny Pl}^2}{M_\text{\tiny IIB}^2} \sim g_s^{-1} \rightarrow \infty \,. 
\end{align}
This places the regime $g_s \to 0$ into the context of the weak coupling limits of \ref{list:limits3}, and
effectively decouples the gauge theory on the D3-branes from the gravity sector. In the present framework, this amounts to the well-known statement that for type IIB orientifolds in the limit $g_s\rightarrow 0$ the open and closed string sectors decouple. The WGC is now trivially satisfied by the open type IIB string excitations that become infinitely super-extremal in the limit $g_s\rightarrow 0$.  

A key difference of such perturbative open string towers compared to their heterotic counterparts is, of course, that the charges of the string states are only determined by the Chan-Paton factors, and hence the highest charge per excitation level does not increase with the level. Strictly speaking, then, the string excitations tower only provides a finite number of super-extremal states for a fixed value of $g_s \ll 1$. This is in contrast with the marginally super-extremal tower of infinitely many states which arise in the weak coupling limit of a heterotic string, even for a fixed value of the heterotic dilaton. See Figure \ref{fig:extrplot} for an illustration.

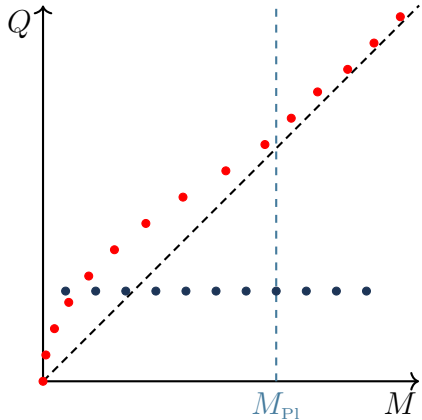
\begin{figure}[!tp]
    \centering
    \begin{tikzpicture}[scale=1,baseline]
        \draw[thick,->] (0,0) -- node[below,pos=0.95] {$M$} (5,0);
        \draw[thick,->] (0,0) -- node[left,pos=0.95] {$Q$} (0,5);
        \draw[thick,densely dashed] (0,0) -- (5,5);
        \draw[thick,dashed,seccolor] (3.10,0) -- node[below,pos=0] {$M_{\text{\tiny Pl}}$} (3.10,5);
        \foreach \y in {0,0.35,...,3} \node[circle,draw=prhigh,fill=prhigh,inner sep=0pt,minimum size=3pt] at (0.31*\y^2,\y) {};
        \foreach \y in {3.15,3.5,...,4} \node[circle,draw=prhigh,fill=prhigh,inner sep=0pt,minimum size=3pt] at (\y-0.2,\y) {};
        \foreach \y in {4.15,4.5,...,5} \node[circle,draw=prhigh,fill=prhigh,inner sep=0pt,minimum size=3pt] at (\y-0.1,\y) {};
         \foreach \x in {0.3,0.7,...,4.5} \node[circle,draw=prcolor,fill=prcolor,inner sep=0pt,minimum size=3pt] at (\x,1.2) {};
    \end{tikzpicture}
    \caption{Super-extremal states for weakly coupled open ({\color{prcolor}{blue}} dots) versus heterotic strings ({\color{prhigh}{red}} dots). For a fixed value of the dilaton, the open string only provides a finite number of super-extremal states, while the heterotic string contains a marginally super-extremal tower of states.}
    \label{fig:extrplot}
\end{figure}

One might wonder to what extent our results are in conflict with the WGC or its tower version.
First, our focus has been entirely on weak coupling limits of the gauge sector and hence on the {\it asymptotic WGC}.
It is in this regime where one has computational control and therefore the best chances of reliably identifying a super-extremal tower. More importantly perhaps, some of the original bottom-up arguments in favor of the WGC primarily hold in the weak coupling limit, and in fact in limits of type \ref{list:limits0} or \ref{list:limits1}: As was argued in \cite{Arkani-Hamed:2006emk}, if extremal black holes were not able to decay, this would be in tension with the covariant entropy bound in the limit $g_\text{\tiny YM} \to 0$, in which an infinite number of extremal black holes can be constructed with mass below every finite cutoff. Away from this asymptotic regime, no immediate contradiction with entropy considerations arises. 
Clearly, this does not mean that the WGC or its tower version does not hold also away from the weak coupling limit, but it is less clear, from a bottom-up point of view, why this would have to be the case.

Indeed, examples where the tower WGC is satisfied even at strong coupling have been studied in detail in \cite{Alim:2021vhs}, in the framework of M-theory compactifications to  five dimensions with eight supercharges: The asymptotic tower of states is formed by BPS states, at least in favorable circumstances. 
As long as the WGC bound and the BPS bound coincide, the appearance of a tower of exactly super-extremal states was found to be protected by supersymmetry. 
Our focus in this paper, by contrast, is on minimally supersymmetric settings in which the WGC states, if any, are not BPS. 
It is striking that in all bona fide weak coupling limits coupled to gravity, a marginally super-extremal tower of states can be identified beyond doubt, from the excitations of an (emergent) perturbative heterotic string. 
In all other weak coupling limits, where no such tower is available among the weakly coupled and hence well-controlled sets of states, 
the WGC is substantially less well motivated 
also from a bottom-up perspective -- either because the gauge theory is no longer weakly coupled after passing to the emergent higher dimensional duality frame or because the gauge and the gravity sector decouple.

While we do therefore not find any contradiction to the asymptotic tWGC (valid in the limit $g_\text{\tiny YM} \to 0$ for a gauge theory coupled to weakly coupled gravity), it would clearly be interesting to continue analyzing other potential sources of super-extremal towers beyond these regimes. 

There are two qualitatively different situations to consider. The first would be a gauge theory coupled to gravity away from the regime $g_\text{\tiny YM} \to 0$, but in situations where the BPS condition does not protect the WGC tower.\footnote{This includes limits of type \ref{list:limits2} after dualizing to the higher dimensional frame where the gauge theory is  -- in general -- not weakly coupled.}
It is tempting to hypothesize that traces of a marginally super-extremal tower should indeed be present whenever a continuous deformation connects the strongly-coupled theory to a weakly coupled regime without decoupling gravity, i.e., in situations where one {\it can} take a limit of type \ref{list:limits1}.
To settle this question, one would have to follow the excitations of 
the emergent perturbative heterotic string regime 
into the strongly coupled region of moduli space.
 First steps in this direction were taken already in \cite{Klaewer:2020lfg}, where subleading corrections to the charge-to-mass ratio were taken into account.
Clearly, whether a full tower of super-extremal states survives parametrically away from weak coupling, and without the protection of a BPS condition, is a considerably more ambitious question.

The second type of challenges for the WGC is posed by those theories which are not smoothly connected to a weak coupling limit of type \ref{list:limits1} and for which the weak coupling limit of the gauge theory necessarily implies a decoupling from gravity. We have exemplified one such instance in our discussion of the gauge theory on D3-branes in Type IIB orientifolds. 
To identify a tower of marginally super-extremal states, the only option seems to be to resort to the D$1$-brane, or more generally to $(p,q)$-strings, either by hypothesizing about their potential excitations or by arguing for stretched, possibly multipronged string networks giving rise to particle states in four dimensions:\footnote{In F-theory, such multipronged string networks can lead to a proliferation of charges, but from all we believe to know this requires strong coupling rather than weak string coupling.}
At the parametric level, $T_\text{\tiny D1} \sim \Lambda^2_\text{\tiny WGC}$, but in the limit $g_s \to 0$ under consideration, this implies that the D$1$-brane theory becomes asymptotically strongly coupled and hence, by design, decouples from the gauge theory because of \eqref{WGCsp-summary}. 
Again, not only does it not seem plausible to obtain a marginally super-extremal tower of states from such strongly coupled objects, it appears not even to be required in the limit in which \eqref{WGCsp-summary} holds.

More generally, a gauge theory in F-theory can never become weakly coupled along a direction compatible with \ref{list:limits1} if the base of the elliptic fibration does not admit a rational fibration, or if the stack of gauge 7-branes does not intersect the rational fiber. As we have shown, at least the (hypothetical) excitations 
of the weakly coupled EFT strings sitting at the weak gravity scale do apparently not satisfy the perturbative repulsive force condition and hence do not yield a marginally super-extremal asymptotic tower. Similarly to the theory on the D3-branes, it is hard to imagine which other states could produce a marginally super-extremal tower instead.
At the same time, a highly super-extremal pseudo-tower of states is provided by the open string sector in such situations; it comprises only finitely many super-extremal states for fixed $g_\text{\tiny YM} \ll 1$, but infinitely many in the limit  $g_\text{\tiny YM} \to 0$.

Upon circle compactification, these states will continue to act as the super-extremal states for the original gauge group; however, taking into account the effect of the KK $\U(1)$ as in \cite{Heidenreich:2015nta}, one might worry that the convex-hull condition will now be violated.
Even if this turned out to be the case, 
the gravity sector of the lower-dimensional theory  continues to decouple from the gauge theory, since the relation $\Lambda^2_\text{\tiny WGC}/\Lambda^2_\text{\tiny sp} \to \infty$ still holds. In this sense, there is no immediate contradiction with the WGC even after circle compactification, according to our more conservative interpretation of the WGC.

Barring these subtleties, it is fair to say that in all instances in which a marginally super-extremal tower of states has been reliably confirmed in gauge theories coupled to gravity in Minkowski space, this is either due to the BPS condition or an artifact of the specific shape of the charge-to-excitation-level diagram of the perturbative heterotic string. 
In fact, in all these cases, the tower WGC (in flat Minkowski space) follows from the Emergent String Conjecture in the sense that the super-extremal tower either furnishes a (dual) KK tower (in situations where the tower is BPS) or corresponds to the excitations of an emergent heterotic string. 
In all other situations, in particular away from weak coupling and amiss of a BPS protection, the tower WGC can in principle be violated and in fact is not even required for bottom-up consistency of the theory.
It will be interesting to find a counter-example to this preliminary interpretation, or, if corroborated, to understand further what it teaches us about the true rationale behind the WGC in general quantum gravity theories.

\section*{Acknowledgments}

We thank Rafael \'Alvarez-Garc\'ia, Daniel Kl\"awer, Ben Heidenreich, Luis Ib\'a\~nez, Seung-Joo Lee, Jacob Leedom, Wolfgang Lerche, Guglielmo Lockhart, Dieter L\"ust, Fernando Marchesano, Luca Martucci, Miguel Montero, Eran Palti,  Nicole Righi, Cumrun Vafa, Irene Valenzuela, and Jian Xiao for helpful discussions and correspondence. 
C. F. C., A. M. and T. W. are supported in part by Deutsche Forschungsgemeinschaft under Germany's Excellence Strategy EXC 2121  Quantum Universe 390833306 and by Deutsche Forschungsgemeinschaft through a German-Israeli Project Cooperation (DIP) grant ``Holography and the Swampland”. M. W. is supported in part by a grant from the Simons Foundation (602883, CV) and also by the NSF grant PHY-2013858. 

\appendix

\section{The geometry of EFT strings in F-theory}
\label{app:geometrytools}

This appendix briefly covers the necessary algebraic geometrical notions for describing EFT strings in F-theory \cite{Lanza:2020qmt,Lanza:2021udy,Lanza:2022zyg}. 
Accordingly, we consider a compact K\"ahler manifold $B_n$ of dimension $n$ as the base of an elliptically fibered Calabi-Yau.  
For further details on our assertions and definitions in what follows, we refer the reader to \cite{MR2095471}.

As a starting point, let us introduce the N\'eron-Severi group, the set of divisors modulo numerical equivalence, which in our case of interest is given by
\begin{equation}
N^1(B_n) = H^{1,1}(B_n) \cap H^2(B_n,\mathbb{Z})/\text{torsion}\,,
\end{equation}
and its real extension  $N^1(B_n)_{\mathbb{R}} = N^1(B_n) \otimes_{\mathbb{Z} } \mathbb{R}$.  
 Similarly, we introduce its intersection pairing counterpart, the set of curves modulo numerical equivalence, which is 
\begin{equation}
N_1(B_n) = H^{n-1,n-1}(B_n) \cap H^{2n-2}(B_n,\mathbb{Z})/\text{torsion}\,,
\end{equation}
and its real extension  $N_1(B_n)_{\mathbb{R}} = N_1(B_n) \otimes_{\mathbb{Z}}\mathbb{R}$.  
In the following, we describe the convex cones in $N^1(B_n)_{\mathbb{R}}$ and $N_1(B_n)_{\mathbb{R}}$ that are relevant to us:

\begin{enumerate}[align=left,wide,  labelindent=0pt]
\item[\textbf{The pseudo-effective cone}:] The cone of effective divisors $\text{Eff}^1(B_n)\subseteq N^1(B_n)_{\mathbb{R}}$ is the convex cone of all effective divisors, i.e.,   positive linear combinations of complex codimension-one cycles on $B_n$.  
The pseudo-effective cone $\overline{\text{Eff}}^1(B_n) \subseteq  N^1(B_n)_{\mathbb{R}}$ is the closure of $\text{Eff}^1(B_n)$.

\item[\textbf{The K\"ahler cone}:] The K\"ahler cone is the set $\mathcal{K}(B_n) \subset H^{1,1}(B_n,\mathbb{R})$ of classes of K\"ahler forms $\{J\}$ on $B_n$. 
We identify 
\begin{equation} 
\mathcal{K}(B_n) =\text{Amp}(B_n) \subset N^1(B_n)_{\mathbb{R}}\,,
\end{equation} which is the convex cone of all ample divisors on $B_n$, 
as $\int_{V} J^{k}  > 0$ for every $J \in \mathcal{K}(B_n)$ and all irreducible $V \subseteq B_n $ with $k = \text{dim}(V)>0$ according to the Nakai-Moishezon criterion for ampleness~\cite{MR2095471}.  
Similarly, we identify its closure to be 
\begin{equation}
\overline{\mathcal{K}}(B_n)  = \text{Nef}(B_n) \subset N^1(B_n)_{\mathbb{R}}\,,
\end{equation}
 which is the convex cone of all nef divisors on $B_n$, as  $\int_{V} J^{k}  \geq 0$ for every $J \in \overline{\mathcal{K}}(B_n)$ and all irreducible $V \subseteq B_n $ with $k = \text{dim}(V)>0$ following Kleiman's theorem~\cite{MR2095471}. 

\item[\textbf{The Mori cone}:] The cone of curves $\text{NE}(B_n) \subseteq N_1(B_n)_{\mathbb{R}}$ is the cone spanned by the classes of all effective  complex one-cycles on $B_n$. 
Its closure $\overline{\text{NE}}(B_n) \subseteq N_1(B_n)_{\mathbb{R}}$, also known as the \textit{Mori cone}, is dual to the nef cone, i.e
\begin{equation}
\overline{\text{NE}}(B_n) = \{ \mathcal{C} \in N_1(B_n)_{\mathbb{R}} \mid J \cdot \mathcal{C} \geq 0 \text{ for all } J \in \text{Nef}(B_n) \} \,.
\end{equation}

\item[\textbf{The Movable cone}:] A movable curve is an irreducible complex one-cycle that is a member of an analytic family $\{C_t\}_{t\in S}$ such that $\bigcup_{t \in S}C_t = B_n$.  
We denote by $\text{ME}(B_n)$ the convex cone generated by movable curves. We call its closure $\text{Mov}_1(B_n) = \overline{\text{ME}}(B_n)$ the \textit{movable cone}. 
If $B_n$ is projective, then the movable cone is dual to the pseudo-effective cone of divisors~\cite{MR3019449}, i.e.
\begin{equation}
\text{Mov}_1(B_n) = \left\{ C \in N_1(B_n)_{\mathbb{R}} \mid D \cdot C \geq 0 \text{ for all } D \in \overline{\text{Eff}}^1(B_n) \right\} \,.
\end{equation}
\end{enumerate}

Finally, let us define the integral cones $\text{Eff}^1(B_n)_{\mathbb{Z}} \equiv \text{Eff}^1(B_n) \cap N^1(B_n)$ and $\text{Mov}_1(B_n)_{\mathbb{Z}} \equiv  \text{Mov}_1(B_n) \cap N_1(B_n)$. Based on the results of~\cite{Lanza:2021udy}, we summarize the correspondence between EFT objects and geometry for a given compact K\"ahler $3$-fold $B_3$:

\begin{center}
\renewcommand{\arraystretch}{1.2}
    \begin{tabular}{c|c}
    EFT data & Geometric cone \\
    \hhline{=|=}
    Saxionic cone $\Delta$ & $\text{ME}(B_3)$ \\ \hline 
    EFT strings $\mathcal{C}_{\text{EFT}}^{S}$ & $\text{Mov}_1(B_3)_{\mathbb{Z}}$  \\ \hline 
    BPS instantons $\mathcal{C}_{I}$& $\text{Eff}^1(B_3)_{\mathbb{Z}}$ \\ 
    \end{tabular}
\end{center}
\vspace{2mm}

\section{Proofs of Propositions 1 - 5}
\label{app:proofs}

In this appendix, we provide the proofs of \cref{prop:primitiveEFT,prop:EFTstring,prop:volumefactorization,prop:q0,prop:nonprimitive}.

\begin{proof}[Proof of Proposition \ref{prop:primitiveEFT}]
In order to achieve $\mathcal{V}_D\rightarrow \infty$, we need to scale up  the volume $v^0$ of some curve $\mathcal{C}^0$ inside the Mori cone contained in $D$.\footnote{In case the K\"ahler cone is simplicial, $\mathcal{C}^0$ is simply the volume of a Mori cone generator, see also Footnote \ref{foot:non-simK}.} Now, if $\mathcal{C}^0$ is contained in another divisor $\hat D\notin \mathcal{I}$, also $\mathcal{V}_{\hat D}\rightarrow \infty$ unless we can compensate by shrinking another holomorphic curve contained in $\hat D$, i.e., we perform a co-scaling. In order for this to be possible, $\hat D$ needs to be a genus-one or rational fibration over $\mathcal{C}^0$ as follows from the results of \cite{Lee:2018urn} (see, e.g., their Section 2.2\footnote{In \cite{Lee:2018urn} only the case of a rational fibration is considered explicitly because of the extra requirement that the infinite distance limits lead to a vanishing gauge coupling for a 7-brane; more generally, also genus-one fibrations are compatible with infinite distance limits leaving the volume of a surface finite.}).\qedhere
\end{proof}

\begin{proof}[Proof of Proposition \ref{prop:EFTstring}]
 Consider a quasi-primitive EFT string associated to a K\"ahler cone generator $J_0$.
 By Definition \ref{def:quasiprimitive}, the EFT string limit is given by the limit $v^0\sim \lambda\rightarrow \infty$. In general, we can now differentiate between the two cases $J_0^2\neq 0$ and $J_0^2=0$ which we are going to discuss separately:

\begin{enumerate}[align=left,wide,  labelindent=0pt]

\item[$\bf{J_0^2\neq 0}$:]
In this case, there exists a non-empty set $\{D_i\}_{i\in \mathcal{J}}$ of generators of $\text{Eff}^1 (B_3)$, for some index set $\mathcal{J}$,  such that 
\begin{align}\label{VDv0}
    \mathcal{V}_{D_i} = \frac{\kappa_i}{2} (v^0)^2 + \ldots\,,
\end{align}
for some integers $\kappa_i$. Since this is the highest  possible power of $v^0$, it is clear that $D_i \in \mathcal{I}$ $\forall i\in \mathcal{J}$. There may now also exist a different set $\{\hat D_{\hat \imath}\}$ of generators whose volumes scale like 
\begin{align}\label{VDhatv0}
    \mathcal{V}_{\hat D} =\eta_{\hat{\imath}} v^0 v^{\hat \imath} + \ldots \,,
\end{align}
for some $v^{\hat \imath}$ and integers $\eta_{\hat \imath}$. Since the scaling differs from $\mathcal{V}_{D_i}$ we need to co-scale $v^{\hat \imath}\sim \frac1\lambda \rightarrow 0$ in order to minimize $|\mathcal{I}|$ in the sense of Definition \ref{def:quasiprimitive} and keep $\mathcal{V}_{\hat D_{\hat \imath}}$ constant.\footnote{As stressed in the main text, to prevent a divisor from shrinking in the limit $v^{\hat \imath}\rightarrow 0$, one might be forced to scale up another $v^{j\neq 0}$. This is fine as long as it does not change $|\mathcal{I}|$ by adding new expanding divisors.}

All instantons with action $\re S \rightarrow \infty$ are thus obtained from D3-branes wrapping divisors containing at least one generator in $\{D_i\}_{i\in \mathcal{J}}$. From \eqref{VDv0}, we infer that these divisors thus have non-zero intersection with a curve $\tilde{C}=J_0^2$, whereas the divisors $\hat D\notin \{D_i\}_{i\in \mathcal{J}}$ have necessarily vanishing intersection with $\tilde{C}$. It is always possible to find a set of curve classes $C^i$ s.t. $D_i\cdot C^j=\delta_i^j$.\footnote{In case $\text{Eff}^1(B_3)$ is simplicial, these are simply the generators of $\text{Mov}_1(B_3)$.} In terms of these curves, we can expand $\tilde{C}$ as
\begin{align}\label{CkappaCa}
    \tilde{C} = \sum_{i\in \mathcal{J}}\kappa_i C^i \,. 
\end{align}
Notice that the $\kappa_i$ are the charges which the effective string carries w.r.t. the $2$-forms dual to $\im T_i$. Hence, an instanton is asymptotically suppressed precisely if it is charged under the D3-brane string wrapped on $\tilde{C}= J_0^2$.  From the general form of the profile of an EFT string \eqref{eq:tiprofile}, it is expected that these charges give the ratio between the saxions that blow up in the EFT string limit. And indeed, from \eqref{VDv0}, we infer that the ratio between the $\mathcal{V}_{D_i}$ for $i\in \mathcal{J}$ is precisely given by the $\kappa_i$. This ratio is not altered if we re-scale all charges of the EFT string by the same factor. Therefore, in general an EFT string giving rise to the EFT string limit in question only needs to wrap a curve $C$ proportional to $\tilde C$, i.e., $C=\alpha \tilde{C}$. Here $\alpha\in \mathbb{Q}_{>0}$ has to be chosen such that all string charges are still integer quantized. Thus, indeed a curve giving rise to the EFT string associated to our limit satisfies Condition \ref{list:EFT1new}.

\item[$\bf{J_0^2 = 0}$:]
In this case there is no generator $D$ of $\text{Eff}^1(B_3)$ whose volume scales as in \eqref{VDv0}. However, there will be a non-empty set $\{D_i\}_{i\in\mathcal{J}}$ of generators of $\text{Eff}^1(B_3)$ with volume scaling as
\begin{align}
    \mathcal{V}_{D_i} =\sum_{j=1}^{h^{1,1}(B_3)}\kappa_{i j} v^0 v^j +\ldots \,,
\end{align}
where the omitted terms are independent of $v^0$.
In order to minimize $|\mathcal{I}|$, we can now fix one $j_0\neq 0$ and co-scale $v^k\rightarrow 0$ for $j_0\neq k\in \{1,\ldots, h^{1,1}(B_3)\}$. 
Since by assumption the quasi-primitive EFT string limit exists, it is possible to find such a co-scaling without taking the volumes of any additional generators of $\text{Eff}^1(B_3)$ to infinity.
Thus, an instanton becomes irrelevant if and only if it is charged under the string obtained from a D3-brane string on $\alpha J_0\cdot J_{j_0}$, $\alpha\in \mathbb{Q}_{>0}$, since only the action of these instantons contains a term $v^0 v^{j_0}$. Thus, in this case the EFT string is obtained from a D3-brane on $C=\alpha J_0\cdot J_{j_0}$ as in Condition \ref{list:EFT2new}.\qedhere
\end{enumerate}
\end{proof}

\begin{proof}[Proof of Proposition \ref{prop:q0}]
By Proposition \ref{prop:EFTstring},
a quasi-primitive EFT string is obtained from a D3-brane wrapping a curve $C^0$ proportional $J_0^2$ or, if $J_0^2=0$, to $J_0\cdot J_1$, where $J_0$ and $J_1$ are  K\"ahler cone generators. In order for the string to have 
$q=0$, we need that $J_0^3=0$.  From the analysis of \cite{Lee:2019tst}, it follows that $C^0$ is the generic fiber of a rational or genus-one fibration $p:\mathbb{P}^1 \rightarrow B_2$ over some base $B_2$. For such a fibration, the cone of effective divisors is generated by the exceptional section $S_-$, divisors obtained by restricting the fibration $p$ to the generators of $\text{Eff}^1(B_2)$ as well as possible exceptional divisors obtained by blowing-up $B_3$.\footnote{In case the genus-one fibration does not have a section, the r\^ole of the section is played by the multi-section.} Except for $S_-$ all generators of $\text{Eff}^1(B_3)$ necessarily contain fibral curves. Hence, the intersection of these divisors with $C^0$ vanishes. Thus, the curve $C^0$ only has non-vanishing intersection with $S_-$. Therefore, the EFT string limit has to correspond to a limit where we scale up the volume of a curve contained in $S_-$ which can be expanded in terms of some of the Mori cone generators of $B_2$. Since the other generators of $\text{Eff}^1(B_3)$ are fibrations over curves containing these generators, by Proposition \ref{prop:primitiveEFT} we conclude that the $q=0$ EFT string limit is primitive. Notice that in order to arrive at the actual primitive EFT string limit, we need to shrink all curves in the fiber of $B_3$. In order for the generators of $\text{Eff}^1(B_3)$ other than $D_0 = S_-$ not to shrink to zero size, the base $B_2$ must expand in a homogeneous way. In fact, from this discussion, it follows that the limit $v^0 \to \infty$ for $J_0^3=0$ and $J_0^2\neq 0$ can always be co-scaled to lead to a primitive EFT string limit. 
\end{proof}

\begin{proof}[Proof of Proposition \ref{prop:volumefactorization}]
Since we restrict ourselves to primitive EFT strings, we can invoke Proposition \ref{prop:EFTstring} and only consider curves $C$ satisfying either Condition \ref{list:EFT1new} or \ref{list:EFT2new} and which are generators of the movable cone. If $\text{Eff}^1(B_3)$ is non-simplicial, we assume that we have fixed a basis $\{D_i\}$ of generators of $\text{Eff}^1(B_3)$ and in the following only deal with those effective cone generators that are contained in this basis.  We now treat the cases $q=0,1,2$ separately: 

\begin{enumerate}[align=left,wide,  labelindent=0pt]
    \item[$\bf{q=2}$:] In this case, by Proposition \ref{prop:EFTstring} the curve $C$ has to be of the form
\begin{equation}
    C = J_0^2\,,
\end{equation}
for some $J_0$ satisfying $J_0^3\neq 0$. For simplicity, we set the parameter $\alpha$ appearing in Proposition \ref{prop:EFTstring} to $\alpha=1$ throughout this proof. On the other hand, since we assume a primitive EFT string limit, we have a single generator, $D_0$, in our basis of $\text{Eff}^1(B_3)$ such that
\begin{align}
    \mathcal{V}_{D_0} \sim (v^0)^2 +\ldots \,,
\end{align}
with the EFT string limit corresponding to $v^0\rightarrow \infty$. By definition, the volume of all other generators of $\text{Eff}^1(B_3)$ remains constant in the primitive EFT string limit so that we find 
\begin{equation}
    \mathcal{V}_{B_3}^2 \rightarrow (v^0)^6 \sim (\re T_0)^3 \,,
\end{equation}
since no other scalar $\re T_{a\neq 0}$ contains a term proportional to $(v^0)^2$. Thus, for $q=2$ we indeed recover \eqref{eq:volumelimit} in the asymptotic EFT string limit.

\item[$\bf{q=1}$:] The case of a $q=1$ primitive EFT string corresponds to the situation in Condition \ref{list:EFT2new}, i.e., to a curve $C$ given by 
\begin{equation}
    C = J_0 \cdot J_1 \,,
\end{equation}
for some K\"ahler cone generator such that $J_0^2=0$ and $J_1^2\neq 0$. The base $B_3$ can be viewed as surface fibration over $\mathbb{P}^1$ with $J_0$ being the class of the generic fiber. In the following, we assume that this fibration is non-trivial and that $B_3$ does not allow for a second surface fibration. The two special cases, corresponding to a trivial fibration or a second surface fibration, will be discussed below. Since by assumption there exists a primitive EFT string limit there also exists a single generator $D_0$ of $\text{Eff}^1(B_3)$ with non-zero intersection with $C$ whose volume receives a contribution 
\begin{equation}
    \mathcal{V}_{D_0}\sim v^0 v^{1}+\ldots \,.
\end{equation}
The volume of any other generators of the chosen basis of $\text{Eff}^1(B_3)$ cannot contain such a term proportional $v^0v^1$ as they would intersect $J_0\cdot J_{1}$ otherwise. Since $J_0^2=0$, there cannot appear a term $(v^0)^2$ in the volume of any divisor. The primitive EFT string limit is given by $v^0\rightarrow \infty$ while keeping $v^{1}$ finite.  Still, in order to find the behavior of $\mathcal{V}_{B_3}$, we need to determine the scalings of the other parameters $v^a$, $0,1\neq a$, in the limit $v^0\rightarrow \infty$. To this end, we group the remaining elements of our basis of K\"ahler cone generators into two classes: 
\begin{equation}
\begin{split}
    \mu \in \mathcal{J}_1 \quad &\text{if} \quad J_\mu^3\neq 0 \,,\\
    r \in \mathcal{J}_2 \quad &\text{if} \quad J_r^3=0\,. 
\end{split}
\end{equation} 
First, consider the scaling of $v^r$ for $r\in \mathcal{J}_2$. If $J_r\cdot J_0\neq 0$ we know from \cite{Lee:2019tst} that $[J_r\cdot J_0]=n [J_r^2]$ for some $n>0$. Notice that this assumes that $J_r\cdot J_r\cdot J_0=0$, which is necessarily the case if the surface fibration is non-trivial.\footnote{To see this, notice that $J_r^2$ is the fiber of a genus-one or rational fibration. $J_r\cdot J_r\cdot J_0\neq 0$ implies that $J_0$ is a section of this fibration. On the other hand, $J_0^2=0$ implies that the genus-one/rational fibration is trivial. This implies, however, that the surface fibration with $J_0$ the generic fiber is also trivial.} Since by assumption $J_r^2$ is a curve different from $J_0\cdot J_{1}$, we know that the volume of a generator $D_r$ of $\text{Eff}^1(B_3)$ satisfying $D_r \cdot J_r\neq 0$ cannot have a contribution proportional to $v^0v^{1}$. On the other hand, it will certainly have a contribution 
\begin{equation}
    \mathcal{V}_{D_r} \sim v^0 v^r \,,
\end{equation}
since by assumption $D_r\cdot  J_r\cdot J_0\neq 0$. In order for this volume not to diverge in the limit $v^{0}\sim \lambda \rightarrow \infty$ we need to scale $v^r\sim 1/\lambda$. If $J_0\cdot J_r=0$ we do not have such a condition and can keep the corresponding $v^r$ finite. 

Let us now turn to the scaling of $v^\nu$ for $1\neq \nu\in \mathcal{J}_1$. Let us split the set $\mathcal{J}_1$ into two sets $\mathcal{J}_1=\mathcal{J}_1'\cup \mathcal{J}_1''$, according to:
\begin{equation}
\begin{aligned}
    \nu'\in\mathcal{J}_1'      \quad & \text{if} &      [J_0\cdot J_{\nu'}] \,& =   & n [J_0\cdot J_{1}] &\quad \text{for some $n$} \,,\\
    \rho'' \in \mathcal{J}_1'' \quad & \text{if} &    [J_0\cdot J_{\rho''}] \,&\neq & n [J_0\cdot J_{1}] &\quad \text{for any $n$} \,.
\end{aligned}
\end{equation}
Thus the volume of the divisor $D_0$ contains a contribution
\begin{equation}
    \mathcal{V}_{D_0} \supsetsim v^0\sum_{\nu' \in \mathcal{J}_1'}a_{\nu'} v^{\nu'} + \ldots \,.
\end{equation}
Since $J_0\cdot J_1\neq J_{\nu'}^2\neq 0$, there needs to be a generator of $\text{Eff}^1(B_3)$ other than $D_0$ for which the volume contains a term $(v^{\nu'})^2$. In order for the volume of this generator to remain constant in the $q=1$ primitive EFT string limit, we need $v^{\nu'}$ to remain at best constant in the limit $v^0\rightarrow 0$. None of the generators $\text{Eff}^1(B_3)$ different from $D_0$ can further contain a term proportional to $v^0 v^{\nu'}$ since this would require a non-trivial intersection with $J_0\cdot J_{1}$. 

For $\rho''\in\mathcal{J}_1''$ we know that if $\mathcal{V}_{D_0}$ contains a factor proportional to $v^{0} v^{\rho''}$, then this term needs to appear in the volume of at least one other generator, $\hat D$, of $\text{Eff}^1(B_3)$ since otherwise $[J_0 \cdot J_{\rho''}]=n[J_0\cdot J_{1}]$ which by definition of $\mathcal{J}_1''$ is not the case. In order for the volume of this additional generator to remain finite in the limit $v^0\sim \lambda \rightarrow \infty$, we need to scale $v^{\rho''}\sim 1/\lambda$. 

To conclude, we have the scalings 
\begin{equation}\label{q1limit}
\begin{aligned}
    v^0\sim\lambda \coma         & v^r \,\,\,\sim \frac1\lambda \coma       & \text{for} &\quad r\in \mathcal{J}_3 \quad \text{if} \quad J_r\cdot J_0\neq 0\,,\\
    v^{\nu'}\precsim \text{const.}\coma &v^{\rho''} \sim \frac1\lambda \coma & \text{for} &\quad \nu'\in\mathcal{J}_1'\coma\rho''\in \mathcal{J}_1''\,. 
\end{aligned}
\end{equation}
Using these scalings, the leading behavior of $\mathcal{V}_{B_3}$ in this limit is set by 
\begin{equation}
 \mathcal{V}_{B_3} \sim \sum_{\mu',\nu'\in \mathcal{J}_1'} \frac{1}{2}\kappa_{0 \mu'\nu'}v^0 v^{\mu'}v^{\nu'}\fstop
\end{equation}
Thus $ \mathcal{V}_{B_3}^2$ is a polynomial of degree-2 in $v^0$. On the other hand, $ \mathcal{V}_{B_3}^2$ is a polynomial of degree-3 in the volumes of the generators of the effective cone. Since $J_0\cdot J_0=0$, none of these volumes is proportional to $(v^0)^2$ and therefore, in the limit determined by \eqref{q1limit}, $ \mathcal{V}_{B_3}^2$ needs to be proportional to the product of two divisor volumes each containing terms proportional to $v^0v^{\nu'}$ with $\nu'\in \mathcal{J}_1'$. From our previous discussion, we know that the only such divisor is $D_0$. Hence 
\begin{equation}
    \mathcal{V}_{B_3}^2\sim (\re T_0)^2\, P_1(\re T_{\mu'})\,,
\end{equation}
where we used that the remaining polynomial $P_1$ has to be independent of $\re T_0$. 

It remains to be shown that \eqref{eq:volumelimit} also holds if $B_3$ is a trivial fibration $\Sigma \times \mathbb{P}^1$ for some surface $\Sigma$. In this case $|\mathcal{J}_1|=0$. The K\"ahler cone is now generated by the class $J_0=[\Sigma]$ and $J_i=[\mathbb{P}^1\times j_i]$, with $j_i$ the K\"ahler cone generators. The cone of curves giving rise to primitive EFT strings in this case agrees with the cone of movable curves and is similarly generated by 
\begin{equation}
    \text{Mov}_1(B_3) =\text{Cone}\left\langle J_{i_0}^2, J_0\cdot J_i \right\rangle \,,
\end{equation}
for some $i_0$. Notice that the first curve class has $q=0$ whereas the second class of curves has $q=1$ unless $j_i^2=0$. The cone of effective divisors, on the other hand, is generated by 
\let\mathcal\mathdutchcal
\begin{equation}
   \text{Eff}^1(B_3)=\text{Cone}\left\langle D_0 = J_0, D_i=\delta_{ij}\left(\mathbb{P}^1 \times \mathcal{c}^j\right)\right\rangle \,,
\end{equation}
with $\mathcal{c}^i$ the Mori cone generators of $\Sigma$. The volume of these effective cone generators are given by 
\let\mathcal\mathcaldefault
\begin{equation}
    \mathcal{V}_{D_0} = \frac{1}{2}\eta_{ij} v^i v^j\coma \mathcal{V}_{D_i}=\delta_{ij}v^0v^j\,,
\end{equation}
where $\eta_{ij}$ is the intersection pairing on $\Sigma$. We are interested in the limit $\mathcal{V}_{D_i}\rightarrow \infty$ while keeping all other volumes finite. This is achieved for $v^0\sim \lambda \rightarrow \infty$, $v^{j\neq i}\sim 1/\lambda$ and $v^i$ finite. Since $\eta_{ii}\neq 0$ by assumption $\mathcal{V}_{D_0}$ remains finite in this limit. We then have 
\begin{equation}
    \mathcal{V}_{B_3}^2 \rightarrow  (v^0 v^i)^2 (v^i)^2 \simeq \mathcal{V}_{D_i}^2 \mathcal{V}_{D_0}\,,
\end{equation}
in accordance with \eqref{eq:volumelimit}. We can thus also confirm \eqref{eq:volumelimit} for the case $q=1$. The case that $B_3$ allows for two surface fibrations corresponds to a $q=0$ case to which we turn now. 

\item[$\bf{q=0}$:] In this case $C$ can either be as in Condition \ref{list:EFT1new}, i.e., $C=J_0^2$ with $J_0^3=0$, or a curve as in Condition \ref{list:EFT2new}, i.e., $C=J_0\cdot J_1$ with $J_0^2=J_1^2=0$. In the latter case, the base $B_3$ has two independent surface fibrations with generic fibers $J_0$ and $J_1$ which intersect over a curve $C= J_0 \cdot J_1$. However, using Proposition 6 in \cite{Lee:2019tst}, both fibrations have to be trivial, and the base is the product $\mathbb P^1 \times \mathbb P^1 \times \mathbb P^1$, for which $\mathcal{V}_{B_3}^2$ is just a product of all three divisor volumes and hence \eqref{eq:volumelimit} is true in any limit. We can thus consider $C = J_0^2$ without losing generality. Notice that $q=0$ is the heterotic string case, and \eqref{eq:volumelimit} follows from the analysis of \cite{Lee:2019tst, Klaewer:2020lfg}. To recover this result, let us first borrow the classification of the K\"ahler cone generators other than $J_0$:\footnote{In \cite{Lee:2019tst, Klaewer:2020lfg}, $\mathcal{K}_{1,2}$ are called $\mathcal{I}_{1,3}$, respectively.}
\begin{equation}
\begin{split}
    J_\mu \cdot J_0^2&\neq 0 \quad \text{if} \quad \mu \in \mathcal{K}_1\,,\\
    J_r\cdot J_0^2=0 \quad \text{and} \quad J_r\cdot J_s \cdot J_0&=0 \quad \text{if}\quad r,s\in \mathcal{K}_2\,. 
\end{split}
\end{equation}
We then observe that the volume of $D_0$ is given by
\begin{equation}
    \mathcal{V}_{D_0} \sim \frac{1}{2} (v^0)^2 + n_r v^r v^0 +\frac{1}{2} n_{rs} v^r v^s\,,
\end{equation}
for $D_0$ the generator of $\text{Eff}^1(B_3)$ dual to $C$, and $r,s \in \mathcal{K}_2$ in the language of \cite{Lee:2019tst}. On the other hand, we have 
\begin{equation}
    \mathcal{V}_{J_0} = \kappa_{00\mu}v^\mu \left(v^0 + n_r v^r\right)\,,
\end{equation}
for $\mu \in \mathcal{K}_1$. Since there is no term proportional to $(v^0)^2$ it follows that 
\begin{equation}
    J_0=\sum_{i\neq 0} a^i D_i \,,
\end{equation}
for some $a^i$ and $D_i\cdot C=0$.  Since in the EFT string limit all $\mathcal{V}_{D_{i\neq 0}}$ need to remain finite, we need to have $\kappa_{00\mu}v^\mu\rightarrow 0$, which is precisely the heterotic emergent string limit \cite{Lee:2019tst,Klaewer:2020lfg}. We then have 
\begin{equation}
    \mathcal{V}_{B_3}^2 \sim \underbrace{\left(\frac{1}{2} (v^0)^2 + n_r v^r v^0 +\frac{1}{2} n_{rs} v^r v^s\right)}_{=\re T_0} \underbrace{\left(\frac{1}{2} (v^0)^2 + n_r v^r v^* +\frac{1}{2} n_{rs} v^r v^s\right) (\kappa_{00\mu} v^\mu)^2}_{=P_2(\re T_{i\neq 0})}\,.
\end{equation}
 Notice that one might have been tempted to identify the first two factors as $\mathcal{V}_{D_0}$. However, the last term goes to zero and therefore cannot correspond to a polynomial in the remaining $\mathcal{V}_{D_i}$ which are assumed to be constant in the limit. We thus can also confirm \eqref{eq:volumelimit} for the case $q=0$. 
\qedhere
\end{enumerate}
\end{proof}

\begin{proof}[Proof of Proposition \ref{prop:nonprimitive}]
There are two ways in which the weak coupling limit for the gauge theory on $D_0$ can fail to be describable as a quasi-primitive EFT string limit:  $i)$ the weak coupling limit is an EFT string limit but not quasi-primitive or $ii)$ the weak coupling limit is not even an EFT string limit according to Definition \ref{def:EFTlimit}. We can further differentiate two cases: the volume of the divisor $D_0$ contains a term $(v^a)^2$ quadratic in some curve volume $v^a$ which scales to infinity, or it does not. Let us start with the first case, i.e., assume that the volume of $D_0$ contains a term 
\begin{align}
    \mathcal{V}_{D_0} \sim (v^0)^2 +\ldots\,. 
\end{align}
If the limit $\mathcal{V}_{D_0}\rightarrow \infty$ cannot be realized as a quasi-primitive EFT string, this means that $D_0$ or some other generators of \text{our basis of} $\text{Eff}^1(B_3)$ contain also terms linear in
$v^0$, but every co-scaling (performed in order that the scaling linear in $\lambda$ for $v^0\sim \lambda\rightarrow \infty$ is compensated) blows up 
 some other $v^1\sim \lambda$ and induces a large volume limit for some other generator in our basis of $\text{Eff}^1(B_3)$, or scales the volume of a divisor to zero so that we lose perturbative control.

 If $J_0^3\neq 0$, the volume of the base scales as $\mathcal{V}_{B_3}\sim \lambda^3 \to \infty$ such that the limit corresponds to an effective decompactification to 10d and any gauge theory on $D_0$ reduces to an 8d defect in 10d -- irrespective of whether the limit is an EFT limit or not. However, there is the possibility that $J_0^3= 0$ but still no (quasi-) primitive EFT string limit exists.\footnote{In Appendix \ref{app:examples} we discuss an example of this kind where any co-scaling necessary to get a set of homogeneously expanding divisors forces some generator of $\text{Eff}^1(B_3)$ to shrink.} In this case $\mathcal{V}_{B_3}\sim \lambda^2\sim \mathcal{V}_{D_0}$ and hence we decompactify effectively to 8d with the gauge theory corresponding to a non-weakly coupled gauge theory in 8d coupled to gravity as opposed to a defect theory. Thus, the weak coupling limits for gauge theories on $D_0$ obtained by blowing up $\mathcal{V}_{D_0}\sim \lambda^2$ and not corresponding to a quasi-primitive EFT string limit or not to an EFT string limit at all are either limits in which we obtain 8d defects in a 10d gravitational bulk theory or non-weakly coupled gauge theories in 8d coupled to gravity.

There is an alternative way to engineer a weak coupling limit for a gauge theory on $D_0$ provided $\mathcal{V}_{D_0}$ contains a term 
\begin{align}
    \mathcal{V}_{D_0} \sim v^0 v^1 + \ldots \,,
\end{align}
but no term quadratic in $v^0$ and/or $v^1$. Without loss of generality, let us assume that there is no such term for $v^1$. We can then try to reach the weak coupling limit for a gauge theory on $D_0$ by scaling $v^1\sim \lambda\rightarrow \infty$ without co-scaling $v^0$.\footnote{In case we also scale $v^0\sim \lambda$ and $v^0$ appears quadratically in the volume of some generator in our basis of $\text{Eff}^1(B_3)$ we can apply the logic of the previous case to this generator. In case $J_0^2=0$, $v^0$ only appears at most linearly in all relevant volumes, and we have a decompactification to a non-weakly coupled gauge theory in 8d since for $J_0^2=0$, we have to leading order $\mathcal{V}_{B_3}\sim \text{const. } v^1 v^0 \sim \mathcal{V}_{D_0}$ and if  both $v^1 \sim \lambda$ and $v^0  \sim \lambda$ are scaled up, this leads to a quadratic expansion
$\mathcal{V}_{B_3}\sim\mathcal{V}_{D_0} \sim \lambda^2$.
} We can now differentiate the three cases $J_1^2=0$, $J_1^2\neq 0$ but $J_1^3=0$ and $J_1^3\neq 0$. 
\begin{enumerate}[align=left,wide,  labelindent=0pt]
    \item[$\bf{J_1^2=0}$:] In this case $v^1$ appears at most linearly in the volume for all generators of $\text{Eff}^1(B_3)$. Thus, the limit $v^1\sim \lambda \rightarrow \infty$ is automatically a quasi-primitive EFT string limit for a $q=1$ string unless we super-impose it with another limit, as discussed in Section \ref{ssec:nonEFTgaugetheory}.

    \item[$\bf J_1^2\neq 0$, $\bf J_1^3=0$:]

    The proof of Proposition \ref{prop:q0} shows that for such intersection numbers, there exists a $q=0$ primitive EFT string limit for a D3-brane on the curve $J_1^2$ which takes $v^1\sim \lambda\rightarrow \infty$. In this $q=0$ limit,
     since $J_1^2\neq 0$, there exists a divisor whose volume scales to infinity as $(v^1)^2 \sim \lambda^2$.
    Since by assumption, $\mathcal{V}_{D_0}$ only scales linearly in $v^1\sim \lambda$ and we do not co-scale $v^0$, the weak coupling limit for $D_0$ under consideration  is certainly not the $q=0$ primitive EFT string limit associated with $J_1^2$.
    
    First, we observe that the weak coupling limit 
    where $\mathcal{V}_{D_0} \sim v^0 v^1$ linearly in $v^1$ can only exist if
    $J_1^2\cdot J_0\neq 0$: Otherwise, it would follow from the intersection numbers that   $J_0\cdot J_1 \sim J_1^2$,  and hence any scaling with $v^1 \sim \lambda$ is necessarily quadratic in $v^1$. If this were
    the case, we would just be considering the primitive EFT string limit associated with the $q=0$ string on the curve $J_1^2$, contrary to our assumption.
    
    Now, as noted already, the difference between the weak coupling with linear scaling in $v^1$ compared to the strict $q=0$ primitive EFT string limit is that we do not impose a co-scaling on $v^0$. Combined with $J_1^2\cdot J_0\neq 0$, this implies that the volume of $B_3$ behaves as 
    \begin{align}
        \mathcal{V}_{B_3} \sim (v^1)^2 v^0 +\ldots \sim \lambda^2\,.
    \end{align}
    As $J_1^2\cdot J_0\neq 0$, the tension of the string on $J_1^2$ has a contribution proportional to $v^0 \sim {\cal O}(1)$. The string scale for the $q=0$ string on the curve $J_1^2$ is therefore of order of $M_\text{\tiny IIB}$. The scaling of $\mathcal{V}_{B_3}$ identifies the limit as a decompactification limit to 8d rather than an emergent string limit. And indeed, using \eqref{eq:Lambdasp}, we find that the species scale of the KK-tower signaling the decompactification to 8d agrees with $M_\text{\tiny IIB}$. Notice that since we have a K\"ahler cone generator with $J_1^3=0$, the manifold $B_3$ needs to admit a rational/genus-one fibration. The decompactification to 8d corresponds to the limit where the base of this fibration becomes large. Now, since a curve $C^0$ in the movable cone with non-zero intersection with $D_0$ cannot be proportional to $J_1^2$, it cannot be just the fiber of the $p$-projection. Hence, in the 8d limit, we resolve the internal directions of the D3-brane wrapping $C^0$. Similarly, since $D_0\cdot J_1^2=0$ the divisor $D_0$ cannot contain the exceptional section of the fibration $p: B_3\rightarrow B_2$. Instead, $D_0$ has to be a combination of $p$-vertical or exceptional divisors. In the 8d theory, the gauge theory on the 7-brane stack wrapping $D_0$ is thus a defect theory in 8d.
    \item[$\bf J_1^3\neq0$:] In this case the volume of $B_3$ again scales as $\mathcal{V}_{B_3}\sim \lambda^3$ indicating a decompactification to 10d as in Section \ref{ssec:stringscalevsSC}. \qedhere
\end{enumerate}
\end{proof}

\section{Examples}
\label{app:examples}

In this appendix, we provide examples for F-theory bases $B_3$ for which we analyze the possible (quasi-)primitive EFT string limits. These example serve to illustrate the subtleties of EFT string limits in the F-theory K\"ahler field space in concrete setups. Based on the example discussed in Section \ref{ssec:P1fibFn}, we explore different blow-ups of $B_3= \PP^1\rightarrow \FF_n$: first in Section \ref{app:toricBlF2} the base $B_2$ is replaced by a blow-up of $\FF_n$ and second in Section \ref{app:toricdatacurveblowup} an example is discussed for which the fiber of $B_3=\mathbb{P}_1\rightarrow \FF_1$ is blown-up. Notice that in all the following examples the K\"ahler cone is simplicial.

\subsection{\texorpdfstring{$\mathbb{P}^1$}{}-fibrations over \texorpdfstring{$\text{Bl}(\mathbb{F}_2)$}{}}
\label{app:toricBlF2}

Let us consider a $\PP^1$-fibration over a base $B_2$ that is $\FF_n$ blown-up in a smooth point on its base $\mathbb{P}^1$. We denote the resulting base by $B_2 = \text{Bl}(\FF_n)$. Due to the blow-up the K\"ahler cone of $\text{Bl}(\FF_n)$ has an additional generator compared to $\FF_n$ which we call $j_2$. The resulting intersection polynomial is
\begin{equation}
    \mathcal{I}(B_2) = nj_0^2 +j_0\cdot j_1+ (n-1) j_2^2 + j_2 \cdot j_1 +nj_0\cdot j_2\fstop
\end{equation}
 The twist of the $\PP^1$-fibration is encapsulated in a line bundle $\mathcal{T}$ with
\begin{equation}
    c_1 (\mathcal{T}) = s j_0 + t j_1 + u j_2\coma s,\, t, \, u \geq 0\coma
\end{equation}
where, compared to \eqref{ex1:twist}, we also allowed for a twist depending on the exceptional divisor in $B_2$. The K\"ahler cone generators of $B_3 = \PP^1 \stackrel{\mathcal{T}}{\rightarrow} \text{Bl}(\FF_n)$ are given by
\begin{equation}
    J_0  = p^*j_0 \coma J_1 = p^*j_1 \coma J_2 = p^*j_2 \coma J_3 = S_- + p^*c_1(\mathcal{T})\coma
    \label{eq:KcP1fibBlFn}
\end{equation}
where $S_-$ is the zero section of the fibration $p :\, \PP^1 \rightarrow B_2$. The intersection ring for $B_3$ reads
\begin{equation}
    \begin{split}
        \mathcal{I}(B_3) =&  \left(s^2 n+2 s t+2 s u n+2 t u+u^2 (n-1)\right)J_3^3 \,+\\
   & +(s n+t+u n)J_0 \cdot J_3^2 +(s n+t+u (n-1))J_2 \cdot J_3^2
   +(s+u)J_1 \cdot J_3^2 \,+ \\
   & +nJ_0^2 \cdot J_3 +(n-1)J_2^2 \cdot J_3 +nJ_0\cdot J_2\cdot J_3 +J_0\cdot J_1\cdot
   J_3+J_1\cdot J_2\cdot J_3\fstop
    \end{split}
\end{equation}
Notice that for $J_2=u=0$, one obtains the example discussed in Section \ref{ssec:P1fibFn}. In what follows, we discuss toric constructions for these geometries when $n=2$. 

An explicit  realization for $B_2  = \text{Bl}(\mathbb{F}_2)$ is encoded by the toric data
\let\mathcal\mathdutchcal
\begin{align}
\begin{blockarray}{crrrrr}
	&&&\mathcal{c}^0&\mathcal{c}^1&\mathcal{c}^2\\
\begin{block}{c(rr|rrr)}
    d_1& 0 & -1& 1 & -2& 0 \\
    d_2& -1& -1& 1 & 0 & -1\\
    d_3& -1& -2& -1& 1 & 1 \\
	d_4& 1 & 0 & 0 & 1 & 0 \\
	d_5& 0 & 1 & 0 & 0 & 1 \\
\end{block}
\end{blockarray}\text{\hspace{0.2cm}.}
	\label{eqn:f10toricdata}
\end{align}
Here $d_a$ and $\mathcal{c}^a$ denote the toric divisors and Mori cone generators for $B_2$ respectively, where the latter intersect with the K\"ahler cone generators $\{j_a\}_{a=0,1,2}$ in $B_2$  as  $\mathcal{c}^a \cdot j_b = \delta^a_b$. 
Notice that the toric divisor $d_2$ is the exceptional divisor associated with the blow-up of $\mathbb{F}_2$. Moreover, we obtain that the first Chern class is $c_1(B_2) = j_0 + j_2$.

\let\mathcal\mathcaldefault
 Now, we construct smooth $\PP^1$-fibrations determined by the twist bundle $c_1(\mathcal{T}) = s j_0 + u j_2$ 
for the toric projective bundle $B_3 = \mathbb{P}(\mathcal{O}_{B_2} \oplus \mathcal{O}_{B_2}(s j_0 + u j_2))  \rightarrow B_2$. 
 Under this consideration we obtain the following toric data for the choices of parameters $s,u \in  \{0,1\}$:
\begin{align}
\begin{blockarray}{crrrrrrr}
	&&&&\mathcal{C}^0&\mathcal{C}^1 & \mathcal{C}^2 & \mathcal{C}^3\\
\begin{block}{c(rrr|rrrr)}
    D_0& 0 & 0 & 1& -s& 0 & -u& 1 \\
    D_1& 0 & -1& s & 1 & -2& 0 & 0 \\
    D_2& -1& -1& 0 & 1 & 0 & -1& 0 \\
    D_3& -1& -2& 0 & -1& 1 & 1 & 0 \\
	D_4& 1 & 0 & 2s& 0 & 1 & 0 & 0 \\
	D_5& 0 & 1 & u & 0 & 0 & 1 & 0 \\
    D_6& 0 & 0 &-1 & 0 & 0 & 0 & 1 \\
\end{block}
\end{blockarray} \text{\hspace{0.2cm}.}
	\label{eqn:ToricP1overBlF2}
\end{align}
Here the anticanonical class reads
\begin{equation}
    \overline{K}(B_3) = 2 J_3 + (1-s) J_0 + (1-u) J_2\,. 
\end{equation}
The $P\vert Q$ matrix \eqref{eqn:ToricP1overBlF2} determines the linear equivalence relations among toric divisors, which read
\begin{align}
\begin{split}
    &D_4 \sim D_2 + D_3\coma D_5 \sim D_1 +D_2+2D_3  \coma \\ 
    &D_6 \sim D_0+(s+u)D_1+(2s+u) D_2 + 2(s+u)D_3\,. 
\end{split}
     \label{eqn:LinEqs1}
\end{align}
As each  $D_{\rho}$ toric divisor is prime, an effective divisor in $B_3$ has the form $\sum_{\rho \in \Sigma(1)} a_\rho [D_\rho]$ with all $a_\rho \geq0$, where $\Sigma(1)$ is the set of rays that span the fan $\Sigma$ for $B_3$. 
From \eqref{eqn:LinEqs1}, we determined $[D_4]$, $[D_5],$ and $[D_6]$ to be positive linear combinations of $\{[D_i]\}_{i=0,1,2,3}$. 
Hence, the latter set is a minimal basis that spans $ N^1(B_3)$ and the cone of
effective divisors is  $\text{Eff}^1(B_3)_{\mathbb{Z}} = \text{Cone} (\{[D_i]\})$. 
Taking into account the basis of K\"ahler divisors $J_i\cdot \mathcal{C}^j = \delta_i^j$, we obtain the expressions
\begin{align}
\begin{split}
 [D_0] = J_3 -sJ_0 -uJ_2\coma [D_1] = J_0-2J_1\coma [D_2] = J_1 - J_2 \coma [D_3] = J_2 + J_3 -J_0\,. 
 \label{eqn:effDivisorsBlF2}
 \end{split}
\end{align}
By abuse of notation, we drop the $[\cdot]$ symbol from now on and refer to each effective divisor $[D_i]$ by $D_i$.  

Before proceeding to discuss the physics for $B_3$, let us remark that we can perform the blow-down $\text{Bl}(\mathbb{F}_2) \rightarrow \mathbb{F}_2$ by removing the lattice point in the toric data that is associated to the exceptional divisor $D_2$, and then, one proceeds with the same computation we used to obtain \eqref{eqn:effDivisorsBlF2}. In this way, we can verify expression \eqref{eq:EffcondivP1Fn-Ka} in our constructed examples.\footnote{Upon blow-down the divisor class $D_0$ maps to $J_3 -s J_0$, while $D_3$ maps to $J_1$.} The same method can be applied for other values $n\neq 2$ for $\mathbb{F}_n$.  

Different twist parameters $s, u\in \{0,1\}$ in~(\ref{eqn:ToricP1overBlF2}) give rise to inequivalent polytopes with a different number of triangulations. For concreteness, in the following discussion and subsections, we focus on the choice $s=1$ and $u=0$. 
In this case, the depicted triangulation in~(\ref{eqn:ToricP1overBlF2}) realizes a  $\PP^1$-fibration over Bl($\FF_2$) with twist bundle  $c_1(\mathcal{T})=j_0$. The cone of effective divisors is  
\begin{align}
\label{eqn:Effabc}
   \text{Eff}^1(B_3)& =  \text{Cone}\left\langle D_0,D_1,D_2,D_3\right\rangle\coma
\end{align}
whose expression in terms of the  K\"ahler cone generators of $B_3$ follows from  \eqref{eqn:effDivisorsBlF2}. 
The volumes of the generators are  given by 
\begin{equation}\label{ex2:divisorvolumes}
\begin{aligned}
    \mathcal{V}_{D_0} &= v^0\left(v^0+v^1\right)+v^2 \left(2v^0+ v^1+ \frac12 v^2\right)\,,\\
    \mathcal{V}_{D_1}&=v^3 v^1\,,\\
    \mathcal{V}_{D_2}&=v^3 v^2\,,\\
    \mathcal{V}_{D_3}&= \frac{1}{2} v^3 \left(2 v^0+ v^3\right)\,.
\end{aligned}
\end{equation}
 We can easily identify the generators of the movable cone $\text{Mov}_1(B_3)=\overline{\text{Eff}}^1(B_3)^\vee$ as
\begin{equation}
 C^0 = J_0\cdot J_1\coma C^1= J_1\cdot J_3\coma C^2= J_2\cdot J_3\coma C^3= J_0 \cdot J_3 \,,
\end{equation}
with volumes
\begin{equation}
   \begin{split}
        \mathcal{V}_{C^0} & = v^3\coma \\
         \mathcal{V}_{C^1} & = v^0+v^2+ v^3\coma\\
        \mathcal{V}_{C^2} & = 2 v^0+v^1+ v^2+2 v^3\coma\\
        \mathcal{V}_{C^3} & = 2 v^0+v^1+2 v^2+2v^3 \fstop
    \end{split}
    \end{equation}
Note that $J_0^2 = 2 J_0\cdot J_1$ and $J_3^2= J_0\cdot J_3$.
\begin{figure}[!htp]
    \centering
      \begin{tikzpicture}[baseline]
    \draw[line width = 1pt] (-8,6) -- (-8,2) -- (-3.5,2) -- (-3.5,6) -- (-8,6);
    \draw[line width = 1pt] (-8,0) -- (-8,-4) -- (-3.5,-4) -- (-3.5,0) -- (-8,0);
   \draw[line width = 1pt] (8,6) -- (8,2) -- (3.5,2) -- (3.5,6) -- (8,6);
    \draw[line width = 1pt] (8,0) -- (8,-4) -- (3.5,-4) -- (3.5,0) -- (8,0);
    \draw[line width = 1pt] (-3,8) -- (-3,-6) -- (3,-6) -- (3,8) -- (-3,8);
    %
    \shade[ball color = prcolor!50!white, opacity = 0.9] (0,6.2) circle (1.6);
    \draw[dashed] (-1.6,6.2) arc (180:0:1.6 and 0.5);
    \draw[dashed] (-1.6,6.2) arc (180:0:1.6 and -0.5);
    \node at (0,6.2) {$\PP^1_f$};
    \node at (-2.5,6.2) {$B_3:$};
    \shade[ball color = prhigh!50!white, opacity = 0.3] (-1,1) circle (1);
    \draw[dashed] (-2,1) arc (180:0:1 and 0.4);
    \draw[dashed] (-2,1) arc (180:0:1 and -0.4);
    \node at (-1,1) {$\PP^1_A$};
    \shade[ball color = prcolor!50!white, opacity = 0.7] (1,1) circle (1);
    \draw[dashed] (0,1) arc (180:0:1 and 0.4);
    \draw[dashed] (0,1) arc (180:0:1 and -0.4);
    \node at (1,1) {$\PP^1_B$};
    \shade[ball color = prhigh!50!white, opacity = 0.5] (0,-4.2) circle (1.6);
    \draw[dashed] (-1.6,-4.2) arc (180:0:1.6 and 0.5);
    \draw[dashed] (-1.6,-4.2) arc (180:0:1.6 and -0.5);
    \node at (0,-4.2) {$\PP^1_b$};
     \draw[line width = 1.5pt,-Triangle,densely dashed] (0,4.2) -- (0,2.2);
     \draw[line width = 1.5pt,-Triangle,densely dashed] (0,-0.2) -- (0,-2.2);
    \shade[ball color = prhigh!50!white, opacity = 0.3] (-6,5.2) circle (0.65);
    \draw[dashed] (-6.65,5.2) arc (180:0:0.65 and 0.3);
    \draw[dashed] (-6.65,5.2) arc (180:0:0.65 and -0.3);
    \node at (-6.,5.2) {$\PP^1_A$};
    \shade[ball color = prcolor!50!white, opacity = 0.7] (-4.7,5.2) circle (0.65);
    \draw[dashed] (-5.35,5.2) arc (180:0:0.65 and 0.3);
    \draw[dashed] (-5.35,5.2) arc (180:0:0.65 and -0.3);
    \node at (-4.7,5.2) {$\PP^1_B$};
    \node at (-7.5,5.2) {$D_0:$};
    \shade[ball color = prhigh!50!white, opacity = 0.5] (-5.35,2.8) circle (0.65);
    \draw[dashed] (-6,2.8) arc (180:0:0.65 and 0.3);
    \draw[dashed] (-6,2.8) arc (180:0:0.65 and -0.3);
    \node at (-5.35,2.8) {$\PP^1_b$};
    \draw[line width = 1.5pt,-Triangle,densely dashed] (-5.35,4.5) -- (-5.35,3.6);
    \shade[ball color = prcolor!50!white, opacity = 0.9] (-5.35,-0.8) circle (0.65);
    \draw[dashed] (-6,-0.8) arc (180:0:0.65 and 0.3);
    \draw[dashed] (-6,-0.8) arc (180:0:0.65 and -0.3);
    \node at (-5.35,-0.8) {$\PP^1_f$};
    \node at (-7.5,-0.8) {$D_1:$};
    \shade[ball color = prhigh!50!white, opacity = 0.5] (-5.35,-3.2) circle (0.65);
    \draw[dashed] (-6,-3.2) arc (180:0:0.65 and 0.3);
    \draw[dashed] (-6,-3.2) arc (180:0:0.65 and -0.3);
    \node at (-5.35,-3.2) {$\PP^1_b$};
    \node at (-5.35,-2) {$\times$};
    \shade[ball color = prcolor!50!white, opacity = 0.9] (6.15,5.2) circle (0.65);
    \draw[dashed] (5.5,5.2) arc (180:0:0.65 and 0.3);
    \draw[dashed] (5.5,5.2) arc (180:0:0.65 and -0.3);
    \node at (6.15,5.2) {$\PP^1_f$};
    \node at (4,5.2) {$D_2:$};
    \shade[ball color = prcolor!50!white, opacity = 0.7] (6.15,2.8) circle (0.65);
    \draw[dashed] (5.5,2.8) arc (180:0:0.65 and 0.3);
    \draw[dashed] (5.5,2.8) arc (180:0:0.65 and -0.3);
    \node at (6.15,2.8) {$\PP^1_B$};
    \node at (6.15,4) {$\times$};
    \shade[ball color = prcolor!50!white, opacity = 0.9] (6.15,-0.8) circle (0.65);
    \draw[dashed] (5.5,-0.8) arc (180:0:0.65 and 0.3);
    \draw[dashed] (5.5,-0.8) arc (180:0:0.65 and -0.3);
    \node at (6.15,-0.8) {$\PP^1_f$};
    \node at (4,-0.8) {$D_3:$};
    \shade[ball color = prhigh!50!white, opacity = 0.3] (6.15,-3.2) circle (0.65);
    \draw[dashed] (5.5,-3.2) arc (180:0:0.65 and 0.3);
    \draw[dashed] (5.5,-3.2) arc (180:0:0.65 and -0.3);
    \node at (6.15,-3.2) {$\PP^1_A$};
    \draw[line width = 1.5pt,-Triangle,densely dashed] (6.15,-1.55) -- (6.15,-2.45);
    \end{tikzpicture}
    \caption{We show the base $B_3=\mathbb{P}^1\rightarrow \text{Bl}(\FF_n)$ (center) for $u=t=0$ as well as the topology of the divisors $D_i$, $i=0,\ldots,3$.}
    \label{fig:fibrationstructureP1BlFn}
\end{figure}
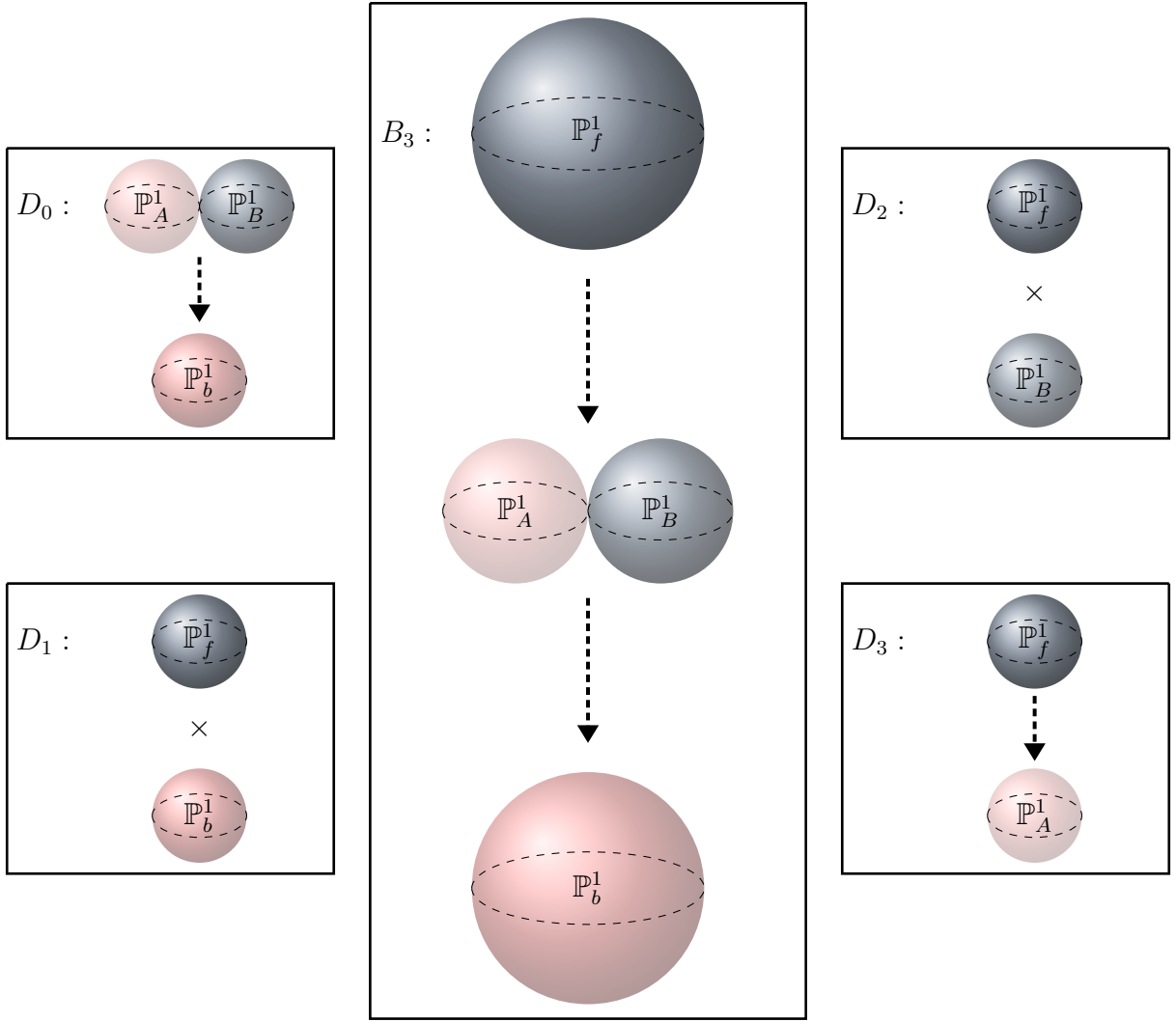 

Let us identify those curves giving rise to the (quasi)-primitive EFT string limits for the chosen $B_3$. Compared to the example in Section \ref{ssec:P1fibFn}, we now have the additional divisor $D_2$ and the associated generator $C^2$ of the movable cone. We first notice that, by Proposition \ref{prop:primitiveEFT}, the limit $\mathcal{V}_{D_2}\rightarrow \infty$ cannot be realized as a primitive EFT string limit within the given K\"ahler cone. To see this, consider Figure \ref{fig:fibrationstructureP1BlFn} where we show $B_3$ and the topology of the divisors. The divisors are expressed in terms of the fibral curve $\mathbb{P}^1_f$ of $B_3$, the base $\mathbb{P}^1_b$ of $\text{Bl}(\FF_2)$ and the two fibral curves $\mathbb{P}^1_A$ and $\mathbb{P}^1_B$ which together give the generic fiber of $\mathbb{F}_n$. Figure \ref{fig:fibrationstructureP1BlFn} illustrates Proposition \ref{prop:primitiveEFT} since the two curves $\mathbb{P}^1_f$ and $\mathbb{P}^1_B$ contained in $D_2$ are also contained in other generators of $\text{Eff}^1(B_3)$ as fibers of a non-trivial fibration. Hence, blowing up either of the two curves will also blow up either $D_0$ or $D_3$ such that there is no primitive EFT string limit associated to $\mathcal{V}_{D_2}\rightarrow \infty$.

\begin{figure}[!htp]
    \centering
   \begin{tikzpicture}[scale=1,rotate around y = 80, rotate around z = 15, rotate around x = 10]
     \draw[fill=black!50!white,opacity=0.2,draw=none] (0,0,0) -- (6,0,0) -- (6,0,6) -- node[black,pos=0.2,left,opacity=1] {$\mathcal{V}_{D_2} = \infty$} (0,0,6) -- (0,0,0);
     \draw[fill=prhigh!80!white,opacity=0.8,draw=none,left color=white, right color=prhigh!80!white, anchor=south] (0,0,0) -- (6,0,0) to[bend left=15] (5,2.5,5) to[bend right=20] (0,0,6) -- (0,0,0);
     \draw[line width=1pt,dashed,draw=black!60!white] (1,0,0) to[bend left=15] (1,0.5,1) to[bend right=20] (0,0,1);
     \draw[line width=1pt,dashed,draw=black!60!white] (2,0,0) to[bend left=15] (2,1,2) to[bend right=20] (0,0,2);
     \draw[line width=1pt,dashed,draw=black!60!white] (3,0,0) to[bend left=15] (3,1.5,3) to[bend right=20] (0,0,3);
     \draw[line width=1pt,dashed,draw=black!60!white] (4,0,0) to[bend left=15] (4,2,4) to[bend right=20] (0,0,4);
     \draw[line width=1.5pt] (0,0,0) -- (5,2.5,5);
      \draw[line width=1.5pt,-Triangle] (0,0,0)  -- node[above left,pos=1]
     {$\mathcal{V}_{D_0}^{-1}$} (6,0,0);
     \draw[line width=1.5pt,-Triangle] (0,0,0) -- node[above,pos=1] {$\mathcal{V}_{D_2}^{-1}$} (0,5,0);
     \draw[line width=1.5pt,-Triangle] (0,0,0) -- node[below,pos=1] {$\mathcal{V}_{D_3}^{-1}$} (0,0,6);
   \end{tikzpicture}
    \caption{A sketch of part of the K\"ahler field space for $B_3=\mathbb{P}^1\rightarrow \text{Bl}(\FF_n)$ with $u=0$.}
    \label{fig:sketch}
\end{figure}
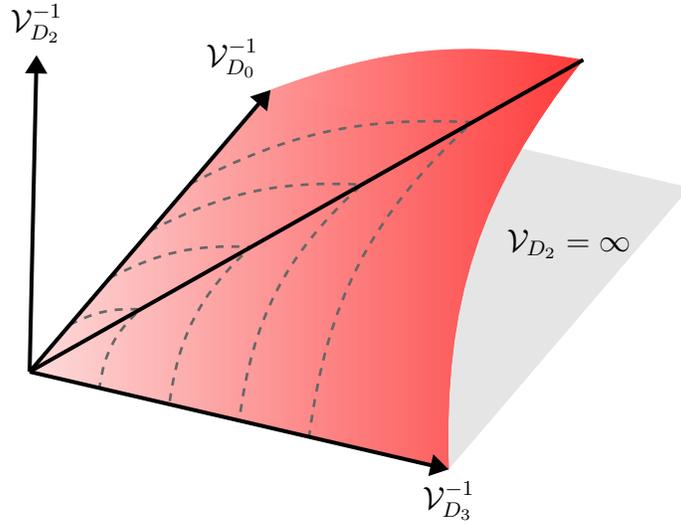

On the other hand, the limit $\mathcal{V}_{D_2}\rightarrow \infty$ cannot be realized as a quasi-primitive EFT string limit either. To see this, notice that we either need to blow-up $v^2\rightarrow \infty$ or $v^3\rightarrow \infty$. However, in both limits we have another generator of $\text{Eff}^1(B_3)$ blowing up at a faster rate such that it is not an EFT string limit. For $v^2\sim \lambda \rightarrow \infty$ one might try to co-scale $v^3\sim \lambda$ to obtain an EFT string limit as defined in Definition \ref{def:EFTlimit}. However, in this way $|\mathcal{I}|$ is not minimized such that $v^2\rightarrow \infty$ with this co-scaling does not qualify as a quasi-primitive EFT string limit according to Definition \ref{def:quasiprimitive}. Thus, in this example we only find three (quasi)-primitive EFT string limits, despite the dimension of the field space being four. In particular, a gauge theory on $D_2$ is of the kind discussed in Section~\ref{ssec:nonEFTgaugetheory}. Thus, the weak coupling limit for this gauge theory cannot necessarily be realized as an EFT string limit within the given chamber of the K\"ahler cone. This is illustrated in Figure \ref{fig:sketch}. There we show parts of the K\"ahler field space, including the directions corresponding to the volumes of the divisors $D_0,D_2$ and $D_3$ in the vicinity of $\mathcal{V}_{D_i}=\infty$. The red surface corresponds to the boundary of the chosen chamber of the K\"ahler cone embedded in the saxionic field space. From our previous discussion, it is clear that the bulk of the gray plane corresponding to $\mathcal{V}_{D_2}=\infty$ cannot be reached within this chamber of the K\"ahler cone chamber, but we need to perform at least one flop transition in the base $B_3$. The necessary flop transitions are discussed in the following. Notice, however, that the origin in Figure~\ref{fig:sketch} can be reached in a linear fashion corresponding to an EFT string obtained by a D3-brane on $a_0 C^0+a_2 C^2+a_3 C^3$ for some coefficients $a_i$. This is, however, not a quasi-primitive EFT string.  

To summarize, only the generators $D_i$ for $i=0,1,3$ are each individually homogeneously expandable. Hence, classically all three curves $C^0$, $C^1$ and $C^3$ are expected to give rise to primitive EFT strings with the following associated limits:
\begin{equation}
\renewcommand{\arraystretch}{1.25}
\begin{tabular}{c|c|l|c}
    \text{Movable curve} & \text{$q$ factor} & \multicolumn{1}{c|}{Primitive EFT limit} & ${\cal V}_{D_i} \to \infty$ \\
    \hline
  $C^0 = J_0\cdot J_1$ & $q=0$ & $v^0, v^1, v^2 \rightarrow \infty \coma v^3\rightarrow 0$ & $D_0$\\
    $C^1 = J_1\cdot J_3$ & $q=1$ & $v^1 \rightarrow \infty \coma v^0,v^2 \rightarrow 0 \coma v^3 \simeq \text{const.}$ & $D_1$\\
    $C^3 = J_0\cdot J_3$ & $q=2$ & $v^3 \rightarrow \infty \coma v^1,v^2\rightarrow 0 \coma v^0 \simeq \text{const.}$ & $D_3$
\end{tabular}
\end{equation}
However, in the limit associated to $C^1$, we effectively blow-down $D_2$. Hence, this limit might be obstructed at the quantum level. On the other hand, co-scaling $v^3\rightarrow \infty$ would increase $|\mathcal{I}|$ and therefore the EFT string limit would not be quasi-primitive anymore. Thus, even though all curves $C^i$ with $i=0,1,3$ are as in Proposition \ref{prop:EFTstring} not all of them are associated to quasi-primitive EFT string limits. In total, we only have two primitive EFT strings, despite the dimension of the field space being four.

The associated polytope of this example -- with $s=1$ and $u=0$ -- allows for two additional triangulations  besides \eqref{eqn:ToricP1overBlF2} that  correspond to different chambers of the extended K\"ahler cone  $\mathcal{K}_{\text{ext}}(B_3)$.  Thus, $B_3$ relates via flop transitions to geometries that we  call $B_3'$ and $B_3''$. In terms of  toric data, we obtain
 \begin{align}
\begin{blockarray}{crrrrrrrrrrrl}
	&&&&\mathcal{C}_{\mathrm{I}}^1 & \mathcal{C}_{\mathrm{I}}^2 &\mathcal{C}_{\mathrm{I}}^3& \mathcal{C}_{\mathrm{I}}^4 & \mathcal{C}_{\mathrm{II}}^2 & \mathcal{C}_{\mathrm{II}}^3 &\mathcal{C}_{\mathrm{II}}^4 & \mathcal{C}_{\mathrm{II}}^5  \\
\begin{block}{c(rrr|rrrr|rrrr)l}
         D_0& 0 &  0 & 1 & -1& -1& 0 & 1 & -2& 0 &  0  & 1 \\
         D_1& 0 &  -1& 1 & -1& 1 & 1 & -1 & 0 & 1 &  -2 & 1\\
         D_2& -1&  -1& 0 & 1 & 0 & 1 & -1 & 1 & 1 &  0  & -1\\
         D_3& -1&  -2& 0 & 0 & 0 & -1& 1  & 0 & -1&  1  & 0 \\
         D_4& 1 &  0 & 2 & 1 & 0 & 0 & 0  & 1 & 0 &  1  & -1\\ 
         D_5& 0 &  1 & 0 & 0 & 1 & 0 & 0  & 1 & 0 &  0  & 0 \\ 
         D_6& 0 &  0 & -1& 0 & 0 & 1 & 0  & 0 & 1 &  0  & 0 \\
\end{block}
\end{blockarray}\text{\hspace{0.2cm}.}
	\label{eqn:f10toricdata-flops}
\end{align}
Here $\mathcal{C}_{\mathrm{I}}^i$ are the Mori cone generators of $B_3'$, while $\mathcal{C}_{\mathrm{II}}^i$ are the Mori cone generators of $B_3''$. 
More explicitly, the flop map can be realized via the following curves' basis transformation: 
\begin{equation}
\text{Flop I: }
\left\{ \renewcommand{\arraystretch}{1.1}\begin{array}{lcl}
\mathcal{C}_{\mathrm{I}}^1 &=& \mathcal{C}^0 + \mathcal{C}^1   \\
\mathcal{C}_{\mathrm{I}}^2 &=& \mathcal{C}^0 + \mathcal{C}^2 \\
\mathcal{C}_{\mathrm{I}}^3 &=& \mathcal{C}^0 + \mathcal{C}^3  \\
\mathcal{C}^4_{\mathrm{I}} &=& -\mathcal{C}^0
\end{array}\right.
\coma
\text{Flop II: }
\left\{ \renewcommand{\arraystretch}{1.1}\begin{array}{lcccl}
\mathcal{C}_{\mathrm{II}}^2 & = & \mathcal{C}_{\mathrm{I}}^1+\mathcal{C}_{\mathrm{I}}^2 & = & 2\mathcal{C}^0 + \mathcal{C}^1 + \mathcal{C}^2 \\
\mathcal{C}_{\mathrm{II}}^3 & = & \mathcal{C}_{\mathrm{I}}^3                            & = & \mathcal{C}^0 + \mathcal{C}^3 \\
\mathcal{C}_{\mathrm{II}}^4 & = & \mathcal{C}_{\mathrm{I}}^1+\mathcal{C}_{\mathrm{I}}^4 & = & \mathcal{C}^1  \\
\mathcal{C}_{\mathrm{II}}^5 & = & -\mathcal{C}_{\mathrm{I}}^1                           & = & -\mathcal{C}^0-\mathcal{C}^1
\end{array}\right.
\,.
\end{equation}
 Nevertheless, notice that the linear equivalence relations among toric divisors are preserved, i.e.,
\begin{equation}
D_4 \sim  D_2+D_3  \coma D_5 \sim D_1 +D_2+ 2 D_3   \coma D_6 \sim D_0 +D_1+2D_2 +2D_3 \,, 
\end{equation}
which implies that, again, the effective cones $\text{Eff}^1(B_3')$ and $\text{Eff}^1(B_3'')$ are spanned by the basis $\{D_i\}_{i=0,1,2,3}$ in each respective case.  

\subsubsection{Primitive EFT string limit of $C^1$}

In order to find the primitive EFT string limits for the other two generators of $\text{Mov}_1(B_3)$ we first perform a flop transition on the base $B_3$ to $B'_3$. The geometry now corresponds to a $\PP^1$-fibration over $\mathbb{F}_1$ with twist bundle $c_1(\mathcal{T})=j_0+j_1$ for which we blow-up the fiber over a point in $B_2$ leading to an additional fibral curve over $B_2$. The flop can thus be understood by blowing-down a fibral curve of $\text{Bl}(\FF_2)$ and replacing it by a fibral curve in the $\PP^1$-fiber over $\mathbb{F}_1$. The twist of the original curve over $\mathbb{P}^1_b$ is reflected now in the change of the twist bundle $c_1(\mathcal{T})=j_0\rightarrow j_0+j_1$. The cone of effective divisors is spanned again by \eqref{eqn:Effabc}, but now expressed in terms of the K\"ahler cone generators of $B'_3$ as 
\begin{equation}
D_0 = J_4-J_2 - J_1 \coma D_1  = J_2+J_3-J_4-J_1 \coma D_2 = J_1 + J_3 -J_4 \coma D_3 = J_4 - J_3 \,.  
\end{equation}
Here we introduced the new generators $\{J_i\}_{i=1,2,3,4}$ of the K\"ahler cone to which we associate the volume $v^i$ of the dual generators of the Mori cone $\{\mathcal{C}_{\rm I}^i\}_{i=1,2,3,4}$.

The intersection ring compatible with the fibration structure of $ B'_3$ and $B_3$ is
\begin{align}
\begin{split}
I(B'_3) = &2 J_3^3+2 J_2\cdot J_3^2+J_1\cdot J_3^2+3 J_4\cdot J_3^2+J_2^2\cdot J_3+3 J_4^2\cdot J_3+J_2\cdot J_1\cdot J_3\,+\\
&+2 J_2\cdot J_4 \cdot J_3+J_1 \cdot J_4\cdot J_3+3 J_4^3+2 J_2\cdot J_4^2+J_1\cdot J_4^2\,+\\
&+J_2^2\cdot J_4+J_2\cdot J_1\cdot J_4\,.
\end{split}
\end{align} 
The volume for the effective divisors read
\begin{align}
\begin{split}
\mathcal{V}_{D_0}&=  \frac{1}{2}v^2( v^2 + v^1)  \,, \\
\mathcal{V}_{D_1}&  = \frac{1}{2}(v^4)^2 + v^3 v^1 + v^4 v^3 + v^1 v^4\,,  \\
\mathcal{V}_{D_2} & = \frac{1}{2} (v^4)^2 + v^3 v^2 + v^4 v^3+ v^2 v^4 \,,\\
\mathcal{V}_{D_3} &  = \frac{1}{2}(v^3)^2  \,.
\end{split}
\end{align}
The movable cone $\text{Mov}_1(B'_3) = \overline{\text{Eff}}^1(B'_3)^\vee$ is then spanned by the 
curves
\begin{align}
C^0 = J_1\cdot J_2 \coma C^1 = J_1\cdot J_3\coma C^2 = J_2\cdot J_4 \coma C^3 = J_3^2 \,.
\end{align}
The volume for these curves are
\begin{align}
\begin{split}
\mathcal{V}_{C^0 }&= v^3  + v^4 \,, \\
\mathcal{V}_{C^1 }&  = v^2 + v^3 + v^4 \,,  \\
\mathcal{V}_{C^2} & = v^1+v^2 + 2v^3+2v^4\,, \\
\mathcal{V}_{C^3} &  =  v^1 + 2v^2 +3 v^3+2v^4 \,.
\end{split}
\end{align}
Notice the splitting of $C^0 $ in terms of the Mori cone generators $\mathcal{C}^3$ and $\mathcal{C}^4$. 
We identify $C^0$ as the heterotic curve since $C^0 \cdot \bar{K}(B'_3) =2$. Using the adjunction formula, we obtain $g(C^3) =0$ where we used $\bar{K}(B'_3) = 2 J_3 + J_2$.   
Also, note that $C^3=J_1^2 $ with $J_1^3=0$ holds. 

In this chamber of the K\"ahler cone, we can now realize the EFT string limit for the D3-brane on $C^1$ as 
\begin{equation}
\renewcommand{\arraystretch}{1.25}
\begin{tabular}{c|c|l|c}
    \text{Movable curve} & \text{$q$ factor} & \multicolumn{1}{c|}{Primitive EFT limit} & ${\cal V}_{D_i} \to \infty$ \\
    \hline
  $C^1 = J_1\cdot J_3$ & $q=1$ & $v^1 \rightarrow \infty \coma v^2\rightarrow 0\coma v^3,v^4\simeq \text{const.}$ & $D_1$\\
\end{tabular}
\end{equation}
without shrinking any additional divisor. This is, however, the only primitive EFT string limit in this chamber that can be reached without leaving the perturbative regime. In order to realize the primitive EFT string limit for the D3-brane on $C^2$ we need to perform a second flop corresponding to yet another triangulation of the polytope.

\subsubsection{Primitive EFT string limit of $C^2$}

To reach the third chamber of the K\"ahler cone, we need to blow-down the curve $\mathcal{C}_{\rm I}^1$ in $B'_3$ and replace it by yet another fibral curve $\mathcal{C}_{\rm II}^5$. The resulting manifold is a $\PP^1$-fibration over $\PP^2$ for which the fiber splits into three curves $\mathcal{C}_{\rm II}^3$, $\mathcal{C}_{\rm II}^4$ and $\mathcal{C}_{\rm II}^5$.  We realize the resulting base $B''_3$ as a third triangulation for the polytope underlying also $B_3$ and $B'_3$.
The generators of the cone of effective divisors \eqref{eqn:Effabc} in terms of the K\"ahler cone generators of $B''_3$ are 
\begin{equation}
D_0 = J_5-2J_2 \coma D_1  = J_3+J_5-2J_4 \coma D_2 = J_3-J_5 + J_2\coma D_3 = J_4 - J_3 \,.  
\end{equation}
Here we take $\{J_i\}_{i=2,3,4,5}$ to be the K\"ahler cone generators dual to $\{\mathcal{C}_{\mathrm{II}}^i\}$. The intersection ring compatible with the fibration structure of $B''_3$ and $B_3$ reads
\begin{align}
\begin{split}
I(B''_3) = &2 J_3^3+4 J_5\cdot J_3^2+2 J_2\cdot J_3^2+3 J_4\cdot J_3^2+4 J_5^2\cdot J_3+J_2^2\cdot J_3+3 J_4^2\cdot J_3\,+\\
&+2 J_5\cdot J_2\cdot J_3+4 J_5\cdot J_4\cdot J_3+2 J_2\cdot J_4\cdot J_3+4 J_5^3+3 J_4^3+J_4\cdot J_2^2\,+\\
   &+4 J_5\cdot J_4^2+2 J_2\cdot J_4^2+2J_5^2\cdot J_2+4 J_5^2\cdot J_4+J_2^2\cdot J_4+2 J_5\cdot J_2\cdot J_4\fstop
\end{split} 
\end{align}
The volumes of the effective divisors are
\begin{align}
\begin{split}
\mathcal{V}_{D_0 }&=  \frac{1}{2}(v^2)^2 \,,  \\
\mathcal{V}_{D_1}&  = \frac{1}{2}(v^4)^2 + v^4 v^3\,,  \\
\mathcal{V}_{D_2} & = \frac{1}{2} (v^4)^2 +  2 v^3 v^5+(v^5)^2+v^3 v^2+v^5 v^2+v^3 v^4+2 v^5 v^4+v^2 v^4\,,\\
\mathcal{V}_{D_3} &  = \frac{1}{2}(v^3)^2   \,.
\end{split}
\end{align}
Moreover, here we have that  $\bar{K}(B''_3) = 2 J_3 + J_2$.  The movable cone $\text{Mov}_1(B''_3) = \overline{\text{Eff}}^1(B''_3)^\vee$ is then spanned by the 
curves $C^i$ given by
\begin{align}
C^0 = J_2^2 \coma C^1 = J_3\cdot \left(J_4-J_2\right)\coma C^2 = J_3\cdot J_2 \coma C^3 = J_3^2\,.
\end{align}
The volumes of such curves read
\begin{align}
\begin{split}
\mathcal{V}_{C^1 }&= v^3+v^4 + v^5 \,, \\
\mathcal{V}_{C^2} &  = 2v^3 + 4v^5 + 2v^2 +3 v^4\,, \\
\mathcal{V}_{C^3 }&  = v^3 + 2v^5 + v^2+v^4 \,,  \\
\mathcal{V}_{C^4} & = 2v^3+2v^5+v^2 + 2v^4\,.
\end{split}
\end{align}
In this case we again have a single primitive EFT string limit corresponding to 
\begin{equation}
\renewcommand{\arraystretch}{1.25}
\begin{tabular}{c|c|l|c}
    \text{Movable curve} & \text{$q$ factor} & \multicolumn{1}{c|}{Primitive EFT limit} & ${\cal V}_{D_i} \to \infty$ \\
    \hline
  $C^2 = J_2\cdot J_3$ & $q=2$ & $v^5 \rightarrow \infty\coma v^2,v^3,v^4\simeq \text{const.}$ & $D_2$ \\
\end{tabular}
\end{equation}
Note that in this chamber of the extended K\"ahler cone we $2J_3\cdot J_2 = J_5^2$ such that the curve $C^2$ is indeed of the form required by Proposition \ref{prop:EFTstring}. Even though also the curve $C^0$ and $C^3$ are of the required form, there are no primitive EFT string limits associated to D3-branes wrapped on these curves in this chamber of the K\"ahler cone, as we would leave the perturbative regime while imposing such a limit.

\subsection{Curve blow-up in a \texorpdfstring{$\PP^1$}{}-fibration over \texorpdfstring{$\FF_n$}{}}
\label{app:toricdatacurveblowup}

Another interesting example is blowing-up a curve in the base $\FF_n$ of geometries in Section \ref{ssec:P1fibFn}. Here, we will discuss the case that the blown-up curve is a generic fiber of $\FF_1$ with twist choices $s=1$ for $B_3 = \mathbb{P}^1\rightarrow \FF_1$. 
The toric data for such a blow-up $\mu:\tilde{B}_3\rightarrow B_3$ takes the form
\begin{align}
\begin{blockarray}{crrrrrrrrrrrl}
	&&&&\mathcal{C}^0 & \mathcal{C}^1 &\mathcal{C}^3& \mathcal{C}^4 & \mathcal{C}_{\mathrm{I}}^0 & \mathcal{C}_{\mathrm{I}}^3 &\mathcal{C}_{\mathrm{I}}^4 & \mathcal{C}_{\mathrm{I}}^5  \\
\begin{block}{c(rrr|rrrr|rrrr)l}
         D_0& 0 &  0 & 1 & -1& 0 & 1 & 0 & - 1& 1 &  0 & 0  \\
         D_1& 0 &  -1& 1 &  1&-1 & 0 & 0 &  1 & 0 &  -1& 1\\
         D_3& 1&   0& 1 & 0 & 0 & -1 & 1 & 0 & -1&  1 & 0\\
         D_4& 1&  0 & 0 & 0 & 1 & 1 & -1  &  0 & 1 &  0 & -1 \\
         D_5& 0 & 1 & 0 & 1 & 0 & 0 & 0  & 1 & 0 &  0 & 0 \\ 
         D_6& -1 & -1 & 0 & 0 & 1 & 0 & 0  &0 & 0 &  1 & -1\\ 
         D_7& 0 &  0 & -1& 0 & -1 & 0 & 1  & 0 & 0 &  0 & 1 \\
\end{block}
\end{blockarray}\text{\hspace{0.2cm}.}
	\label{eqn:blowcurve}
\end{align}
Here the polytope realizing $\tilde{B}_3$ has another triangulation, whose associated geometry we call $\tilde{B}_3'$.  In \eqref{eqn:blowcurve}, we denote by $\mathcal{C}^i$ and $\mathcal{C}_{\mathrm{I}}^i$ the Mori cone curves for $\tilde{B}_3$ and $\tilde{B}_3'$ respectively.  

In $\tilde{B}_3$, the toric divisor $D_4$ is the  exceptional divisor $E$ that shrinks to the $\PP^1$-fiber of $\FF_1$ upon blow-down.  Moreover, we notice the following linear equivalence relations among toric divisors,
\begin{equation}
D_5 \sim  D_1 +  D_3 + D_4\,,\quad D_6 \sim D_3 + D_4 \coma  D_7 \sim D_0 +D_1 +D_3\,, 
\end{equation}
which imply that the cone of effective divisors is spanned by the divisors 
\begin{equation}
\text{Eff}^1(\tilde{B}_3) = \text{Cone}\left\langle D_0, D_1, D_3, D_4\right\rangle\,,
\label{eq:Effabccurveblow}
\end{equation}
which can be expressed in terms of the   K\"ahler cone generators $\{ J_i\}_{i=0,1,3,4}$  of $\tilde{B}_3$ as
\begin{equation}\label{ex3:effectiveKahler}
D_0 = J_3 - J_0 \coma D_1  = J_0-J_1 \coma D_3 = J_4 - J_3 \coma  D_4 = J_3 - J_4 +J_1\,.  
\end{equation}
Before moving on, let us analyze the geometry in some detail, see also Figure \ref{fig:fibrationstructureBlP1F1}. Therefore, we start with a  $\PP^1$-fibration over $\mathbb{F}_1$ with twist given by the section $j_0$ of $\mathbb{F}_1$ satisfying $j_0^2=1$. The generic fiber of $\tilde{B}_3$ is given by $\mathcal{C}^3$, the base of $\FF_{1}$ by $\mathcal{C}^1$ and a generic fiber of $\FF_1$ by $\mathcal{C}^0$. The blow-up now replaces $\mathcal{C}^3$ by $\mathcal{C}^3+\mathcal{C}^4$ along the exceptional fiber $\mathcal{C}^0_p$ of $\mathbb{F}_1$ at one point in $p\in \mathcal{C}^1$. Schematically, this is shown in Figure \ref{fig:fibrationstructureBlP1F1}. The blow-up is done in such a way that the zero section $D_0$ of $B_3$ wraps $\mathcal{C}^4$ over $\mathcal{C}^0_p$. Therefore, the divisor $D_0$ should be a connected sum made up by a copy of $\FF_1$ and $\mathcal{C}_p^0\times \mathcal{C}^4$ glued together along $\mathcal{C}_p^0$. We therefore expect the volume of $D_0$ to be given by 
\begin{align}
    \mathcal{V}_{D_0} = \mathcal{V}_{\FF_1} + v^0 v^4\,. 
\end{align} 
We can think of the fibration as a surface fibration over $\mathcal{C}^1$ with one exceptional fiber. The volume of a non-exceptional surface fiber corresponds to $\mathcal{V}_{D_3}$ and is not affected by the blow-up and hence simply given by \begin{align}
    \mathcal{V}_{D_3}=\frac12 (v^3)^2 +v^3v^0\,. 
\end{align}
The exceptional divisor $D_4$ is contained in the exceptional surface fiber and its volume can be calculated by taking the volume as 
\begin{align}
  \mathcal{V}_{D_4} = \underbrace{\frac{1}{2} \left(v^3+v^4\right)^2 + v^0 \left(v^3+v^4\right)}_{=\mathcal{V}_{J_1}} - \underbrace{\left(\frac{1}{2} (v^3)^2 +v^0v^3\right)}_{=\mathcal{V}_{D_3}} = \frac12 (v^4)^2 +v^0 v^4+ v^4 v^3\,.
\end{align} 
Here, we notice that $J_1$ is the pull-back of the $\mathcal{C}^0$ to the full fibration. Finally, we have the volume associated to the pull-back $D_1$ of $\mathcal{C}^1$. Since $p\in \mathcal{C}^1$ the fiber over $\mathcal{C}^1$ will contain both $\mathcal{C}^3$ and $\mathcal{C}^4$. Since the original $\tilde{B}_3$ does not include a twist over $\mathcal{C}^1$ for $\mathcal{V}_{D_1}$ we need to subtract a term proportional $(v^3)^2$ such that we arrive at 
\begin{align}
    \mathcal{V}_{D_1} = v^1 \left(v^3+v^4\right) + \frac{1}{2} \left(v^3+v^4\right) -\frac12 (v^3)^2 \,.
\end{align}

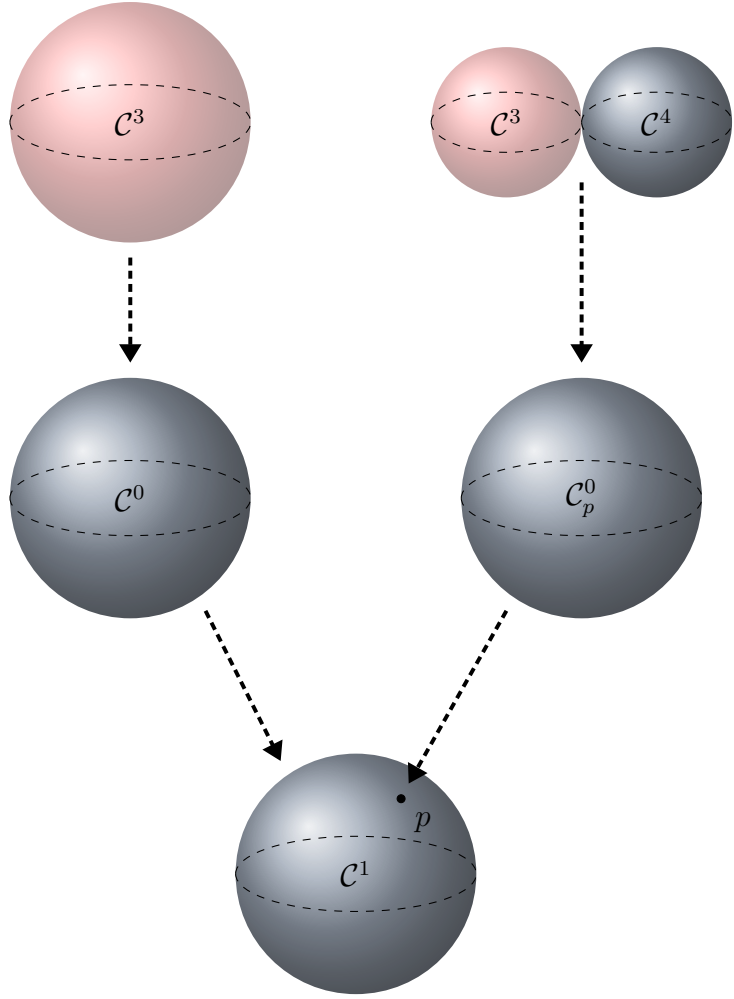
\begin{figure}[!htp]
    \centering
      \begin{tikzpicture}[baseline]
    \shade[ball color = prhigh!50!white, opacity = 0.5] (-3,6) circle (1.6);
    \draw[dashed] (-4.6,6) arc (180:0:1.6 and 0.5);
    \draw[dashed] (-4.6,6) arc (180:0:1.6 and -0.5);
    \node at (-3,6) {$\mathcal{C}^3$};
    \shade[ball color = prhigh!50!white, opacity = 0.5] (2,6) circle (1);
    \draw[dashed] (1,6) arc (180:0:1 and 0.4);
    \draw[dashed] (1,6) arc (180:0:1 and -0.4);
    \node at (2,6) {$\mathcal{C}^3$};
    \shade[ball color = prcolor!50!white, opacity = 0.9] (4,6) circle (1);
    \draw[dashed] (3,6) arc (180:0:1 and 0.4);
    \draw[dashed] (3,6) arc (180:0:1 and -0.4);
    \node at (4,6) {$\mathcal{C}^4$};
    \shade[ball color = prcolor!50!white, opacity = 0.9] (-3,1) circle (1.6);
    \draw[dashed] (-4.6,1) arc (180:0:1.6 and 0.5);
    \draw[dashed] (-4.6,1) arc (180:0:1.6 and -0.5);
    \node at (-3,1) {$\mathcal{C}^0$};
    \shade[ball color = prcolor!50!white, opacity = 0.9] (3,1) circle (1.6);
    \draw[dashed] (1.4,1) arc (180:0:1.6 and 0.5);
    \draw[dashed] (1.4,1) arc (180:0:1.6 and -0.5);
    \node at (3,1) {$\mathcal{C}^0_p$};
    \shade[ball color = prcolor!50!white, opacity = 0.9] (0,-4) circle (1.6);
    \draw[dashed] (-1.6,-4) arc (180:0:1.6 and 0.5);
    \draw[dashed] (-1.6,-4) arc (180:0:1.6 and -0.5);
    \node at (0,-4) {$\mathcal{C}^1$};
    \draw[line width = 1.5pt,-Triangle, densely dashed] (-3,4.2) -- (-3,2.8);
    \draw[line width = 1.5pt,-Triangle, densely dashed] (3,5.2) -- (3,2.8);
    \draw[line width = 1.5pt,-Triangle, densely dashed] (-2,-0.5) -- (-1,-2.5);
    \draw[line width = 1.5pt,-Triangle, densely dashed] (2,-0.5) -- (0.7,-2.8);
    \node[circle,draw=black,fill=black,inner sep=0pt,minimum size=3pt,label={below right:$p$}] at (0.6,-3) {};
     \end{tikzpicture}
    \caption{We show the base $\tilde{B}_3$ of a blow-up of the $\PP^1$-fiber of a $\PP^1$-fibration over $\FF_1$. The zero section is in the {\color{prcolor}{darker}} curves.}
    \label{fig:fibrationstructureBlP1F1}
\end{figure}
The volumes for the generators of $\text{Eff}^1(\tilde{B}_3)$ can be calculated also straightforwardly using the intersection ring, which reads
\begin{equation}
\begin{split}
\mathcal{I}(\tilde{B}_3) = &J_3^3+2 J_3^2\cdot J_4 +J_0 J_3^2+J_1\cdot J_3^2+2  J_3\cdot J_4^2+J_0^2\cdot J_3+2  J_0 \cdot J_3 \cdot J_4\,+\\ 
&+ J_1\cdot
   J_3 \cdot J_4+J_0 \cdot J_1 \cdot J_3+2 J_4^3+J_0^2\cdot J_4 +2  J_0\cdot J_4^2+ J_1\cdot J_4^2\,+\\
   &+ J_0 \cdot J_1 \cdot J_4\fstop
   \end{split}
\end{equation}
Using \eqref{ex3:effectiveKahler} we can calculate 
\begin{align}\label{volumeex3}
\begin{split}
\mathcal{V}_{D_0} &  = \frac{1}{2}(v^0)^2 +v^4  v^0 +v^0 v^1 \,, \\
\mathcal{V}_{D_1}&  = \frac{1}{2}(v^4)^2 + v^3 v^4 + v^3 v^1 + v^4 v^1\,,  \\
\mathcal{V}_{D_3}&= \frac{1}{2} (v^3)^2 + v^3 v^0  \,, \\
\mathcal{V}_{D_4} & = \frac{1}{2} (v^4)^2 + v^3 v^4 + v^4 v^0 \,,
\end{split}
\end{align}
in agreement with our previous discussion. 

In order to find the curves giving rise to (quasi-)primitive EFT string in this limit, let us first consider the movable cone $\text{Mov}_1(\tilde{B}_3) = \overline{\text{Eff}}^1(\tilde{B}_3)^\vee$ being generated by 
\begin{equation}
\begin{split}
    C^0= J_0\cdot J_1\coma C^1=J_1\cdot J_3\coma C^3= J_0\cdot J_3\coma C^4=J_3\cdot( J_4 -  J_1)\fstop
\end{split}
\end{equation}
Let us note that, as in Section \ref{ssec:P1fibFn}, $J_0^2=sJ_0\cdot J_1$ and $J_3^2=sJ_0\cdot J_3$, where here we use $s=1$. The volumes for such curves are
\begin{equation}
\begin{split}
\mathcal{V}_{C^0} &  = v^3 + v^4 \,, \\
\mathcal{V}_{C^1}&  = v^0+v^3 + v^4  \,,  \\
\mathcal{V}_{C^3}&= v^0  + v^1+ v^3 +2 v^4  \,, \\
\mathcal{V}_{C^4} & = v^0 + v^1+ v^3 + v^4 \,.
\end{split}
\end{equation}
Notice the splitting of $C^0 $ in terms of the Mori cone generators $\mathcal{C}^3$ and $\mathcal{C}^4$. Hence, we identify $C^0$ as the heterotic curve since $C^0 \cdot \bar{K}(\tilde{B}_3) =2$  and using the adjunction formula we obtain $g(C^0) =0$. Here we used that  $\bar{K}(\tilde{B}_3) = J_0 + J_3 +J_4$.  Also, note that $C^0=J_0^2 $ with $J_0^3=0$ holds. Compared to the previous example, we observe that here, one of the generators of $\text{Mov}_1(B_3)$, $C^4$ is not of the form $J_i \cdot J_j$ for two K\"ahler cone generators $J_i$ and $J_j$. Still, by Proposition \ref{prop:EFTstring} we expect that the curves giving rise to quasi-primitive EFT strings to be of this form. And indeed, we find the primitive EFT limits:
\begin{equation}
\renewcommand{\arraystretch}{1.25}
\begin{tabular}{c|c|l|c}
    \text{Movable curve} & \text{$q$ factor} & \multicolumn{1}{c|}{Primitive EFT limit} & ${\cal V}_{D_i} \to \infty$ \\
    \hline
         $C^0=J_0\cdot J_1$ & $q=0$ & $v^0,v^1\rightarrow \infty \coma v^3,v^4\rightarrow 0$ & $D_0$\\
          $C^1=J_1\cdot J_3$ & $q=1$ & $v^1\rightarrow \infty \coma v^0\rightarrow 0\coma v^3,v^4 \simeq \text{const.}$ & $D_1$ \\
    $C^3=J_0\cdot J_3$ & $q=2$ & $v^3 \rightarrow \infty \coma v^1,v^4 \rightarrow 0 \coma v^0 \simeq \text{const.}$ & $D_3$
\end{tabular}
\end{equation}
There is a quasi-primitive limit obtained for
\begin{equation}
    \renewcommand{\arraystretch}{1.25}
    \begin{tabular}{c|c|l|c}
    \text{Movable curve} & \text{$q$ factor} & \multicolumn{1}{c|}{Quasi-primitive EFT limit} & ${\cal V}_{D_i} \to \infty$ \\
    \hline
    $\tilde{C}^4= J_4^2$ & $q=2$ & $v^4 \rightarrow \infty \coma v^0 \rightarrow 0 \coma v^1,v^3 \simeq \text{const.}$ & $D_1$, $D_4$
\end{tabular}
\end{equation}
Let us also note that $\tilde{C}^4 = C^1+C^4 = J_4^2$ for this choice of twist. Hence, we again find that all curves leading to (quasi-)primitive EFT strings are as in Proposition \ref{prop:EFTstring}. 

\subsubsection{Primitive EFT string limit for $C^4$}

Let us now discuss the second chamber of the extended K\"ahler cone in which we can realize the primitive EFT string limit for the divisor $D_4$. The corresponding base $\tilde{B}'_3$ can be obtained by flopping the curve $\mathcal{C}^1$. 
 More precisely, the flop map reads
 \begin{equation}
    \text{Flop I: } \left\{ \renewcommand{\arraystretch}{1.1}\begin{array}{lcl}
\mathcal{C}_{\mathrm{I}}^0 &=& \mathcal{C}^0    \\
\mathcal{C}_{\mathrm{I}}^3 &=& \mathcal{C}^3  \\
\mathcal{C}_{\mathrm{I}}^4 &=& \mathcal{C}^1 + \mathcal{C}^4  \\
\mathcal{C}_{\mathrm{I}}^5 &=& -\mathcal{C}^1
\end{array}\right.
\fstop
 \end{equation}
Taking the K\"ahler cone generators $\{J_i\}_{i=0,3,4,5}$ such that $J_i \cdot \mathcal{C}_{\mathrm{I}}^j = \delta_i^j$, the  basis of effective divisors in $\tilde{B}_3'$ can be expressed  as
\begin{equation}
 D_0 = J_3- J_0\coma  D_1  = J_5 +J_0-J_4  \coma  D_3 = J_4  - J_3  \coma  D_4 = J_3 - J_5\,.  
\end{equation}
The anticanonical class is $\bar{K}(\tilde{B}_3') = J_0 +J_3 +J_4$ and the intersection ring of $\tilde{B}'_3$  reads
\begin{align}
\begin{split}
I(\tilde{B}'_3) = &J_3^3+J_5\cdot J_3^2+J_0\cdot  J_3^2+2 J_4\cdot  J_3^2+J_0^2\cdot  J_3+2 J_4^2\cdot  J_3+J_5\cdot  J_0\cdot  J_3 \,+\\
& +J_5\cdot  J_4\cdot  J_3+2J_0\cdot  J_4\cdot  J_3+2 J_4^3+J_5\cdot  J_0^2+J_5\cdot  J_4^2+2 J_0\cdot  J_4^2\,+\\
&+J_0^2\cdot  J_4+J_5\cdot  J_0\cdot  J_4\,.
\end{split}
\end{align}
The volumes of the effective divisors follows as
\begin{align}
\begin{split}
\mathcal{V}_{D_0} &  = \frac{1}{2}(v^0)^2 +v^0 v^4  \,, \\
\mathcal{V}_{D_1 }&  = \frac{1}{2}(v^4)^2 + v^3 v^4\,,  \\
\mathcal{V}_{D_3 }&= \frac{1}{2} (v^3)^2 + v^0 v^3  \,, \\
\mathcal{V}_{D_4} & = \frac{1}{2}(v^4)^2 +  v^5 v^3 +v^5 v^0+v^5 v^4  +v^3 v^4 +v^0 v^4\,.
\end{split}
\end{align}
We can compute that the movable cone $\text{Mov}_1(\tilde{B}'_3) = \overline{\text{Eff}}^1(\tilde{B}'_3)^\vee$ is then spanned by the
curves $C^i$ given by
\begin{align}
 C^0 = J_0^2 \coma  C^1 = J_3\cdot(J_4-J_5) \coma  C^3 = J_3^2 \coma  C^4 = J_5\cdot J_3\,.
\end{align}
The volume for these curves read
\begin{align}
\begin{split}
\mathcal{V}_{C^0} &  = v^3 +v^4 +v^5\,, \\
\mathcal{V}_{C^1 }&  =v^0+v^3 +v^4+v^5 \,,  \\
\mathcal{V}_{C^3 }&= v^0 + v^3 + 2v^4  + v^5 \,, \\
\mathcal{V}_{C^4} & = v^0 + v^3  + v^4\,.
\end{split}
\end{align}
From these expressions, we see that there only are two (quasi-)primitive EFT string limits that can be realized in this chamber of the K\"ahler cone corresponding to 
\begin{equation}
\begin{tabular}{c|c|l|c}
    \text{Movable curve} & \text{$q$ factor} & \multicolumn{1}{c|}{Primitive EFT limit} & ${\cal V}_{D_i} \to \infty$ \\
    \hline
     $C^3=J_3^2$ & $q=2$ & $v^3\rightarrow \infty \coma v^4,v^5\rightarrow 0\coma v^0 \simeq \text{const.}$ & $D_3$ \\
     $C^4=J_3\cdot J_5$ & $q=1$ & $v^5 \rightarrow \infty\coma v^0,v^3,v^4\simeq \text{const.}$ & $D_4$\\
\end{tabular}
\end{equation}
Notice that the first primitive EFT string limit can also be obtained in the K\"ahler cone chamber discussed previously whereas the second limit corresponds to the one we could not realize in the other chamber of the K\"ahler cone. Further, note that the curve $C^4$ gives rise to a $q=1$ EFT string. Therefore, after flopping the curve $\mathcal{C}^1$, the resulting $\tilde{B}'_3$ allows for a surface fibration over $\mathbb{P}^1$ with the generic fiber corresponding to $J_5$.

\bibliographystyle{JHEP}
\bibliography{mybib}

\end{document}